\documentclass[prx,aps,twocolumn,notitlepage,superscriptaddress,showpacs,nofootinbib]{revtex4-2}

\usepackage{enumerate,appendix}
\usepackage{amsmath, amsthm, amssymb,commath}
\usepackage{color,calc,graphicx}
\usepackage[usenames,dvipsnames,svgnames,table,cmyk,hyperref]{xcolor}
\usepackage{mathtools}
\usepackage{physics}

\usepackage[colorlinks]{hyperref}
\usepackage{optidef}
\hypersetup{
	colorlinks = true,
	urlcolor = {blue},
	citecolor = {blue},
	linkcolor= {blue}
}

\usepackage{graphicx}
\usepackage{amsmath}
\usepackage{latexsym}
\usepackage{bbm}

\usepackage[charter,cal=cmcal,sfscaled=false]{mathdesign}
\usepackage{booktabs}
\usepackage{multirow}
\usepackage{subcaption}
\usepackage{dcolumn}
\usepackage{mathrsfs}

\def \be {\begin{equation}}
\def \ee {\end{equation}}

\if0

\newcommand{\Tr}{\mathrm{Tr}}

\fi

\def \argmin{\mathop{\rm argmin}}

\def \cH{{\cal H}}

\def \sofc2{{\cal S}({\mathbb C}^2)}

\def\>{\rangle}
\def\<{\langle}

\newcommand{\frS}{\mathfrak{S}}
\def\blambda{\bm{\lambda}}

\def\SU{\mathop{\rm SU}}
\def \mix{\mathop{\rm mix}\nolimits}

\def \FF{{\mathbb F}}

\newtheorem{definition}{Definition}
\newtheorem{theorem}{Theorem}
\newtheorem{lemma}[theorem]{Lemma}
\newtheorem{proposition}[theorem]{Proposition}
\newtheorem{corollary}[theorem]{Corollary}
\newtheorem{remark}[theorem]{Remark}

\def\Label#1{\label{#1}\ [\ \text{#1}\ ]\ }
\def\Label{\label}

\newcommand{\mh}[1]{#1}

\newenvironment{proofof}[1]{\vspace*{5mm} \par \noindent{\it Proof of #1:\hspace{2mm}}}{\endproof
\hfill$\Box$ \vspace*{3mm}}

\begin{document}	
\title{General detectability measure}
\author{Masahito Hayashi}\email{hmasahito@cuhk.edu.cn}
\affiliation{School of Data Science, The Chinese University of Hong Kong, Shenzhen, Longgang District, Shenzhen, 518172, China}
\affiliation{International Quantum Academy, Futian District, Shenzhen 518048, China}
\affiliation{Graduate School of Mathematics, Nagoya University, Nagoya, 464-8602, Japan}

\begin{abstract}
Distinguishing resource states from resource-free states is a fundamental task in quantum information. 
We have approached the state detection problem through a hypothesis testing framework, with the alternative hypothesis set comprising resource-free states in a general context. 
Consequently, we derived the optimal exponential decay rate of the failure probability for detecting a given 
$n$-tensor product state when the resource-free states are separable states, positive partial transpose (PPT) states, or the convex hull of the set of stabilizer states. 
This optimal exponential decay rate is determined by the minimum of the reverse relative entropy, indicating that this minimum value serves as the general detectability measure.
The key technique of this paper is 
a quantum version of empirical distribution.
\end{abstract}


\maketitle

\section{Introduction}\Label{S1}
In quantum information theory, various resources play crucial roles, such as entangled states (non-separable states), non-positive partial transpose (non-PPT) states, and magic states (non-stabilizer states) \cite{PhysRevLett77Peres,HORODECKI19961,Veitch_2014}. These resource states are characterized by their exclusion from specific convex cones, for example, the cone of separable states, the cone of PPT states, and the cone of stabilizer states. Consequently, verifying whether a given state lies outside a particular convex cone is of fundamental importance. Recently, the paper \cite{H-24} proposed that such a detection problem can be framed as a hypothesis testing problem where the alternative hypothesis $H_1$ is composite, consisting of states in a given convex cone, and the null hypothesis $H_0$ corresponds to the specific state under detection. When $H_1$ is defined as the set of $n$-tensor product states, this hypothesis testing problem has been studied under the framework of the quantum Sanov theorem \cite{Bjelakovic,Notzel}
although the paper
\cite{H-02}
showed that the asymptotically optimal measurement can be chosen depending only on the null hypothesis $H_0$\footnote{Since the measurement used by \cite{H-02}
does not depend on the state in $H_1$,
it covers the case when the elements of $H_0$
is finite.}. The paper \cite{H-24} explored the relationship between the reverse relative entropy of entanglement and the exponential rate of decrease in the non-detection probability, proving the equality for maximally correlated states. Subsequently, the works \cite{LBR,FFF} extended this equality to general entangled states, addressing the hypothesis testing problem within general frameworks. 
The paper \cite{LBR} covers the case of separable states.
However, these papers did not investigate non-PPT states or magic (non-stabilizer) states. In particular, although \cite{FFF} developed a general theory for this hypothesis testing problem, it demonstrated that PPT states do not satisfy their general condition \cite[Discussion after Eq. (334)]{FFF}.
In addition, hypothesis testing has a close relation with various problems in quantum information, e.g.,
classical-quantum channel coding \cite{1207373,Wang_Renner2012},
entanglement distillation \cite{PhysRevA.57.1619,1624631}, \cite[Theorem 8.7]{H-text},
quantum data compression \cite{PhysRevA.66.032321},\cite[Lemma 10.3]{H-text},\cite[Remark 16]{4069150},
quantum resource theory \cite{BP2010,HY}, and
classical-quantum arbitrarily varying channel \cite{Bjelakovic2013,DWH},
etc.
In particular, 
the references \cite{1207373,1624631,H-text,PhysRevA.66.032321,4069150} linked the respective topics 
to quantum hypothesis testing via the method of quantum information spectrum.

In this paper, we introduce a comprehensive method to handle this hypothesis testing problem within a general framework. Assuming broad conditions, we derive the exponential decay rate of the non-detection probability, which equals the minimum of the reverse relative entropy over various non-resource states. This minimum value serves as a general detectability measure. 
Our conditions involve the existence of a compatible measurement, 
which was defined in \cite{PhysRevLett.103.160504}
and was used in \cite{BHLP}, 
alongside partial conditions from \cite{BP}. 
The compatible measurement must also be tomographically complete
 \cite{DARIANO200025}. While the paper \cite{LBR} imposes a condition related to compatible measurements, our condition is considerably simpler and easier to verify. Specifically, the sets of separable states, PPT states, and the convex hull of stabilizer states satisfy our conditions.

A pivotal component of our method is the quantum version of empirical distribution, introduced in \cite{another} to formulate an alternative quantum Sanov theorem. Hence, we begin by distinguishing between the two types of quantum Sanov theorems from \cite{Bjelakovic,Notzel} and \cite{another}. Empirical distribution plays a central role in information theory and statistics. The original classical Sanov theorem \cite{Sanov,DZ} characterizes its asymptotic behavior under the i.i.d. condition, providing the exponential decay rate for events where the empirical distribution belongs to a subset excluding the true distribution.

In the quantum context, analogous i.i.d. conditions are considered through $n$-fold tensor products. However, defining a quantum analog of empirical distribution is nontrivial. The papers \cite{Bjelakovic,Notzel} introduced a version of the quantum Sanov theorem without quantum empirical distributions. They framed a hypothesis testing problem where one hypothesis is a composite set of i.i.d. states, and the other is a single i.i.d. state
because the performance of this hypothesis testing
can be easily derived from the original Sanov theorem under
the classical setting.
The exponential decay rate for error probability was derived under a constant constraint on the composite hypothesis error probability. This result was termed the quantum Sanov theorem. The paper \cite{H-24} refined this analysis by applying an exponential constraint instead.

The quantum Sanov theorem by \cite{Bjelakovic,Notzel} 
has the following problem. 
When we extend 
proof-techniques-based
on empirical distributions 
in various other settings in information theory
to a quantum setting, this version
of quantum Sanov theorem does not work well. 

That is, to extend such a proof-technique to a quantum case, we need
to establish a quantum version of empirical distributions,
and their asymptotic behavior. 
For instance, \cite{watanabe2014strong} used empirical distributions to analyze channel resolvability's second-order asymptotics. 
Extending such results requires a quantum empirical distribution framework. 
The paper \cite{another} provided an alternative quantum Sanov theorem, establishing quantum empirical distributions' asymptotic behavior. 
Using this result, the paper \cite{HCG} proved strong converse results for classical-quantum channel resolvability. 
Therefore, we refer to the papers \cite{Bjelakovic,Notzel}'s result as the quantum Sanov theorem for hypothesis testing and the paper \cite{another}'s as the quantum Sanov theorem for quantum empirical distributions.

In this paper, we leverage the latter for general hypothesis testing with composite hypotheses. First, we address classical composite hypotheses under permutation symmetry using empirical distributions. This approach directly applies the classical Sanov theorem, where permutation symmetry reduces sufficient statistics to empirical distributions. Although our conditions initially involve empirical distribution properties, we simplify them by using compatibility conditions from \cite{BHLP}, making verification straightforward.
Next, we extend this framework to quantum hypothesis testing by using quantum empirical distributions. Our method reproduces the same exponential decay rate as that by \cite{LBR,FFF}, i.e., 
the general detectability measure, while requiring simpler conditions. 
Although verifying the separable state condition is nontrivial in 
the paper \cite{LBR}, 
our condition is immediately satisfied. This demonstrates the practicality of our approach and the relevance of quantum empirical distributions in hypothesis testing.

Finally, we point out the technical difference from the paper \cite{HY}.
The paper \cite{HY} consider the exponent of the error probability when the true state belong to the set of free states.
For simplicity, we will simplify the error probability when the true state is $\rho$, to the error probability of $\rho$.
%When the first argument is the iid state,
%When our test is fixed,
The exponent of the error probability (error exponent)
of 
the convex combination of two states $\rho_1$ and $\rho_2$
equals the smaller one of 
the exponents of the error probabilities of 
two states $\rho_1$ and $\rho_2$.
That is, to consider a smaller error exponent,
it is sufficient to consider 
the convex combination of two states $\rho_1$ and $\rho_2$.
The paper \cite{HY} essentially employs this property for the converse art.

In contrast, this paper considers the constant error probabilities of free states.
The error probability of the convex combination of two states $\rho_1$ and $\rho_2$ equals the convex combination of the error probabilities of $\rho_1$ and $\rho_2$.
When the magnitude relation between the error probabilities of $\rho_1$ and $\rho_2$ depends on the choice of our test, the magnitude relation among the error probabilities of $\rho_1$, $\rho_2$, and their convex combination also depends on the choice of our test.
Hence, this problem requires more complicated analysis.

The structure of this paper is as follows. Section \ref{S2} introduces the formulation and key notations. Section \ref{S2B} outlines existing results and our main findings, showing that separable states, PPT states, and stabilizer state convex hulls meet our assumptions. 
Sections \ref{S4} and \ref{S5} detail classical generalizations by using two methods, both serving as precursors to our final results. 
Section \ref{S6} reviews the quantum Sanov theorem for quantum empirical distributions and presents newly obtained and related formulas. 
Finally, Section \ref{S7} presents our main quantum generalization of the quantum Sanov theorem for hypothesis testing.

\section{Preparation}\Label{S2}
\subsection{Formulation}
\begin{figure}[tbhp]
\begin{center}
\includegraphics[scale=0.6]{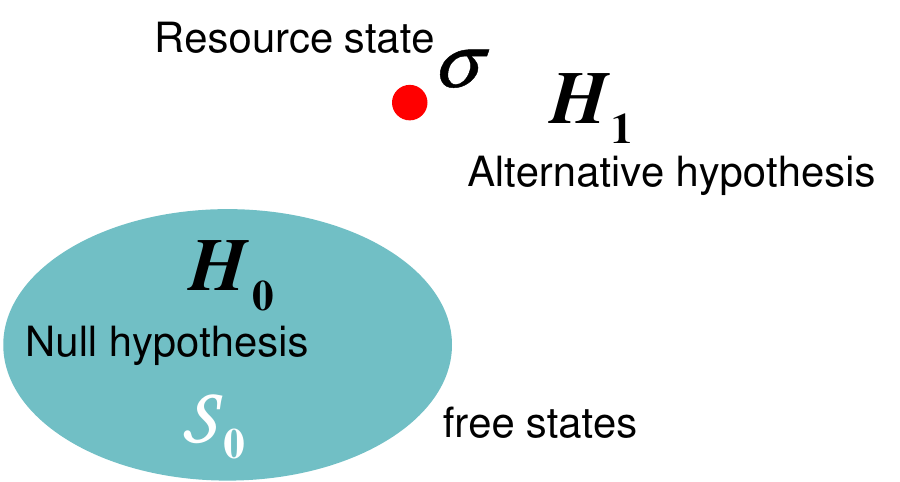}
\end{center}
\caption{Our testing problem}
\Label{resource}
\end{figure}
In this paper, we study the problem of quantum hypothesis testing 
to distinguish a given state $\sigma$ from a set ${\cal S}_0$
of resource-free states.
In the standard framework of quantum hypothesis testing, two quantum states, 
$\rho$ and $\sigma$, on a given Hilbert space ${\cal H}$ represent the two simple hypotheses. 
As illustrated in Fig. \ref{resource}, 
one is the alternative hypothesis $H_1$, which is composed of 
a state $\sigma$,
the other is the null hypothesis $H_0$, which is composed of 
a state $\rho$.
The aim of the hypothesis testing is 
to accept the alternative hypothesis $H_1$ by rejecting 
the null hypothesis $H_0$ with a fixed precision level.
To keep this fixed precision level for the rejection of $H_0$, 
we impose the probability for incorrectly rejecting $H_0$ to be upper bounded by a threshold $\epsilon>0$.
This probability $\epsilon$ is called the significance level of our test.

That is, we impose the following condition for our test operator $\Pi$
as
\begin{align}  
{\cal D}_{\epsilon,\rho} \coloneqq \Big\{ \Pi \Big| 0 \leq \Pi \leq I, \Tr[(I-\Pi)\rho] \le \epsilon \Big\},  \Label{SUE}
\end{align}  
where $I$ denotes the identity operator.  
Then, we optimize our test under this condition as follows.
\begin{align}  
\beta_\epsilon(\rho\| \sigma) \coloneqq \min_{\Pi\in{\cal D}_{\epsilon,\rho}} \Tr[\Pi \sigma].
\end{align}  

Since our aim is 
to distinguish a given state $\sigma$ from a compact set 
${\cal S}_0$ of resource-free states,
we impose the following condition instead of \eqref{SUE}.
\begin{align}  
{\cal D}_{\epsilon,{\cal S}_0} \coloneqq 
\Big\{ \Pi \Big| 0 \leq \Pi \leq I, \max_{\rho\in{\cal S}_0} \Tr[(I-\Pi)\rho] \le \epsilon \Big\}.  
\end{align}  
A hypothesis characterized by a set of states is known as a ``composite hypothesis''. 
Our optimization problem is formulated as
\begin{align}  
\beta_\epsilon\left({\cal S}_0 \middle\| \sigma\right) 
\coloneqq \min_{\Pi\in{\cal D}_{\epsilon,{\cal S}_0}}  \Tr[\Pi \sigma].
\end{align}  

In statistics, the probability $1-\beta_\epsilon\left({\cal S}_0 \middle\| \sigma\right)$
is called the detection probability of $\sigma$
under the null hypothesis ${\cal S}_0$
at a significance level $\epsilon$ \cite{Lehmann}.
When we make the detection by rejecting the null hypothesis
${\cal S}_0$ in statistics,
we need to fix a significance level $\epsilon$.
Hence, we can consider that the quantity 
$\beta_\epsilon\left({\cal S}_0 \middle\| \sigma\right)$
is linked to the detectability of the state $\sigma$
under the null hypothesis ${\cal S}_0$
rather than 
$\beta_\epsilon\left(\sigma \middle\|{\cal S}_0 \right)$
because 
$1-\beta_\epsilon\left(\sigma \middle\|{\cal S}_0 \right)$
shows the probability to correctly accept the null hypothesis ${\cal S}_0$.
Thus, when  
a sequence of subsets ${\cal F}_n \subset {\cal S}({\cal H}^{\otimes n})$
is given as our null hypothesis,
the exponent
$\frac{-1}{n}\log \beta_\epsilon\left({\cal F}_n \middle\| \sigma^{\otimes n}\right)$ shows 
the strength of the detection power of the state $\sigma$
under the null hypothesis ${\cal F}_n$
at a significance level $\epsilon$.
Therefore, when the limit 
$\lim_{n\to \infty}\frac{-1}{n}\log \beta_\epsilon\left({\cal F}_n \middle\| \sigma^{\otimes n}\right)$ exists, 
it can be considered as
the detectability measure.
The aim of this paper is to characterize 
the detectability measure under a general setting.

This quantity can be generalized by using another compact set of states $\mathcal{S}_1$ as
\begin{align}  
\beta_\epsilon\left({\cal S}_0 \middle\| \mathcal{S}_1\right) \coloneqq \min_{\Pi\in{\cal D}_{\epsilon,{\cal S}_0}} \max_{\sigma \in \mathcal{S}_1} \Tr[\Pi \sigma].
\end{align}  
For convenience, we introduce the following auxiliary quantities:  
\begin{align}  
\beta_\epsilon(\rho\|{\cal S}_2) \coloneqq 
\min_{0 \le \Pi \le I} \Big\{ \max_{\sigma \in {\cal S}_2} \Tr[\sigma \Pi] \Big| \Tr[(I-\Pi)\rho] \le \epsilon \Big\}.  
\end{align}  

We also employ the relative entropy and the sandwiched R\'enyi relative entropy \cite{MDS+13, WWY14}, defined as  
\begin{align}  
D(\rho\|\sigma) &\coloneqq \Tr \rho (\log \rho - \log \sigma), \\  
D_{\alpha}(\rho\|\sigma) &\coloneqq \frac{1}{\alpha - 1} \log \Tr \left( \sigma^{\frac{1-\alpha}{2\alpha}} \rho \sigma^{\frac{1-\alpha}{2\alpha}} \right)^{\alpha}.  
\end{align}  
Specifically, $D_1(\rho\|\sigma)$ is equivalent to the relative entropy $D(\rho\|\sigma)$.  

In preparation for the following discussion, we define the notations:  
\begin{align}  
D_{\alpha}({\cal S}_0\|\sigma) &\coloneqq \min_{\rho \in {\cal S}_0} D_{\alpha}(\rho\|\sigma), \\  
D_{\alpha}(\rho\|\mathcal{S}_1) &\coloneqq \min_{\sigma \in \mathcal{S}_1} D_{\alpha}(\rho\|\sigma), \\  
D_{\alpha}({\cal S}_0\|\mathcal{S}_1) &\coloneqq \min_{\rho \in {\cal S}_0, \sigma \in \mathcal{S}_1} D_{\alpha}(\rho\|\sigma).  
\end{align}  
In addition, in this paper, $\{X \ge Y\}$
expresses the projection 
$\sum_{j: z_j \ge 0}E_j$, where
$X-Y=\sum_{j} z_j E_j$.

\subsection{Basic Lemmas}  
Before proceeding with detailed analysis, we present basic lemmas. Using the max-min inequality~\cite[Section~5.4.1]{boyd_vandenberghe_2004}, we obtain  
\begin{align}  
\max_{\sigma \in \mathcal{S}_1} \beta_{\epsilon}\left({\cal S}_0 \middle\| \sigma\right) \leq \beta_{\epsilon}\left({\cal S}_0 \middle\| \mathcal{S}_1\right).  
\label{BH4Y}  
\end{align}  
When $\mathcal{S}_1$ is convex, the minimax theorem~\cite{v1928theorie, sion1958general, 10.2996/kmj/1138038812} yields the following results:  

\begin{lemma}\label{L1}  
For any $\epsilon \geq 0$, any compact set ${\cal S}_0$, any convex compact set $\mathcal{S}_1$, and any state $\rho$, the following hold:  
\begin{align}  
\max_{\sigma \in \mathcal{S}_1} \beta_{\epsilon}\left(\rho \middle\| \sigma\right) &= \beta_{\epsilon}\left(\rho \middle\| \mathcal{S}_1\right), \label{BH4-state} \\
\max_{\sigma \in \mathcal{S}_1} \beta_{\epsilon}\left({\cal S}_0 \middle\| \sigma\right) &= \beta_{\epsilon}\left({\cal S}_0 \middle\| \mathcal{S}_1\right). \label{BH4}  
\end{align}  
Additionally, for any $\epsilon \geq 0$, any convex compact sets ${\cal S}_0$ and $\mathcal{S}_1$, and any state $\sigma$, the following hold:  
\begin{align}  
\max_{\rho \in {\cal S}_0} \beta_{\epsilon}\left(\rho \middle\| \sigma\right) &= \beta_{\epsilon}\left({\cal S}_0 \middle\| \sigma\right), \label{BTD2} \\
\max_{\rho \in {\cal S}_0} \beta_{\epsilon}\left(\rho \middle\| \mathcal{S}_1\right) &= \beta_{\epsilon}\left({\cal S}_0 \middle\| \mathcal{S}_1\right). \label{BTD3}  
\end{align}  
\end{lemma}  

\begin{proof}  
The convexity and compactness of ${\cal D}_{\epsilon,\rho}$ allow us to derive \eqref{BH4-state} from the minimax theorem~\cite{v1928theorie, sion1958general, 10.2996/kmj/1138038812}, as shown in Lemma~3 of Ref.~\cite{HY}. Similarly, \eqref{BH4} follows from the convexity and compactness of ${\cal D}_{\epsilon,{\cal S}_0}$, as demonstrated in Ref.~\cite{HY}.  

We now establish \eqref{BTD2}. Let  
\begin{align}  
\alpha \coloneqq \max_{\rho \in {\cal S}_0} \beta_{\epsilon}\left(\rho \middle\| \sigma\right).  
\end{align}  
Since ${\cal D}_{\epsilon,{\cal S}_0} \subseteq {\cal D}_{\epsilon,\rho}$ for all $\rho \in {\cal S}_0$, we have  
\begin{align}  
\beta_\epsilon\left(\rho \middle\| \sigma\right) \leq \beta_\epsilon\left({\cal S}_0 \middle\| \sigma\right),  
\end{align}  
implying  
\begin{align}  
\alpha \leq \beta_\epsilon\left({\cal S}_0 \middle\| \sigma\right). \label{eq:beta_upper}  
\end{align}  
For any $\rho \in {\cal S}_0$, the inequality  
$\beta_{\epsilon}\left(\rho \middle\| \sigma\right) \leq \alpha$  
leads to  
\begin{align}  
\beta_{\alpha}\left(\sigma \middle\| \rho\right) \leq \epsilon. \label{BN1}  
\end{align}  
Since ${\cal S}_0$ is convex, \eqref{BH4-state} implies  
\begin{align}  
\beta_{\alpha}\left(\sigma \middle\| {\cal S}_0\right) = \max_{\rho \in {\cal S}_0} \beta_{\alpha}\left(\sigma \middle\| \rho\right) \leq \epsilon. \label{BN2}  
\end{align}  
Thus, there exists a POVM $\{\Pi^\ast, I - \Pi^\ast\}$ such that  
\begin{align}  
\max_{\rho \in {\cal S}_0} \Tr[\Pi^\ast \rho] \leq \epsilon, \quad 
\Tr[(I - \Pi^\ast) \sigma] \leq \alpha.  
\end{align}  
This implies  
\begin{align}  
\beta_{\epsilon}\left({\cal S}_0 \middle\| \sigma\right) \leq 
\Tr[(I - \Pi^\ast) \sigma] \leq \alpha. \label{eq:beta_lower}  
\end{align}  
Combining \eqref{eq:beta_upper} and \eqref{eq:beta_lower}, we obtain \eqref{BTD2}.  

Finally, \eqref{BTD3} follows from \eqref{BH4-state}, \eqref{BH4}, and \eqref{BTD2}, since  
\begin{align}  
&\max_{\rho \in {\cal S}_0} \beta_{\epsilon}\left(\rho \middle\| \mathcal{S}_1\right) = \max_{\rho \in {\cal S}_0} \max_{\sigma \in \mathcal{S}_1} \beta_{\epsilon}\left(\rho \middle\| \sigma\right) \notag \\
=& \max_{\sigma \in \mathcal{S}_1} \max_{\rho \in {\cal S}_0} \beta_{\epsilon}\left(\rho \middle\| \sigma\right) 
= \max_{\sigma \in \mathcal{S}_1} \beta_{\epsilon}\left({\cal S}_0 \middle\| \sigma\right) 
= \beta_{\epsilon}\left({\cal S}_0 \middle\| \mathcal{S}_1\right).  
\end{align}  
\end{proof}

When 
sets  ${\cal S}_0$ and $ \mathcal{S}_1$ are 
convex and closed for a unitary representation $U$ of a group, 
the joint convexity of relative entropy guarantees 
\begin{align}
D\left({\cal S}_0\middle\| \mathcal{S}_1\right)
=D\left(\mathcal{S}_{0,inv}\middle\| \mathcal{S}_{1,inv}\right),
\Label{D-inv}
\end{align}
where we denote the set of invariant elements for $U$ in ${\cal S}_j$ by ${\cal S}_{j,inv}$ for $j=0,1$.
The minimum error probability also satisfies the same property as follows.
\begin{lemma}\Label{LG7}
Given a unitary representation $U$ of a group $G$ on ${\cal H}$,
 we assume that
the sets  ${\cal S}_0$ and $ \mathcal{S}_1$ are 
convex and closed for $U$, 
i.e.,
$\mathcal{S}_i=\{ U(g)  \rho U(g)^\dagger\}_{\rho \in {\cal S}_0}
$ for any element $g \in G$ and $i=0,1$.
Then, we have
\begin{align}
\beta_{\epsilon}\left({\cal S}_0\middle\| \mathcal{S}_1\right)
=\beta_{\epsilon}\left(\mathcal{S}_{0,inv}\middle\| \mathcal{S}_{1,inv}\right).
\end{align}
Further, since any elements in 
$\mathcal{S}_{0,inv}$ and $ \mathcal{S}_{1,inv}$
are invariant for $U$,
our test can be limited to invariant tests.
\end{lemma}
\begin{proof}
Since ${\cal S}_0$ is a closed set for $U$,
the set ${\cal D}_{\epsilon,{\cal S}_0}$ is also closed for $U$.
Hence, we have
\begin{align}
\beta_{\epsilon}\left({\cal S}_0\middle\| \sigma\right)
=\beta_{\epsilon}\left({\cal S}_0\middle\| U(g)\sigma U(g)^\dagger\right)
\end{align}
for $g \in G$.
Thus, we have
\begin{align}
&\beta_{\epsilon}\left({\cal S}_0\middle\| 
\sum_{g\in G}\frac{1}{|G|}
U(g)\sigma U(g)^\dagger\right)\notag \\
\ge &
\sum_{g\in G}\frac{1}{|G|}
\beta_{\epsilon}\left({\cal S}_0\middle\| U(g)\sigma U(g)^\dagger\right) 
=
\beta_{\epsilon}\left({\cal S}_0\middle\| \sigma\right).
\end{align}
Since $\sum_{g\in G}\frac{1}{|G|}
U(g)\sigma U(g)^\dagger$ is an invariant element in $\mathcal{S}_1$,
we have
\begin{align}
\beta_{\epsilon}\left({\cal S}_0\middle\| \mathcal{S}_1
\right)
=
\max_{\sigma \in \mathcal{S}_1}
\beta_{\epsilon}\left({\cal S}_0\middle\| \sigma\right) %\notag\\
= 
\max_{\sigma \in \mathcal{S}_{1,inv}}
\beta_{\epsilon}\left({\cal S}_0\middle\| \sigma\right)
\Label{NMI-Y}.
%=\beta_{\epsilon}\left({\cal S}_0\middle\| \mathcal{S}_{1,inv}
%\right).
\end{align}

When $\sigma$ is an invariant state for $U$, we have
\begin{align}
\beta_{\epsilon}\left(U(g)\rho U(g)^\dagger\middle\| \sigma\right) 
= 
\beta_{\epsilon}\left(\rho\middle\| \sigma\right).
\end{align}
Hence, 
\begin{align}
&\beta_{\epsilon}\left(
\sum_{g\in G}\frac{1}{|G|}
U(g)\rho U(g)^\dagger
\middle\|\sigma
\right)\notag\\
\ge &
\sum_{g\in G}\frac{1}{|G|}
\beta_{\epsilon}\left(U(g)\rho U(g)^\dagger\middle\| \sigma\right) 
=
\beta_{\epsilon}\left(\rho\middle\| \sigma\right).
\end{align}
Since $\sum_{g\in G}\frac{1}{|G|}
U(g)\rho U(g)^\dagger$ is an invariant element in ${\cal S}_0$,
we have
\begin{align}
& \beta_{\epsilon}\left({\cal S}_0\middle\| \sigma
\right)
=
\max_{\rho \in {\cal S}_0}
\beta_{\epsilon}\left(\rho\middle\| \sigma\right) \notag\\
= &
\max_{\rho \in \mathcal{S}_{0,inv}}
\beta_{\epsilon}\left(\rho\middle\| \sigma\right) 
=\beta_{\epsilon}\left(\mathcal{S}_{0,inv}\middle\| \rho
\right).\Label{NMI2}
\end{align}
The combination of \eqref{NMI-Y} and \eqref{NMI2} yields 
\begin{align}
& \beta_{\epsilon}\left({\cal S}_0\middle\| \mathcal{S}_1
\right)
=
\max_{\sigma \in \mathcal{S}_{1,inv}}
\beta_{\epsilon}\left({\cal S}_0\middle\| \sigma\right)\notag\\
=&
\max_{\sigma \in \mathcal{S}_{1,inv}}
\beta_{\epsilon}\left(\mathcal{S}_{0,inv}\middle\| \sigma\right)
=
\beta_{\epsilon}\left(\mathcal{S}_{0,inv}\middle\| 
\mathcal{S}_{1,inv}
\right).
\end{align}
\end{proof}

\section{Our results and existing results}\Label{S2B}
\subsection{Existing results}
Before presenting our results, we review existing results.
We consider a $d$-dimensional system ${\cal H}$
spanned by the computational basis
$\{|j\rangle\}_{j=1}^d$
and the set of density matrices on ${\cal H}$, which is denoted by
${\cal S}({\cal H})$.
We denote the set of diagonal states 
with respect to the computation basis
by
${\cal S}_c({\cal H})$.
We focus on 
a sequence of subsets ${\cal F}_n \subset {\cal S}({\cal H}^{\otimes n})$.
In this subsection, we present several conditions for ${\cal F}_n$
that were introduced by preceding studies.
According to the paper \cite{BP}, the paper \cite{LBR} considers the following conditions
\begin{description}
\item[(A1)]
Each ${\cal F}_n$ is a convex and closed subset of ${\cal S}({\cal H}^{\otimes n})$, 
and hence also compact.
\item[(A2)]
${\cal F}_1$ contains some full-rank state $\sigma_0 (\ge e^{-r_0}>0)\in {\cal F}_1$.
\item[(A3)]
The family ${\cal F}_n$ is closed under partial traces, 
i.e. if $\sigma\in {\cal F}_{n}$ then 
$\Tr_{n} \sigma \in {\cal F}_{n-1}$
and $\Tr_{1,\ldots,n-1} \sigma \in {\cal F}_1$, 
where 
$\Tr_{n}$ denotes the partial trace over the $n$-th subsystem
and
$\Tr_{1,\ldots,n-1}$ denotes the partial trace over 
the $1$st, $2$nd,$\ldots$, $n-1$-th subsystems.
\item[(A4)] The family ${\cal F}_n$ is closed under tensor products, 
i.e. if $\sigma \in {\cal F}_n$ and 
$\sigma' \in {\cal F}_m$ then $\sigma \otimes \sigma' \in {\cal F}_{n+m}$.
\item[(A5)] Each ${\cal F}_n$ is closed under permutations, 
i.e. if $\sigma \in {\cal F}_n$ and $\pi\in {\cal S}_n$ 
denotes an arbitrary permutation of a set of $n$ elements, 
then also $U_\pi \sigma U_\pi^\dagger \in {\cal F}_n$, where $U_\pi$ 
is the unitary implementing $\pi$ over ${\cal H}^{\otimes n}$.
Under this condition, we denote the set of permutation-invariant states in ${\cal F}_n$ by ${\cal F}_{n,inv}$.
\end{description}
Recently, the paper \cite{LBR} introduced the following condition.
\begin{description}
\item[(A6)]
The regularised relative entropy of resource is faithful, i.e.
for any element $\rho \notin {\cal F}_1$, we have
\begin{align}
D^{\infty}(\rho \|{\cal F}):=
\lim_{n\to \infty}\frac{1}{n}D(\rho^{\otimes n}\| {\cal F}_n)> 0.
\end{align}
\end{description}

The paper \cite{LBR} showed the following statement
as a classical extension of quantum Sanov theorem for hypothesis testing.
\begin{proposition}\Label{Pr1}\cite[Theorem 8]{LBR}
When a sequence of subsets ${\cal F}_n \subset {\cal S}_c({\cal H}^{\otimes n})$
satisfies the conditions (A1)--(A6), 
any state $\sigma\in {\cal S}_c({\cal H})$ 
satisfies
\begin{align}
\lim_{n\to \infty}-\frac{1}{n}\log \beta_\epsilon({\cal F}_n\|\sigma^{\otimes n})
=
D({\cal F}_1\|\sigma)
=
\lim_{n\to \infty}\frac{1}{n}D({\cal F}_n\|\sigma^{\otimes n}) .
\Label{BNAT1}
\end{align}
\end{proposition}

\begin{definition}\cite[Definition 4]{BHLP}
A POVM $M=\{M_j\}_j$ on ${\cal H}$ is called compatible with
$({\cal F}_n)_n$ when 
$ \Tr_n \big((I^{\otimes n-1}\otimes M_j )\rho_n\big)/
\Tr \big((I^{\otimes n-1}\otimes M_j )\rho_n\big)
$ belongs to ${\cal F}_{n-1}$ for $\rho_n \in {\cal F}_n$.
\end{definition}

Also, the paper \cite{LBR} introduced the following condition.
\begin{description}
\item[(A7)]
There exists a sequence $r_n \in (0, 1]$ to satisfy the following condition.
We define the set ${\cal M}_n$ of POVMs as
\begin{align}
{\cal M}_n:=
\Big\{
\Big\{\frac{1}{2}(I^{\otimes n}+X_n), \frac{1}{2}(I^{\otimes n}-X_n)\Big\}
\Big|\| X_n\|\le r_n \Big\}. \Label{GAH}
\end{align}
Any POVM in ${\cal M}_n$ is compatible with
$({\cal F}_n)_n$. 
\end{description}

The paper \cite{LBR} showed the following statement
as a quantum extension of quantum Sanov theorem for hypothesis testing.
\begin{proposition}\Label{Pr1B}\cite[Theorem 14]{LBR}
When a sequence of subsets ${\cal F}_n \subset {\cal S}({\cal H}^{\otimes n})$
satisfies the conditions (A1)-(A5) and (A7), 
any state $\sigma$ 
satisfies 
\begin{align}
\lim_{n\to \infty}-\frac{1}{n}\log \beta_\epsilon({\cal F}_n\|\sigma^{\otimes n})
=
D({\cal F}_1\|\sigma)
=
\lim_{n\to \infty}\frac{1}{n}D({\cal F}_n\|\sigma^{\otimes n}) .
\Label{BNAT2}
\end{align}
\end{proposition}

Later, the paper \cite{FFF} considered another condition.
For a set ${\cal F}$, they define the subset
\begin{align}
({\cal F})_+^\circ :=
\{X \ge 0 | \Tr XY \le 1, \forall Y\in {\cal F}\}.
\end{align}
Then, the paper \cite{FFF} introduced the condition;
\begin{description}
\item[(A8)]
The relation
\begin{align}
({\cal F}_n)_+^\circ \otimes ({\cal F}_m)_+^\circ \subset 
({\cal F}_{n+m})_+^\circ
\end{align}
holds for any two positive integers $n,m$.
\end{description}

The paper \cite{FFF} showed the following statement
as another quantum extension of 
quantum Sanov theorem for hypothesis testing.
\begin{proposition}\Label{Pr2}\cite[Theorem 26]{FFF}
When a sequence of subsets ${\cal F}_n \subset {\cal S}({\cal H}^{\otimes n})$
satisfies the conditions 
(A1), (A4), (A5), and (A8),
any state $\sigma$ satisfies 
\begin{align}
\lim_{n\to \infty}-\frac{1}{n}\log \beta_\epsilon({\cal F}_n\|\sigma^{\otimes n})
=
\lim_{n\to \infty}\frac{1}{n}D({\cal F}_n\|\sigma^{\otimes n}) .
\Label{BNAT3}
\end{align}
\end{proposition}

\begin{proposition}\Label{Pr3}\cite[Theorem 1]{HY}
Assume that a sequence of subsets ${\cal F}_n 
\subset {\cal S}({\cal H}^{\otimes n})$
satisfies the conditions 
(A1), (A2), and (A4).
For any state $\sigma$,
the limits 
$\lim_{n\to \infty}-\frac{1}{n}\log \beta_\epsilon(
\sigma^{\otimes n}\|{\cal F}_n)$
and $\lim_{n\to \infty}\frac{1}{n}D(
\sigma^{\otimes n}\|{\cal F}_n)$ exist, and the relation
\begin{align}
\lim_{n\to \infty}-\frac{1}{n}\log \beta_\epsilon(
\sigma^{\otimes n}\|{\cal F}_n)
=
\lim_{n\to \infty}\frac{1}{n}D(
\sigma^{\otimes n}\|{\cal F}_n)
\Label{BHY}
\end{align}
holds.
\end{proposition}

Proposition \ref{Pr1B} states the equality between 
$D({\cal F}_1\|\sigma)$ and 
$\lim_{n\to \infty}\frac{1}{n}D({\cal F}_n\|\sigma^{\otimes n})$.
However, this equality can be derived only 
satisfies the conditions (A1), (A3) and (A4).
The condition (A4) implies 
\begin{align}
\frac{1}{n}D({\cal F}_n\|\sigma^{\otimes n})
\le D({\cal F}_1\|\sigma).
\Label{BHY41}
\end{align}
In contrast, the existing study \cite[Eq. (14)]{LBR}
shows that the conditions (A1) and (A3) imply 
\begin{align}
\frac{1}{n}D({\cal F}_n\|\sigma^{\otimes n})
\ge D({\cal F}_1\|\sigma).
\Label{BHY42}
\end{align}
The combination of \eqref{BHY41} and \eqref{BHY42}
yields the following proposition.

\begin{proposition}\Label{PPH}\cite[Section IV-B]{LBR}
Assume that a sequence of subsets ${\cal F}_n 
\subset {\cal S}({\cal H}^{\otimes n})$
satisfies the conditions 
(A1), (A3), and (A4).
Any state $\sigma$ satisfies
the additivity property without the asymptotic limit:
\begin{align}
\frac{1}{n}D({\cal F}_n\|\sigma^{\otimes n})
=
D({\cal F}_1\|\sigma).
\Label{BHY4}
\end{align}
\end{proposition}

\subsection{Our results}
To state our obtained result for the classical setting,
we introduce the following condition for a sequence of subsets ${\cal F}_n \subset {\cal S}_c({\cal H}^{\otimes n})$
by strengthening the condition (A3).
\begin{description}
\item[(B1)]
The measurement based on the computation basis 
is compatible with the sequence of subsets ${\cal F}_n \subset {\cal S}_c({\cal H}^{\otimes n})$.
%That is, all elements of ${\cal F}_n$ are diagonal states.
That is, this condition for the sequence of subsets ${\cal F}_n$
is composed of the following two conditions.
All elements of ${\cal F}_n$ are diagonal states.
For any element $\rho_n \in {\cal F}_n$,
the state $\Tr_{1,\ldots,n-1} \rho_n$ belongs to ${\cal F}_1$ and
the state $\Tr_n \rho_n (I^{\otimes n-1} \otimes |j\rangle \langle j|)/
\Tr \rho_n (I^{\otimes n-1} \otimes |j\rangle \langle j|)$
belongs to ${\cal F}_{n-1}$.
In other words, any conditional distribution of $\rho_n$ with condition on the $n$-th system belongs to ${\cal F}_{n-1}$.
\end{description}

\begin{theorem}\Label{THC3-C}
When a sequence of subsets ${\cal F}_n \subset {\cal S}_c({\cal H}^{\otimes n})$
satisfies the conditions (A1), (A5), and (B1),
any state $\sigma$ satisfies
satisfies \eqref{BNAT1}.
\end{theorem}

Theorem \ref{THC3-C} will be shown from Theorem \ref{THC3} 
in Section \ref{SSS3D}.
Theorem \ref{THC3} will be shown in Sections \ref{S4} and \ref{S5}.
The key point of Theorem \ref{THC3-C}
is to use the compatibility condition, which was also used in \cite[Theorem 16]{BHLP}.

To state our obtained result for the quantum setting,
we introduce the following definition.
\begin{definition}
A POVM $M=\{M_j\}_j$ on ${\cal H}$ is called 
tomographically complete
when the set $\{M_j\}$ forms a basis on the set of Hermitian matrices
on ${\cal H}$ \cite{DARIANO200025}.
Given a state $\rho$ on ${\cal H}$,
the classical state 
$\sum_j (\Tr M_j \rho) |j \rangle \langle j|$ is denoted by $M(\rho)$.
\end{definition}
Then, we introduce the following condition for a sequence of subsets ${\cal F}_n \subset {\cal S}({\cal H}^{\otimes n})$.
\begin{description}
\item[(B2)]
There exists a tomographically complete POVM $M=\{M_j\}_j$ 
on ${\cal H}$ with finite measurement outcomes such that 
$(M^{\otimes n}({\cal F}_n))_n$
satisfies the condition (B1).
\end{description}

The following condition implies the condition (B2) when the condition (A5) holds.
\begin{description}
\item[(B3)]
There exists 
a tomographically complete
 POVM $M=\{M_j\}_j$ 
on ${\cal H}$ that is compatible with $({\cal F}_n)_n$.
\end{description}

\begin{lemma}
When the condition (A5) holds,
the condition (B3) implies the condition (B2). 
\end{lemma}
\begin{proof}
Assume that
a tomographically complete POVM $M=\{M_j\}_j$ 
on ${\cal H}$ is compatible with $({\cal F}_n)_n$.
Then,
$ \Tr_n (I^{\otimes n-1}\otimes M_j )\rho_n/
\Tr (I^{\otimes n-1}\otimes M_j )\rho_n
$ belongs to ${\cal F}_{n-1}$ for $\rho_n \in {\cal F}_n$.
Hence,
$\Tr_n M^{\otimes n}(\rho_n) (I^{\otimes n-1} \otimes |j\rangle \langle j|)/
\Tr \rho_n (I^{\otimes n-1} \otimes |j\rangle \langle j|)
=M^{\otimes (n-1)}(
 \Tr_n (I^{\otimes n-1}\otimes M_j )\rho_n/
\Tr (I^{\otimes n-1}\otimes M_j )\rho_n)
$ belongs to $M^{\otimes (n-1)}({\cal F}_{n-1})$.
Hence, 
$\Tr_n M^{\otimes n}(\rho_n)$ also
belongs to $M^{\otimes (n-1)}({\cal F}_{n-1})$.
Repeating the same procedure, we find that
$\Tr_{2,\ldots, n} M^{\otimes n}(\rho_n)$ also
belongs to $M({\cal F}_{1})$.
Since the condition (A5) guarantees the permutation-invariance,
$\Tr_{1,\ldots, n-1} M^{\otimes n}(\rho_n)$ also
belongs to $M({\cal F}_{1})$.
Hence, $(M^{\otimes n}({\cal F}_n))_n$ satisfies the condition (B1), which implies the condition (B2).
\end{proof}

Although (B3) is a stronger condition than (B2),
(B3) can be more easily checked than (B2).
Hence, we often check (B3) instead of (B2). 
Our main result is given as follows.

\begin{theorem}\Label{THC4-C}
When a sequence of subsets ${\cal F}_n \subset {\cal S}({\cal H}^{\otimes n})$
satisfies the conditions (A1), (A5), and (B2),
any state $\sigma$ satisfies 
\begin{align}
\lim_{n\to \infty}-\frac{1}{n}\log \beta_\epsilon({\cal F}_n\|\sigma^{\otimes n})
=
D({\cal F}_1\|\sigma)
=
\lim_{n\to \infty}\frac{1}{n}D({\cal F}_n\|\sigma^{\otimes n}) .
\Label{BNAT}
\end{align}
\end{theorem}

Further, we can show that the condition (A7) implies (B3) as follows.
There exists a finite set of Hermitian matrices
${\cal W}_n$ such that
the set ${\cal W}_n$ linearly spans the set of Hermitian matrices
on ${\cal H}^{\otimes n}$
and any element $X_n\in {\cal W}_n$ satisfies $\|X_n\|\le r_n$.
The condition (A7) guarantees that the POVM
$\Big\{\frac{1}{2}(I^{\otimes n}+X_n), \frac{1}{2}(I^{\otimes n}-X_n)\Big\}$ is compatible with $({\cal F}_n)_n$
for $X_n\in {\cal W}_n$.
Then, the POVM $
\Big\{\frac{1}{2l}(I^{\otimes n}+X_n),
\frac{1}{2l}(I^{\otimes n}-X_n)\Big\}_{X_n \in {\cal W}_n}$
with $l=|{\cal W}_n|$ is 
compatible with $({\cal F}_n)_n$ and tomographically complete,
which implies the condition (B3).
Therefore, the condition of Theorem \ref{THC4-C} is a weaker condition than Proposition \ref{Pr1B}.

\subsection{Application to several examples}
We consider a bipartite system ${\cal H}_A\otimes {\cal H}_B$.
We define 
${\cal F}_{sep,n}$ and 
${\cal F}_{ppt,n}$
to be the sets of separable states and
positive partial transpose states
%, and Rains' states
on $({\cal H}_A\otimes {\cal H}_B)^{\otimes n}$, respectively.
That is, these sets are given as
\begin{align}
{\cal F}_{sep,n}:=&
\Big\{\sum_l p_l \rho_{A,l}\otimes \rho_{B,l}\Big| 
\rho_{A,l}\in {\cal S}({\cal H}_A^{\otimes n}),
\rho_{B,l}\in {\cal S}({\cal H}_B^{\otimes n})
\Big\}\\
{\cal F}_{ppt,n}:=&
\Big\{
\rho\in {\cal S}(({\cal H}_A\otimes{\cal H}_B)^{\otimes n})\Big|
\Gamma (\rho) \in {\cal S}(({\cal H}_A\otimes{\cal H}_B)^{\otimes n})
\Big\},
\end{align}
where $\Gamma$ is the operator for partial transpose on $A$.

It is easy to find that the sequences of the sets
${\cal F}_{sep,n}$, ${\cal F}_{ppt,n}$, 
%${\cal F}_{rains,n}$, 
and ${\cal F}_{st,n}$
satisfy the conditions (A1), (A2), (A3), (A4), and (A5).

We choose tomographically complete POVMs 
$M_A=\{M_{A,j}\}_j$ on ${\cal H}_A$ and $M_B=\{M_{B,k}\}_k$ on ${\cal H}_B$. 
Then, the POVM $M_A\otimes M_B$ is also a 
tomographically complete
 POVM on ${\cal H}_A\otimes {\cal H}_B$.
In addition, 
the sequence of the sets
$(M_A\otimes M_B)^{\otimes n}({\cal F}_{sep,n})$ 
satisfies the conditions (B1).
Since a positive partial transpose state $\rho_{n}$ on  
$({\cal H}_A\otimes {\cal H}_B)^{\otimes n}$
satisfies
\begin{align}
&\Gamma (\Tr_{n} (\rho_{n} 
((I_A \otimes I_B)^{\otimes (n-1)} \otimes
M_{A,j}\otimes M_{B,k})) \notag\\
=&
\Tr_{n} ( \Gamma (\rho_{n} )
\Gamma ((I_A \otimes I_B)^{\otimes (n-1)} \otimes
M_{A,j}\otimes M_{B,k}))
 \notag\\
=&
\Tr_{n} (\Gamma (\rho_{n} )
((I_A \otimes I_B)^{\otimes (n-1)} \otimes
M_{A,j}^T\otimes M_{B,k}))
\ge 0,
\end{align}
the sequence of the sets
$(M_A\otimes M_B)^{\otimes n}({\cal F}_{ppt,n})$ 
also satisfies the conditions (B1).
Hence, the sequences of the sets
${\cal F}_{sep,n}$ and 
${\cal F}_{ppt,n}$
satisfy the condition (B3).

Also, we define 
${\cal F}_{st,n}$ to be the convex hull of the set of stabilizer states on 
${\cal H}_p^{\otimes n}$, where ${\cal H}_p$ is a $p$-dimensional space.
The definition of a stabilizer state is related to 
the $2n$-dimensional space $\FF_p^{2n}$
over the finite field $\FF_p:= \mathbb{Z}/p \mathbb{Z}$.
The symplectic inner product $g(\cdot,\cdot)$ is defined as
\begin{align}
g(b_1,b_2):= \sum_{i=1}^n (s_{i,1}t_{i,2}-s_{i,2}t_{i,1})
\end{align}
for $b_l=(s_{1,l}, t_{1,l}, \ldots,s_{n,l}, t_{n,l}) \in \FF_p^{2n}$.
We define 
$X:= \sum_{j=0}^{p-1} |j+1\rangle \langle j| $
and $Z:=\sum_{j=0}^{p-1} e^{2\pi j/d}|j\rangle \langle j|$.
Using these operators, we define
discrete Weyl representation $W_n$, which is a 
projective unitary representation of 
$\FF_p^{2n}$ on ${\cal H}_p^{\otimes n}$ as follows \cite[Chapter 8]{H-q-text}.
For an element $b=(s_1, t_1, \ldots,s_n, t_n) \in \FF_p^{2n}$, 
we define
\begin{align}
W_n(b):= X^{s_1}Z^{t_1}\otimes \cdots \otimes X^{s_n}Z^{t_n}.
\end{align}
Then, a pure state $|\psi\rangle$ on ${\cal H}_p^{\otimes n}$
is called a stabilizer state 
when there exists an $n$-dimensional subgroup 
$N \subset \FF_p^{2n}$ satisfying the following conditions.
(i) $N$ is self-orthogonal, i.e., any two elements
$b_1,b_2\in N$ satisfies $g(b_1,b_2)=0$.
(ii) $|\psi\rangle$ is an eigenvector of 
$W_n(b)$ for any element $b \in N$.

To construct a useful POVM,
we denote the projection valued measurement given by spectral decomposition of $XZ^{k}$ by $E_k=\{E_{k,j}\}_{j=0}^{p-1}$ for $k=0,\ldots, p-1$.
We denote the projection valued measurement given by spectral decomposition of $Z$ by $E_p=\{E_{p,j}\}_{j=0}^{p-1}$.
We define the POVM $M:=\{ \frac{1}{p+1} E_{k,j}\}_{(k,j)\in [0,p]\times [0,p-1]}$.
Then, the POVM $M$ is also a 
tomographically complete POVM on ${\cal H}_p$,
and we have the following lemma.
\begin{lemma}\Label{LL10}
The sequence of the sets
$((M)^{\otimes n}({\cal F}_{st,n}))_n$
satisfies the conditions (B1).
\end{lemma}

Hence, the sequence of the sets
${\cal F}_{st,n}$ satisfies the condition (B3).
Therefore, due to Theorem \ref{THC4-C} and Proposition \ref{PPH}, 
${\cal F}_{sep,n}$, ${\cal F}_{ppt,n}$, and 
${\cal F}_{st,n}$
satisfy \eqref{BNAT} and \eqref{BHY4}.
In fact, the references \cite[Proposition 2]{Eisert_2003}
and \cite[Theorem 2]{Rubboli2024mixedstate}
%and \cite[Proposition 1]{Rubboli2024mixedstate} 
showed that 
${\cal F}_{sep,n}$ and
${\cal F}_{st,n}$
%${\cal F}_{rains,n}$
satisfy \eqref{BHY4}, respectively.
But, the relation \eqref{BNAT} was shown only for 
${\cal F}_{sep,n}$ by the papers \cite{LBR}.

\begin{proofof}{Lemma \ref{LL10}}
In order to show this lemma,
we fix an element $X^sZ^t \in \{XZ^{k}\}_{k=0}^{p-1}\cup \{Z\} $.
A stabilizer state $|x\rangle \in {\cal H}_p^{\otimes n}$
is characterized as a simultaneous eigenvector 
of $\{W_n(b_j)\}_{j=1}^n$, where
$\{b_j=(s_{j,1}, t_{j,1}, \ldots,s_{j,n}, t_{j,n})\}_{j=1}^n\subset \FF_p^{2n}$
are linearly independent and 
orthogonal to each other in the sense of the symplective inner product.
We denote the resultant state $|x(\omega)\rangle$ with the initial state $|x\rangle$
when we measure the $n$-th system by $X^s Z^t$
and obtains the outcome $\omega$.
It is sufficient to show that the state $|x(\omega)\rangle$ is also a stabilizer state.

We can show that the set $\{(s_{j,n}, t_{j,n})\}_{j=1}^n$
spans the vector space $\FF_p^2$ as follows.
To show this statement, we define 
$b_j^{n-1}:= (s_{j,1}, t_{j,1}, \ldots,s_{j,n-1}, t_{j,n-1})\in \FF_p^{2(n-1)}$.
If the set $\{(s_{j,n}, t_{j,n})\}_{j=1}^n$ does not span the vector space $\FF_p^2$,
$\{b_j^{n-1}\}_{j=1}^n$ are linear independent and commutative with each other in $\FF_p^{2(n-1)}$, 
However, the number of such vectors is limited to up to $n-1$, which yields a contradiction.

Thus, there exists a matrix $(a_{j,l})\in {\rm GL}(n, \FF_p)$ such that
$(s_{1,n}', t_{1,n}')=\cdots =(s_{n-1,n}', t_{n-1,n}')=(s,t)$, where
$(s_{j,n}', t_{j,n}'):= \sum_{l=1}^n a_{j,l} (s_{j,n}, t_{j,n})$ for $j=1, \ldots, n$.
We define the vectors 
$b'_j=(s_{j,1}', t_{j,1}', \ldots,s_{j,n}', t_{j,n}'):= \sum_{l=1}^n a_{j,l} b_j\in \FF_p^{2 n}$
and ${b'_j}^{n-1}:=(s_{j,1}', t_{j,1}', \ldots,s_{j,n-1}', t_{j,n-1}')\in \FF_p^{2(n-1)}$.
When $\omega_j$ is 
the eigenvalue of $W_n(b'_j)$ under the state
$|x\rangle \in {\cal H}_p^{\otimes n}$,
$\omega_j \omega^{-1}$
is the eigenvalue of $W_{n-1}({b'_j}^{n-1})$ under the state
$|x(\omega)\rangle\in {\cal H}_p^{\otimes (n-1)}$.
Therefore, 
$|x(\omega)\rangle\in {\cal H}_p^{\otimes (n-1)}$
is a stabilizer state with stabilizers
$\{ W_{n-1}({b'_j}^{n-1})\}_{j=1}^{n-1}$ in ${\cal H}_p^{\otimes n}$.
\end{proofof}

\subsection{Alternative condition sets}\Label{SSS3D}
The references \cite[Appendix IV]{HY}, \cite{HMH}
clarified that the convexity condition is needed for the second argument in the relative entropy in order to establish the relation
\eqref{BHY}.
However, the convexity condition is not so essential
for establishing the relation \eqref{BNAT} as follows.
The convexity condition is used for restricting on the set of permutation-invariant elements as in Lemma \ref{LG7}.
Also, the convexity is natural from the viewpoint of resource theory \cite{BP}.
Once we focus on the set of permutation-invariant elements, 
the convexity condition is not needed.
To clarify this point, we introduce alternative condition sets 
instead of (A1) and (A5)
as follows.
\begin{description}
\item[(C1)]
Each ${\cal F}_n$ is a compact subset of ${\cal S}({\cal H}^{\otimes n})$.
\item[(C2)] Any element of ${\cal F}_n$ is invariant under permutations.
\end{description}

Then, instead of Theorems \ref{THC3-C} and \ref{THC4-C}, we have the following theorems.

\begin{theorem}\Label{THC3}
When a sequence of subsets ${\cal F}_n \subset {\cal S}_c({\cal H}^{\otimes n})$
satisfies the conditions (C1), (C2), and (B1),
any state $\sigma$ satisfies
satisfies \eqref{BNAT1}.
\end{theorem}

\begin{theorem}\Label{THC4}
When a sequence of subsets ${\cal F}_n \subset {\cal S}({\cal H}^{\otimes n})$
satisfies the conditions (C1), (C2), and (B2),
any state $\sigma$ satisfies \eqref{BNAT}.
\end{theorem}

Lemma \ref{LG7} guarantees that Theorems \ref{THC3} and \ref{THC4} imply Theorems \ref{THC3-C} and \ref{THC4-C}, respectively. 
That is, the relation among conditions 
(A1), (A5), (C1), (C2), (B1), (B2) is summarized as Fig. \ref{D3}.
Therefore, the latter sections show
Theorems \ref{THC3} and \ref{THC4} instead of Theorems \ref{THC3-C} and \ref{THC4-C}. 

\begin{figure}[tbhp]
\begin{center}
\includegraphics[scale=0.35]{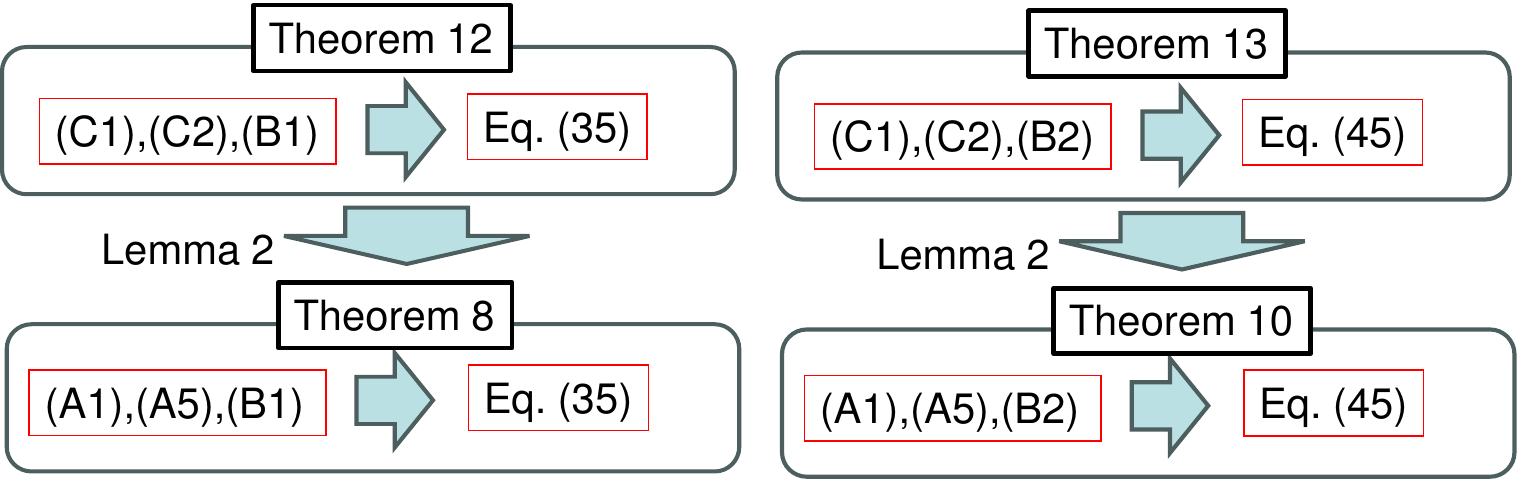}
\end{center}
\caption{Relation among conditions (A1), (A5), (C1), (C2), (B1), (B2)}
\Label{D3}
\end{figure}

In the original quantum Sanov theorem for hypothesis testing, 
the set ${\cal F}_n$ is given as the set $\{\rho^{\otimes n}\}_{\rho \in {\cal S}}$ for an arbitrary set ${\cal S}$ of states so that 
it does not satisfy the convexity condition, in general, in \cite{Bjelakovic,Notzel}. 
Hence, it cannot be derived from Propositions \ref{Pr1},\ref{Pr1B},\ref{Pr2} and Theorems \ref{THC3-C},\ref{THC4-C}. 
But, it can be recovered by Theorem \ref{THC4}.

\subsection{Non-convex examples}
Next, we consider non-convex examples as free states.
Given a compact subset ${\cal S}_0\subset {\cal S}({\cal H})$,
we focus on ${\cal S}^{\otimes n}$ as an example of ${\cal F}_n$.

Since ${\cal S}^{\otimes n}$ is not convex,
it does not satisfy the condition (A1).
But, ${\cal S}^{\otimes n}$ satisfies (C1),(C2).
For any tomographycally complete POVM $M$, 
$(M^{\otimes n}({\cal S}^{\otimes n})_n$
satisfies the condition (B1).
Hence, ${\cal S}^{\otimes n}$ satisfies (B2).
Thus, Theorem \ref{THC4} recovers
quantum Sanov theorem for hypothesis testing
while Propositions \ref{Pr1B} and \ref{Pr2}
do not.

For the next example, we consider a bipartite system ${\cal H}_A\otimes {\cal H}_B$.
We define the set ${\cal F}_{ind,n}$ 
to be the sets of the tensor product 
$\rho_A\otimes \rho_B$ of 
permutation invariant states 
$\rho_A$ and $\rho_B$
on ${\cal H}_A^{\otimes n}$ and ${\cal H}_B^{\otimes n}$, respectively.
This example was studied in \cite[Proposition 3.4]{Composite} while
the classical case of this example was studied in \cite{8231191}.
The set ${\cal F}_{ind,n}$ satisfies the conditions (C1) and (C2).

We choose tomographically complete POVMs 
$M_A=\{M_{A,j}\}_j$ on ${\cal H}_A$ and $M_B=\{M_{B,k}\}_k$ on ${\cal H}_B$. 
Then, the POVM $M_A\otimes M_B$ is also a tomographically complete
 POVM on ${\cal H}_A\otimes {\cal H}_B$.
In addition, the sequence of the sets
$(M_A\otimes M_B)^{\otimes n}({\cal F}_{ind,n})$ 
satisfies the conditions (B1) as follows.
Any state $\rho_A\otimes \rho_B \in {\cal F}_{ind,n}$ satisfies
\begin{align}
& 
\Tr_{n} (\rho_A\otimes \rho_B 
((I_A \otimes I_B)^{\otimes (n-1)} \otimes
M_{A,j}\otimes M_{B,k})) \notag\\
=&
\Tr_{n} (\rho_A (I_A^{\otimes (n-1)} \otimes M_{A,j}))
\otimes
\Tr_{n} (\rho_B (I_B^{\otimes (n-1)} \otimes M_{B,k})).
\Label{NMN}
\end{align}
Since $ \rho_A$ is permutation-invariant for the first $n-1$ 
systems, 
$\Tr_{n} (\rho_A (I_A^{\otimes (n-1)} \otimes M_{A,j}))$
is also permutation-invariant.
Similarly, 
$\Tr_{n} (\rho_B (I_B^{\otimes (n-1)} \otimes M_{B,k}))$
is also permutation-invariant.
Hence, the normalization of \eqref{NMN} belongs to 
${\cal F}_{ind,n-1}$.
Hence,
the set ${\cal F}_{ind,n}$ satisfies the condition (B2).
That is, it satisfies the conditions of Theorem \ref{THC4}.

%The set ${\cal F}_{ind,n}$ does not satisfy the condition (A7).
The set ${\cal F}_{ind,n}$ does not satisfy (A1), 
nor does it satisfy (A7).
Assume that ${\cal H}_A$ and ${\cal H}_B$ are spanned by
$|0\rangle,|1\rangle$.
We consider the case with $n=2$.
As a typical element, 
we choose an entangle state $|\Phi\rangle :=\frac{1}{\sqrt{2}}
(|00\rangle +|11\rangle)$.
We choose a positive real number $r_2>0$.
We choose 
a Hermitian matrix $X$ with $\|X\|=r_2$ on ${\cal H}_A\otimes {\cal H}_B$
such that $\Tr X=0$
and $\frac{1}{4}(I+X)$ is not a product state
between ${\cal H}_A$ and ${\cal H}_B$.
Then,
\begin{align}
\Tr_2 (\frac{1}{2}(I+X) \otimes I_{AB}) 
|\Phi\rangle\langle \Phi|_A \otimes |\Phi\rangle\langle \Phi|_B
=\frac{1}{8}(I+X).
\end{align}
Since $\frac{1}{8}(I+X)$ does not have a product form,
the set ${\cal F}_{ind,n}$ does not satisfy (A7).

To consider another example,
we consider a tripartite system 
${\cal H}_A\otimes {\cal H}_B\otimes {\cal H}_C$.
We assume that the third system ${\cal H}_C$ is a classical system
spanned by $\{|j \rangle\}_{j=1}^d$.
A state $\rho$ on 
${\cal H}_A\otimes {\cal H}_B\otimes {\cal H}_C$
is called $A-C-B$ Markovian
when 
$\rho$ has the form 
$\sum_{j=1}^d p_j \rho_{A,j}\otimes \rho_{B,j}\otimes |j\rangle\langle j|$.
We define the set ${\cal F}_{Mar,n}$ on 
the system $({\cal H}_A\otimes {\cal H}_B\otimes {\cal H}_C)^{\otimes n}$ as the set of 
permutation-invariant $A-C-B$ Markovian states.
The classical case of this example was studied in \cite{8231191}.
The set ${\cal F}_{Mar,n}$ satisfies the conditions (C1) and (C2).

We choose tomographically complete POVMs $M_A$ and $M_B$
as the same way as the above.
We set the POVM $M_C$ as the measurement based on the computation basis $\{|j \rangle\}_{j=1}^d$.
The sequence of the sets
$(M_A\otimes M_B\otimes M_C)^{\otimes n}({\cal F}_{Mar,n})$ 
satisfies the conditions (B1) as follows.
A state $\rho_{ABC}=
\sum_{j_1,\ldots, j_n} p_{j_1,\ldots, j_n} \rho_{A,j_1,\ldots, j_n}
\otimes \rho_{B,j_1,\ldots, j_n}\otimes 
|j_1,\ldots, j_n\rangle\langle j_1,\ldots, j_n|
\in {\cal F}_{Mar,n}$
satisfies
\begin{align} 
& \Tr_{n} (\rho_{ABC}
((I_A \otimes I_B\otimes I_C)^{\otimes (n-1)} \otimes M_{A,k_1}\otimes M_{B,_2} \otimes |j\rangle \langle j| ))\notag \\
=& \Tr_{n} \Big(\sum_{j_1,\ldots, j_n} p_{j_1,\ldots, j_n} \rho_{A,j_1,\ldots, j_n}
\otimes \rho_{B,j_1,\ldots, j_n}\otimes 
|j_1,\ldots, j_n\rangle\langle j_1,\ldots, j_n| \notag\\
&((I_A \otimes I_B\otimes I_C)^{\otimes (n-1)} \otimes M_{A,k_1}\otimes M_{B,_2} \otimes |j\rangle \langle j| )\Big)\notag \\
=& \sum_{j_1,\ldots, j_n} p_{j_1,\ldots, _{n-1},j} 
\Tr_{n}
(\rho_{A,j_1,\ldots, j_n} (I_A^{\otimes (n-1)} \otimes M_{A,k_1}))\notag\\
&\otimes 
\Tr_{n}
(\rho_{B,j_1,\ldots, j_{n-1}}
(I_A^{\otimes (n-1)} \otimes M_{B,k_2}))
\otimes 
|j_1,\ldots, j_n\rangle\langle j_1,\ldots, j_{n-1}| \notag\\
=&(\Tr_{n} (\rho_A\otimes \rho_B 
((I_A \otimes I_B)^{\otimes (n-1)} \otimes
M_{A,j}\otimes M_{B,k})) \notag\\
=&
(\Tr_{n} (\rho_A (I_A^{\otimes (n-1)} \otimes M_{A,j}))
\otimes
(\Tr_{n} (\rho_B (I_B^{\otimes (n-1)} \otimes M_{B,j})),
\Label{NMN}
\end{align}
which satisfies the $A-C-B$ Markovian condition.

Since $ \rho_{ABC}$ is permutation-invariant for the first $n-1$ 
systems, 
$ \Tr_{n} (\rho_{ABC}
((I_A \otimes I_B\otimes I_C)^{\otimes (n-1)} \otimes M_{A,k_1}\otimes M_{B,_2} \otimes |j\rangle \langle j| ))$
is also permutation-invariant.
Hence,
the set ${\cal F}_{Mar,n}$ satisfies the condition (B2).
That is, it satisfies the conditions of Theorem \ref{THC4}.
Since the set ${\cal F}_{Mar,n}$ is not convex,
it does not satisfy the condition (A1).

\section{First derivation for classical generalization: 
Proof of Theorem \ref{THC3}}\Label{S4}
The aim of this section is to prove
Theorem \ref{THC3} and to prepare the derivation of Theorem \ref{THC4}.
Although we have two derivations for Theorem \ref{THC3},
this section presents a more complicated derivation because 
this derivation works as a preparation of our proof of 
Theorem \ref{THC4}.
The technical lemmas used in both derivations
are needed for our proof of the quantum case.
In this section, we use the following condition for 
a sequence of subsets ${\cal F}_n \subset {\cal S}_c({\cal H}^{\otimes n})$ instead of the condition (B1) in the beginning.

\begin{description}
\item[(D1)]
For any element $p \notin {\cal F}_1$ and $\epsilon>0$, 
we choose any sequence $p_n \in {\cal T}_n$ such that
$ p_n \to p$.
Then, we have
\begin{align}
\liminf_{n\to \infty}
-\frac{1}{n}\log \beta_\epsilon
(\rho_n(p_n)\| {\cal F}_n)> 0.\Label{XAI}
\end{align}
\end{description}

The structure of this section is illustrated by Fig. \ref{D1}.
Although the aim of this section is the derivation of 
\eqref{BNAT1} from the conditions (C1), (C2), and (B1)
to show Theorem \ref{THC3},
we derive 
\eqref{BNAT1} from the conditions (C1), (C2), and (D1)
at the fisrt step.
Later, we derive the condition (D1) from 
the conditions (C1), (C2), and (B1).
In fact, as shown in Appendix \ref{S8},
the conditions (A6) and (D1) are equivalent
under the conditions (A1), (A2), (A4), and (A5).

\begin{figure}[tbhp]
\begin{center}
\includegraphics[scale=0.45]{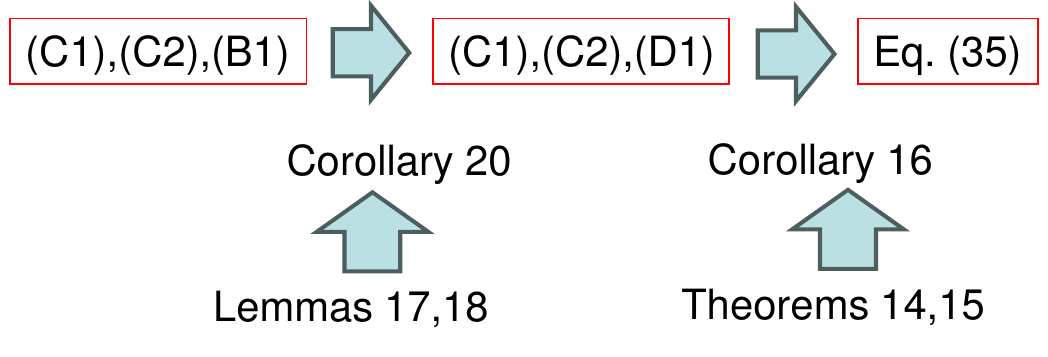}
\end{center}
\caption{Organization of Section \ref{S4}}
\Label{D1}
\end{figure}

%We consider the uniform distribution on ${\cal H}^{\otimes n}$ by $\rho_{n,\mix}$.
We denote the test function over a subset $\Omega$
by $1_\Omega$.
That is, $1_\Omega$ is defined as
\begin{align}
1_\Omega:= \sum_{j \in \Omega} |j\rangle \langle j|.
\Label{FG1}
\end{align}
We denote the set of types with length $n$ by ${\cal T}_n$, 
i.e., ${\cal T}_n$ is the set of possible empirical distribution with $n$ trials
\cite{720546}.
Given $p \in {\cal T}_n$, 
we denote the set of elements 
whose empirical distribution is $p$ by ${\cal T}_{n,p}$.
We denote the uniform distribution over ${\cal T}_{n,p}$
by $\rho_n(p)$.
We also define the operator $1^n_{p}:=1_{{\cal T}_{n,p}}$.
In this paper, we identify the distribution $p$ and the 
diagonal state $\sum_{j} p_j|j\rangle \langle j|$ 
with respect to the computational basis.
In this section, because all of the matrices are diagonal,
they are commutative.
Then, our notation is summarized in Table \ref{notation-c}.

\begin{table}[t]
\caption{Notations of classical system}
\label{notation-c}
\begin{center}
\begin{tabular}{|c|l|c|}
\hline
Symbol &Description &Eq. number    \\
\hline
$1_\Omega$ &Test function on $\Omega$& Eq. \eqref{FG1}     \\
\hline
${\cal T}_n$ & Set of types with length $n$ &   \\
\hline
\multirow{2}{*}{${\cal T}_{n,p}$} & Set of elements whose & \\
&empirical distribution is $p$ & \\
\hline
$\rho_n(p)$ & Uniform distribution over ${\cal T}_{n,p}$ & \\
\hline
$1^n_{p}$ & Test function on ${\cal T}_{n,p}$, i.e.,
$1_{{\cal T}_{n,p}}$ & \\
\hline
\end{tabular}
\end{center}
\end{table}

%Applying Theorem \ref{THY} to the case when ${\cal L}$ is ${\cal F}_1$,
In the following theorem, we employ a subset ${\cal L}$.
This theorem will be used to show 
\eqref{BNAT1} by substituting ${\cal F}_1$ into ${\cal L}$.

\begin{theorem}\Label{THY}
Assume that a sequence of subsets ${\cal F}_n \subset {\cal S}_c({\cal H}^{\otimes n})$
satisfies the conditions (C1) and (C2).
Also, given a subset ${\cal L} \subset {\cal S}_c({\cal H})$,
we assume that
any sequence $p_n \in {\cal T}_n$ with
$ p_n \to p \notin {\cal L}$ satisfies
\begin{align}
\liminf_{n\to \infty}
-\frac{1}{n}\log \beta_\epsilon
(\rho_n(p_n)\| {\cal F}_n)> 0.\Label{BNU}
\end{align}
Then, any state $\sigma$ satisfies
\begin{align}
\liminf_{n\to \infty}-\frac{1}{n}\log \beta_\epsilon({\cal F}_n\|\sigma^{\otimes n})
\ge 
D({\cal L}\|\sigma).
\Label{NHA}
\end{align}
\end{theorem}

For further discussion, we prepare the following theorem,
which will be shown in Appendix \ref{S9}.
\begin{theorem}\Label{BFG8}
Assume 
that all states in ${\cal F}_n$ are commutative with 
a state $\sigma_n$ on ${\cal H}^{\otimes n}$, and
that the relation
\begin{align}
\liminf_{n\to \infty}\frac{-1}{n}\log \beta_\epsilon({\cal F}_n\|\sigma_n)
\ge 
\limsup_{n\to \infty}\frac{1}{n}D({\cal F}_n\|\sigma_n)
\Label{NH4}
\end{align}
holds for any $\epsilon>0$.
Then, the above limits exist and
the relation
\begin{align}
\lim_{n\to \infty}\frac{-1}{n}\log \beta_\epsilon({\cal F}_n\|\sigma_n)
=\lim_{n\to \infty}\frac{1}{n}D({\cal F}_n\|\sigma_n)
\Label{NH5}
\end{align}
holds for any $\epsilon>0$.
\end{theorem}
This theorem does not assume that
elements of ${\cal F}_n$ are commutative with each other. That is, ${\cal F}_n$ is not restricted to 
the classical case.
\begin{corollary}\Label{THC}
When a sequence of subsets ${\cal F}_n \subset {\cal S}({\cal H}^{\otimes n})$
satisfies the conditions (C1), (C2) and (D1),
any state $\sigma$ satisfies
\eqref{BNAT1}.
\if0
\begin{align}
\lim_{n\to \infty}-\frac{1}{n}\log \beta_\epsilon({\cal F}_n\|\sigma^{\otimes n})
=
D({\cal F}_1\|\sigma)
=
\lim_{n\to \infty}\frac{1}{n}D({\cal F}_n\|\sigma^{\otimes n}) .
\Label{NHU}
\end{align}
\fi
\end{corollary}

\begin{proof}
Applying Theorem \ref{THY} to the case when ${\cal L}$ is ${\cal F}_1$,
we have
\begin{align}
\liminf_{n\to \infty}-\frac{1}{n}\log \beta_\epsilon({\cal F}_n\|\sigma^{\otimes n})
\ge &
D({\cal F}_1\|\sigma) \notag\\
\ge &
\limsup_{n\to \infty}\frac{1}{n}D({\cal F}_n\| \sigma^{\otimes n}) .
\Label{NH1A}
\end{align}
Then, Theorem \ref{BFG8} with $\sigma_n=\sigma^{\otimes n}$
implies \eqref{BNAT1}.
%\eqref{NHU}.
\end{proof}

\begin{proofof}{Theorem \ref{THY}}
{\bf Step 0:} Preparation.

Since $\rho_n(p_n')$ is permutation-invariant, 
due to Lemma \ref{LG7},
we can replace the set ${\cal F}_n$
by the set ${\cal F}_{n,inv}$
of permutation-invariant states of ${\cal F}_n$.
Since any elements of ${\cal F}_n$ and $\rho_n(p_n')$ are diagonal, 
our test can be restricted into a diagonal test.
Given any test $T$, we define the permutation-invariant test
$\Pi_{inv}:=\frac{1}{|{\cal S}_n|}\sum_{\pi\in {\cal S}_n} U_\pi \Pi U_\pi^\dagger$.
Then, we have
\begin{align*}
\Tr (\Pi_{inv} \rho_n(p_n'))&= \Tr (\Pi \rho_n(p_n')), \\
\Tr (\Pi_{inv} \rho) &= \Tr (\Pi \rho)
\end{align*}
for $\rho \in {\cal F}_{n,inv}$.
Hence, 
in the following discussion,
our test can be restricted to a diagonal permutation-invariant test, which the form $\sum_{p \in {\cal T}_n}c(p) 1_p^n$.
 
{\bf Step 1:}
%Choice of $m,\epsilon_0$, and 

We choose an arbitrary compact set ${\cal G}$ included in
${\cal L}^c$.
We choose the test 
$\Pi_n:= \sum_{ p \in {\cal G}\cap {\cal T}_n }1^n_{p}$.
Then, we show the following relation by contradiction;
\begin{align}
\max_{\sigma' \in {\cal F}_n} \Tr (\Pi_n \sigma' ) \to 0.
\Label{NI8}
\end{align}
Assume that \eqref{NI8} does not hold.
There exists a sequence of permutation-invariant states
$\sigma_n' \in {\cal F}_n$ such that
\begin{align}
c:=\limsup_{n\to \infty} \Tr ( \Pi_n \sigma_n') >0.
\Label{NI2}
\end{align}
We choose a subsequence $\{n_k\}$ such that
\begin{align}
c=\lim_{k\to \infty} c_{n_k} ,\quad
c_n:= \Tr (\Pi_{n} \sigma_{n}' ).
\Label{NI3}
\end{align}
Since
\begin{align*}
c_n= \sum_{ p \in {\cal G}\cap {\cal T}_n }
\Tr (1^n_{p} \sigma_{n}'),
\end{align*}
There exists an element $p_n' \in {\cal T}_n$
such that
\begin{align}
\Tr (1^n_{p_n'} \sigma_{n}) 
\ge \frac{c_n}{|{\cal G}\cap {\cal T}_n|}
\ge \frac{c_n}{|{\cal T}_n|}
\ge c_n (n+1)^{-(d-1)}.
\Label{NI4}
\end{align}
Since ${\cal G}$ is compact, 
there exists a subsequence $\{m_k\}$ of $\{n_k\}$ such that
$p_{m_k}'$ converges to $p' \in {\cal G}$.

Now, we consider the discrimination between $\rho_n(p_n')$
and ${\cal F}_n$.
To accept $\rho_n(p_n')$ with probability $1-\epsilon$,
the test $\Pi_n'$ needs to be $(1-\epsilon)1^n_{p_n'}$
because our test is restricted to the form $\sum_{p \in {\cal T}_n}c(p) 1_p^n$
as explained in Step 0, and 
$\Tr (\rho_n(p_n') 1^n_{p_n''})=\delta_{p_n',p_n''}$.

Then, we have
\begin{align}
&\lim_{k\to \infty}-\frac{1}{m_k}\log \Tr
( \Pi_{m_k}' \sigma_{m_k})\notag\\
=&
\lim_{k\to \infty}-\frac{1}{m_k}\Big(\log \Big( \Tr 1^{m_k}_{p_{m_k}'} \sigma_{m_k}\Big)
+\log (1-\epsilon)\Big) \notag\\
\le &
\lim_{k\to \infty}-\frac{1}{m_k}
\Big(\log c_{m_k} (m_k+1)^{-(d-1)}
+\log (1-\epsilon)\Big)
=0,
\end{align}
which contradicts with \eqref{BNU}.
Hence, we obtain \eqref{NI8}.

{\bf Step 2:}
We have
\begin{align}
\Tr (\sigma^{\otimes n} (I-\Pi_n))
=\sum_{p \in {\cal G}^c\cap {\cal T}_n} 
\Tr (1^n_{p} \sigma^{\otimes n}).
\end{align}
Thus, the original Sanov theorem (See Proposition \ref{Sanov-ori}) implies
\begin{align}
&\liminf_{n\to \infty}\frac{-1}{n}\log 
\Tr (\sigma^{\otimes n} (I-\Pi_n))
\notag\\
=&  \liminf_{n\to \infty}\frac{-1}{n}\log 
\sum_{p \in {\cal G}^c\cap {\cal T}_n}
\Tr (1^n_{p}\sigma^{\otimes n}) %\notag\\
= \inf_{p \in {\cal G}^c}D(p \|\sigma )
=D({\cal G}^c \|\sigma ).\Label{MNT5}
\end{align}
The combination of \eqref{NI8} and \eqref{MNT5}
implies 
\begin{align}
\liminf_{n\to \infty}\frac{-1}{n}\log \beta_\epsilon({\cal F}_n\|\sigma^{\otimes n})
\ge D({\cal G}^c \|\sigma ).
\end{align}
Since ${\cal G}$ is an arbitrary compact set 
contained in ${\cal L}^c$, we have
\begin{align}
\liminf_{n\to \infty}\frac{-1}{n}\log \beta_\epsilon({\cal F}_n\|\sigma^{\otimes n})
\ge D({\cal L}\|\sigma).
\end{align}
\end{proofof}

\begin{lemma}\Label{Pr6}
Assume that a sequence of subsets ${\cal F}_n 
\subset {\cal S}_c({\cal H}^{\otimes n})$
satisfies the conditions (C1) and (B1).
Then, we have
\begin{align}
D_\alpha({\cal F}_n\| \sigma^{\otimes n})
\ge
n D_\alpha({\cal F}_1\| \sigma).
\end{align}
\end{lemma}

\begin{proof}
It is sufficient to show
\begin{align}
D_\alpha({\cal F}_{n}\| \sigma^{\otimes n})
\ge
D_\alpha({\cal F}_{n-1}\| \sigma^{\otimes n-1})
+
D_\alpha({\cal F}_{1}\| \sigma).\label{NMN3}
\end{align}
We fix $\alpha<1$.
For $ \rho_n \in {\cal F}_{n}$, 
$\frac{\Tr_n (\rho_n (I^{\otimes n-1}\otimes |j\rangle \langle j|))}
{\Tr (\rho_n (I^{\otimes n-1}\otimes |j\rangle \langle j|))}$
belongs to ${\cal F}_{n-1}$.
Hence, we have
\begin{align}
&\Tr 
\Big(\frac{\Tr_n (\rho_n (I^{\otimes n-1}\otimes |j\rangle \langle j|))}
{\Tr (\rho_n (I^{\otimes n-1}\otimes |j\rangle \langle j|))}\Big)^{\alpha}
(\sigma^{\otimes n-1})^{1-\alpha} \notag\\
\le & 
e^{(\alpha-1)D_\alpha({\cal F}_{n-1}\| \sigma^{\otimes n-1})}.
\end{align}
Thus, 
\begin{align}
&e^{(\alpha-1)D_\alpha({\cal F}_{n}\| \sigma^{\otimes n})}
=
\max_{\rho_n \in {\cal F}_{n}}
\Tr ( \rho_n^{\alpha}(\sigma^{\otimes n})^{1-\alpha})\notag\\
=&
\max_{\rho_n \in {\cal F}_{n}}
\sum_{j} 
\Tr 
((
\Tr_n (\rho_n (I^{\otimes n-1}\otimes |j\rangle \langle j|))
)^{\alpha}
(\sigma^{\otimes n-1})^{1-\alpha})
\langle j|\sigma|j\rangle^{1-\alpha} \notag\\
=&
\max_{\rho_n \in {\cal F}_{n}}
\sum_{j} 
\Tr \Big(
\Big(\frac{\Tr_n (\rho_n (I^{\otimes n-1}\otimes |j\rangle \langle j|))}
{\Tr (\rho_n (I^{\otimes n-1}\otimes |j\rangle \langle j|))}\Big)^{\alpha}
(\sigma^{\otimes n-1})^{1-\alpha} \Big)\notag\\
&\cdot 
(\Tr (\rho_n (I^{\otimes n-1}\otimes |j\rangle \langle j|))
)^{\alpha}
\langle j|\sigma|j\rangle^{1-\alpha} \notag\\
\le&
e^{(\alpha-1)D_\alpha({\cal F}_{n-1}\| \sigma^{\otimes n-1})}
\max_{\rho_{n} \in {\cal F}_{n}}
\sum_{j} 
(\Tr ( (\Tr_{1,\ldots,n-1} \rho_n) |j\rangle \langle j|))^{\alpha}
\langle j|\sigma|j\rangle^{1-\alpha} \notag\\
=&
e^{(\alpha-1)D_\alpha({\cal F}_{n-1}\| \sigma^{\otimes n-1})}
e^{(\alpha-1)D_\alpha({\cal F}_{1}\| \sigma)},
\end{align}
which implies \eqref{NMN3}.
The case $\alpha>1$ can be shown in the same way.
\end{proof}

\begin{lemma}\Label{NK8}
When $\sigma \ge c \sigma'$,
a POVM $M=\{M_j\}_j$ and a state $\rho$ satisfy
\begin{align}
D_\alpha(M(\rho)\| M(\sigma'))
\ge D_\alpha(M(\rho)\| M(\sigma))+ \log c.\Label{ZI3}
\end{align}
\end{lemma}

\begin{proof}
Assume that $\alpha \in [0,1)$.
Since 
$\Tr ( M_j\sigma )\ge c \Tr (M_j \sigma')$,
we have
\begin{align}
&e^{(\alpha-1)D_\alpha(M(\rho)\| M(\sigma))}
=\sum_j (\Tr M_j \rho)^{\alpha} (\Tr M_j \sigma)^{1-\alpha} \notag\\
\ge &\sum_j (\Tr M_j \rho)^{\alpha} (c \Tr M_j \sigma')^{1-\alpha} \notag\\
= &c^{1-\alpha} \sum_j (\Tr M_j \rho)^{\alpha} (\Tr M_j \sigma')^{1-\alpha} \notag\\
=&c^{1-\alpha}e^{(\alpha-1)D_\alpha(M(\rho)\| M(\sigma'))}.
\end{align}
Thus, we obtain \eqref{ZI3} with $\alpha \in [0,1)$.

When $\alpha>1$, we have
\begin{align}
&e^{(\alpha-1)D_\alpha(M(\rho)\| M(\sigma))}
=\sum_j (\Tr M_j \rho)^{\alpha} (\Tr M_j \sigma)^{1-\alpha} \notag\\
\le &\sum_j (\Tr M_j \rho)^{\alpha} (c \Tr M_j \sigma')^{1-\alpha} \notag\\
= &c^{1-\alpha} \sum_j (\Tr M_j \rho)^{\alpha} (\Tr M_j \sigma')^{1-\alpha} \notag\\
=&c^{1-\alpha}e^{(\alpha-1)D_\alpha(M(\rho)\| M(\sigma'))}.
\end{align}
Thus, we obtain \eqref{ZI3} with $\alpha>1$.
\end{proof}

To check the condition (D1), the following lemma is useful. 
\begin{lemma}\Label{LB89T}
Assume that a sequence of subsets 
${\cal F}_n \subset {\cal S}_c({\cal H}^{\otimes n})$
satisfies the conditions (C1) and (C2).
Then, any diagonal state $\sigma\in {\cal S}_c({\cal H})$
satisfies
\begin{align}
-\frac{1}{n}\log \beta_\epsilon
(\rho_{n}(p_n)
\| {\cal F}_n) 
\ge 
\frac{(1-\alpha)}{n}
D_\alpha({\cal F}_{n}\| \rho_{n}(p_n))
\Label{FG3T}
\end{align}
for $\alpha \in [0,1)$.
\end{lemma}

In fact, when the limit of RHS of \eqref{FG3T} is strictly positive,
the condition (D1) holds.

\begin{proof}
Any element of ${\cal F}_{n}$ and 
$\rho_{n}(p_n)$ are written with a form
$\sum_{p \in {\cal T}_n} c_p 1_p^n$ with $c_p \in [0,1]$.
Thus, to accept $\rho_{n}(p_n)$ with probability $1-\epsilon$,
the operator $\Pi_n$ to support $\rho_{n}(p_n)$ needs to be $(1-\epsilon)1^n_{p_n}$.
Since $\Tr \rho_{n}(p_n)1^n_{p_n}=1$,
we have 
$\rho_{n}(p_n)1^n_{p_n} \ge 
\sigma_n 1^n_{p_n} $ for $\sigma_n\in {\cal F}_{n}$.
Thus, we have
\begin{align}
& \beta_\epsilon
(\rho_{n}(p_n)
\| {\cal F}_n) 
= 
\max_{\sigma_n \in {\cal F}_{n}}
\Tr (\sigma_n (1-\epsilon) 1^n_{p_n} )\notag\\
\le &
(1-\epsilon) \max_{\sigma_n \in {\cal F}_{n}}
\Tr (\sigma_n^{\alpha} 
\rho_{n}(p_n)^{1-\alpha} )\notag\\
=&(1-\epsilon) e^{-(1-\alpha)D_\alpha
({\cal F}_{n}\| \rho_{n}(p_n))} 
\le e^{-(1-\alpha)D_\alpha
({\cal F}_{n}\| \rho_{n}(p_n))} 
%=(1-\epsilon) e^{-(1-\alpha) D_\alpha({\cal F}_{n}\| \rho_{n}(p_n))}
,\Label{BNC}
\end{align}
%Since the minmum value is attained by an invariant state,
which implies \eqref{FG3T}.
\end{proof}

Then, Theorem \ref{THC3} follows from the following corollary.
\begin{corollary}\Label{THC6}
When a sequence of subsets 
${\cal F}_n \subset {\cal S}_c({\cal H}^{\otimes n})$
satisfies the conditions (C1), (C2), and (B1),
the condition (D1) holds.
\if0
any state $\sigma$ satisfies
\begin{align}
\lim_{n\to \infty}-\frac{1}{n}\log \beta_\epsilon({\cal F}_n\|\sigma^{\otimes n})
=
D({\cal F}_1\|\sigma)
=
\lim_{n\to \infty}\frac{1}{n}D({\cal F}_n\|\sigma^{\otimes n}) .
\Label{NHU5}
\end{align}
\fi
\end{corollary}
Then, the combination of 
Corollaries \ref{THC} and \ref{THC6} implies 
\eqref{BNAT1}, i.e., Theorem \ref{THC3}.

\begin{proof}
Assume the conditions (C1), (C2), and (B1).
We apply Lemma \ref{NK8}
to the inequality $p^{\otimes n} \ge \frac{1}{(n+1)^{d-1}} \rho_n(p)$
and the measurement $\{ |j\rangle\langle j|\} $ on the computation basis.
Then, any state $\rho\in {\cal F}_n$ satisfies
\begin{align}
D_\alpha(\rho\| \rho_n(p))+ (d-1)\log (n+1) 
\ge  D_\alpha(\rho\| p^{\otimes n}).
\end{align}
Notice that the states
$\rho, \rho_n(p),p^{\otimes n}$ are not changed by this measurement.
Taking the minimum for $\rho\in {\cal F}_n$, we have
\begin{align}
D_\alpha({\cal F}_n\| \rho_n(p))+ (d-1)\log (n+1) 
{\ge}  D_\alpha({\cal F}_n\| p^{\otimes n}).
\end{align}
Lemma \ref{Pr6} implies
\begin{align}
D_\alpha({\cal F}_n\| p^{\otimes n})
{\ge} n D_\alpha({\cal F}_1\| p).
\end{align}
Therefore, we have
\begin{align}
 D_\alpha({\cal F}_n\| \rho_n(p))+ (d-1)\log (n+1) 
\ge n D_\alpha({\cal F}_1\| p).\Label{ZXI}
\end{align}
The combination of Lemma \ref{LB89T} and \eqref{ZXI} implies (D1).
\end{proof}

\section{Second derivation for classical generalization
and further classical preparation}\Label{S5}
Theorem \ref{THC3} has a much simpler derivation as follows. 
although the previous derivation is a key step for our final quantum result, Theorem \ref{THC4}.
The additional aim of this section is 
further classical preparation for our analysis on the quantum case.

\begin{theorem}\Label{Pr7}
Assume that a sequence of subsets ${\cal F}_n 
\subset {\cal S}_c({\cal H}^{\otimes n})$
satisfies the conditions (C1) and (C2).
Then, we have
\begin{align}
&-\frac{1}{n}\log \beta_{e^{-nr}}({\cal F}_n\| \sigma^{\otimes n})
\nonumber \\
\ge & 
-\frac{d-1}{n}\log (n+1)
+\max_{0 \le \alpha \le 1}
\frac{\frac{1-\alpha}{n} D_\alpha({\cal F}_n\| \sigma^{\otimes n})
- \alpha r}{1-\alpha}\Label{NMY}.
\end{align}
\end{theorem}

The combination of Lemma \ref{Pr6} and Theorem \ref{Pr7}
yields Corollary \ref{THC6}, which implies Theorem \ref{THC3}. 

\begin{proofof}{Theorem \ref{Pr7}}
%Since we consider the classical case,
We denote the set of computation basis by 
${\cal X}$.
Hence, a state $\rho \in {\cal S}_c({\cal H}^{\otimes n})$
can be considered as a function on ${\cal X}^n$.
For any state $\rho_n \in {\cal F}_n$, we define 
\begin{align}
\Omega_{n}(\sigma,\rho_n,R):=&\{x^n \in {\cal X}^n|
\rho_n(x^n) \ge e^{nR}\sigma^{\otimes n}(x^n)\} \Label{VB1}\\
R_{n,\alpha,r}:= &\frac{-\frac{1-\alpha}{n} D_\alpha({\cal F}_n\| \sigma^{\otimes n})-r}{1-\alpha}.\Label{VB2}
\end{align}
We define 
\begin{align}
\Omega_{n}(\sigma,r)
:=\bigcup_{\rho_n \in {\cal F}_{n}} 
\Omega_{n}(\sigma,\rho_n,R_{n,\alpha,r}).\Label{VB3}
\end{align}
Since all states $\rho_n\in {\cal F}_n$ and $\sigma^{\otimes n}$
are permutation-invariant,
it is sufficient to discuss whether 
${\cal T}_{n,p}$ is contained in 
$\Omega_{n}(\sigma,r)$
for each $p \in {\cal T}_n$.

When ${\cal T}_{n,p}$ is contained in 
$\Omega_{n}(\sigma,r)$, 
there exists 
a state $\rho_{n,p} \in  {\cal F}_{n}$ such that
${\cal T}_{n,p} \subset \Omega_{n}(\sigma,\rho_{n,p},R_{n,\alpha,r})$.
When ${\cal T}_{n,p}$ is not contained in 
$\Omega_{n}(\sigma,r)$, 
the intersection ${\cal T}_{n,p}\cap \Omega_{n}(\sigma,r)$ is an empty set.

The above statement can be shown as follows.
When ${\cal T}_{n,p}$ is contained in $\Omega_{n}(\sigma,r)$, 
there exists a state $\rho_{n,p} \in  {\cal F}_{n}$ such that
${\cal T}_{n,p} \cap \Omega_{n}(\sigma,\rho_{n,p},R_{n,\alpha,r})
\neq \emptyset$.
We choose an element 
$x^n\in {\cal T}_{n,p} \cap \Omega_{n}(\sigma,\rho_{n,p},R_{n,\alpha,r})$.
Given a permutation $g$ on $\{1, \ldots, n\}$,
we define $g(x^n)=(x_{g(1)}, \ldots, x_{g(n)})$.
Hence, any element of ${\cal T}_{n,p}$ is written as a form 
$g(x^n)$.
Also, since the set $\Omega_{n}(\sigma,\rho_{n,p},R_{n,\alpha,r})$ is permutation invariant,
$g(x^n)$ belongs to $\Omega_{n}(\sigma,\rho_{n,p},R_{n,\alpha,r})$.
Hence, 
${\cal T}_{n,p} \subset \Omega_{n}(\sigma,\rho_{n,p},R_{n,\alpha,r})$.

Assume that ${\cal T}_{n,p}$ is not contained in 
$\Omega_{n}(\sigma,r)$.
If there exists an element $x^n \in {\cal T}_{n,p} \cap  \Omega_{n}(\sigma,r)$,
there exists a state $\rho_{n,p} \in  {\cal F}_{n}$ such that
$x^n \in{\cal T}_{n,p} \cap \Omega_{n}(\sigma,\rho_{n,p},R_{n,\alpha,r})$.
Using the above discussion, we find that
${\cal T}_{n,p} \subset \Omega_{n}(\sigma,\rho_{n,p},R_{n,\alpha,r})$, which contradicts the above assumption.
Hence, there exists no element $x^n \in {\cal T}_{n,p} \cap  \Omega_{n}(\sigma,r)$, i.e., 
the intersection ${\cal T}_{n,p}\cap \Omega_{n}(\sigma,r)$ is an empty set.

In summary, for any $p \in {\cal T}_n$, there exists 
a state $\rho_{n,p} \in  {\cal F}_{n}$ such that
\begin{align}
{\cal T}_{n,p} \cap \Omega_{n}(\sigma,r)
= \Omega_{n}(\sigma,\rho_{n,p},R_{n,\alpha,r}).
\label{BN9}
\end{align}
Therefore, 
\begin{align}
\Omega_{n}(\sigma,r)
=\bigcup_{p \in {\cal T}_{n}} 
\Omega_{n}(\sigma,\rho_{n,p},R_{n,\alpha,r}).\Label{VB4}
\end{align}
Then, $\rho_n \in {\cal F}_{n}$ satisfies
\begin{align}
&\Tr (\rho_n (I-1_{\Omega_{n}(\sigma,r)}))
\stackrel{(a)}{\le}  \Tr (\rho_n (I-1_{\Omega_{n}(\sigma,\rho_n,R_{n,\alpha,r})})) \notag\\
\stackrel{(b)}{\le} & 
\Tr ( (\rho_n)^{\alpha} (\sigma^{\otimes n})^{1-\alpha})
e^{n(1-\alpha) R_{n,\alpha,r}}
\nonumber \\
\le & 
e^{(\alpha-1)D_\alpha({\cal F}_n\|\sigma^{\otimes n})}
e^{n(1-\alpha) R_{n,\alpha,r}} 
\stackrel{(c)}{=} e^{-nr},\Label{NBC1}
\end{align}
where $(a)$, $(b)$ and $(c)$ follow from \eqref{VB4},
\eqref{VB1}, and \eqref{VB2}, respectively.

Similarly, we have
\begin{align}
&\Tr (\sigma^{\otimes n} 1_{\Omega_{n}(\sigma,r)})
\stackrel{(a)}{\le}
\sum_{p \in {\cal T}_n} 
\Tr (\sigma^{\otimes n} 1_{\Omega_{n}(\sigma,\rho_{n,p},R_{n,\alpha,r})} )\notag\\
\stackrel{(b)}{\le} &
\sum_{p \in {\cal T}_n} 
\Tr ((\rho_{n,p})^{\alpha} (\sigma^{\otimes n})^{1-\alpha})
e^{-n\alpha R_{n,\alpha,r}} 
\nonumber \\
\le & 
|{\cal T}_n|
\max_{\rho_n \in {\cal F}_{n}}
\Tr ( (\rho_{n})^{\alpha} (\sigma^{\otimes n})^{1-\alpha})
e^{-n\alpha R_{n,\alpha,r}} 
 \nonumber \\
=& 
|{\cal T}_n|
e^{(\alpha-1)D_\alpha({\cal F}_n\|\sigma^{\otimes n})}
e^{-n\alpha R_{n,\alpha,r}} 
\stackrel{(c)}{=} 
|{\cal T}_n|
e^{-
n\frac{
\frac{1-\alpha}{n} D_\alpha({\cal F}_n\| \sigma^{\otimes n})
- \alpha r}{1-\alpha}},
\Label{NMR1}
\end{align}
where $(a)$, $(b)$ and $(c)$ follow from \eqref{VB4},
\eqref{VB1}, and \eqref{VB2}, respectively.

Since $|{\cal T}_n|\le (n+1)^{d-1}$,
the combination of \eqref{NBC1} and \eqref{NMR1}
yields \eqref{NMY}.
\end{proofof}

Next, we make further classical preparation for our analysis on the quantum case.
Theorem \ref{Pr7} can be generalized as follows.
Given a parametrized set of states 
${\cal S}:=\{ \sigma_\theta\}_{\theta\in \Theta}$,
we define the set
\begin{align}
{\cal S}^{\otimes n}
:=\{\sigma_\theta^{\otimes n}\}_{\theta\in \Theta}.
\end{align}

\begin{theorem}\Label{Pr7-B}
Assume that a sequence of subsets ${\cal F}_n 
\subset {\cal S}_c({\cal H}^{\otimes n})$
satisfies the conditions (C1) and (C2).
Then, we have
\begin{align}
&-\frac{1}{n}\log \beta_{e^{-nr}}
\Big({\cal F}_n\Big\| 
\int_{\Theta}\sigma_\theta^{\otimes n} \mu(d\theta)
\Big)
\nonumber \\
\ge & 
-\frac{d-1}{n}\log (n+1)
+\max_{0 \le \alpha \le 1}
\frac{\frac{1-\alpha}{n} D_\alpha({\cal F}_n\| \int_{\Theta}\sigma_\theta^{\otimes n} \mu(d\theta))
- \alpha r}{1-\alpha}\Label{NMYB}.
\end{align}
and
\begin{align}
&-\frac{1}{n}\log \beta_{ (n+1)^{d-1}e^{-nr}}({\cal F}_n\| 
{\cal S}^{\otimes n})\nonumber \\
\ge & 
-\frac{d-1}{n}\log (n+1)
+\max_{0 \le \alpha \le 1}\frac{
\frac{1-\alpha}{n} D_\alpha({\cal F}_n\| {\cal S}^{\otimes n})
- \alpha r}{1-\alpha}\Label{NMY2}.
\end{align}
\end{theorem}

The quantity
$D_\alpha({\cal F}_n\| \sigma^{\otimes n})$
can be calculated by Lemma \ref{Pr6}
when (B1) holds.

Theorem \ref{Pr7-B} is related to \cite[Theorem 1.1]{Composite} and \cite[Theorem 5]{BDS} as follows.
Theorem \ref{Pr7-B} considers only the commutative case
while \cite[Theorem 1.1]{Composite} and \cite[Theorem 5]{BDS}consider the 
non-commutative case.
While Theorem \ref{Pr7-B} covers a general set of permutation invariant states for the first argument,
\cite[Theorem 1.1]{Composite} restricts 
the set in the first argument to be a set of $n$-tensor product states,
and \cite[Theorem 5]{BDS} restricts 
the set in the first argument to be the convex hull of a set of 
$n$-tensor product states.

\begin{proofof}{Theorem \ref{Pr7-B}}
The relation \eqref{NMYB}
can be shown in the same way as \eqref{NMY}.
To show \eqref{NMY2}, we define
\begin{align}
\Omega_{n}(r)
:=\bigcap_{\theta \in \Theta} 
\Omega_{n}(\sigma_\theta,r),
\end{align}
which implies
\begin{align}
\Omega_{n}(r)^c
=\bigcup_{\theta \in \Theta} 
\Omega_{n}(\sigma_\theta,r)^c.
\end{align}
In the same way as \eqref{BN9},
for any $p \in {\cal T}_n$, there exists 
a state $\sigma_{p} \in  {\cal S}$ such that
\begin{align}
{\cal T}_{n,p} \cap \Omega_{n}(r)^c
= \Omega_{n}(\sigma_p,r)^c.
\Label{BN8}
\end{align}
Thus,
\begin{align}
\Omega_{n}(r)^c
=\bigcup_{p \in {\cal T}_{n}} 
\Omega_{n}(\sigma_p,r)^c.
\end{align}
Then, $\rho_n \in {\cal F}_{n}$ satisfies
\begin{align}
\Tr (\rho_n (I-1_{\Omega_{n}(r)}))
\le \sum_{p \in {\cal T}_{n}} 
\Tr (\rho_n (I-1_{\Omega_{n}(\sigma_p,r)})) %\notag\\
\stackrel{(a)}{\le} |{\cal T}_{n}| e^{-nr},\Label{NBC2}
\end{align}
where $(a)$ follows from \eqref{NBC1}.

Also, $\sigma_\theta \in {\cal S}$ satisfies
\begin{align}
\Tr (\sigma_\theta^{\otimes n} 1_{\Omega_{n}(r)})
\le \Tr (\sigma_\theta^{\otimes n} 1_{\Omega_{n}(\sigma,r)})
\le
|{\cal T}_n|
e^{-n\frac{
\frac{1-\alpha}{n} D_\alpha({\cal F}_n\| \sigma^{\otimes n})
- \alpha r}{1-\alpha}}
.\Label{NMR2}
\end{align}
Hence, the combination of \eqref{NBC2} and \eqref{NMR2}
yields \eqref{NMY2}.
\end{proofof}

The following lemma evaluates
$D_\alpha({\cal F}_n\| \int_{\Theta}\sigma_\theta^{\otimes n} \mu(d\theta))$.

\begin{lemma}\Label{Pr6-5}
Assume that a sequence of subsets ${\cal F}_n 
\subset {\cal S}_c({\cal H}^{\otimes n})$
satisfies the conditions (C1), (C2), and (B1).
We consider a continuously parametrized set of states 
$\{ \sigma_\theta\}_{\theta\in \Theta}$
and a probability distribution $\mu$ on $\Theta$.  
Then, we have
\begin{align}
& D_\alpha \Big({\cal F}_n\Big\| 
\int_{\Theta}\sigma_\theta^{\otimes n} \mu(d\theta)\Big) \notag\\
\ge &
n \min_{ \theta\in \Theta} D_\alpha({\cal F}_1\| \sigma_\theta)
-\frac{(d-1)}{1-\alpha}\log (n+1)\Label{ASK1}
\end{align}
for $\alpha \in [0,1)$.
%When the condition (A4) holds additionally,
In addition, when
the set $\Theta$ is a subset of a real vector space and
the set $\Theta$ is
the support of the probability density function of 
$\mu$ with respect to the Lebesgue measure,
we have
\begin{align}
\lim_{n\to \infty}\frac{1}{n}D
\Big(
{\cal F}_n
\Big\| 
\int_{\Theta}\sigma_\theta^{\otimes n} \mu(d\theta)
\Big)
= \min_{ \theta\in \Theta} D({\cal F}_1\| \sigma_\theta).
\Label{CBQ3}
\end{align}
\end{lemma}

\begin{proof}
\if0
Since 
$1^n_{p}
\int_{\Theta}\sigma_\theta^{\otimes n} \mu(d\theta)$
is a constant, 
$\int_{\Theta}\sigma_\theta^{\otimes n} \mu(d\theta)
=\sum_{p \in  {\cal T}_n}1^n_{p}
\int_{\Theta}\sigma_\theta^{\otimes n} \mu(d\theta)$
can be considered as a function of 
$p \in  {\cal T}_n$.
That is, 
$\int_{\Theta}\sigma_\theta^{\otimes n} \mu(d\theta)$
can be considered as an element of 
$|{\cal T}_n|$-dimensional space.
Caratheodory theorem guarantees that
there are $|{\cal T}_n|+1$ elements 
$\theta_1, \ldots, \theta_{|{\cal T}_n|+1} \in \Theta$ and a distribution 
$q$ on $\{1, \ldots, |{\cal T}_n|+1\}$
such that
\begin{align}
\int_{\Theta}\sigma_\theta^{\otimes n} \mu(d\theta)
=
\sum_{j=1}^{|{\cal T}_n|+1} q(j) \sigma_{\theta_j}^{\otimes n}.
\end{align}
We consider 
the states
$\sum_{j=1}^{|{\cal T}_n|+1} q(j) |j\rangle \langle j| \otimes \rho_n$
and
$\sum_{j=1}^{|{\cal T}_n|+1} q(j) |j\rangle \langle j| 
\otimes \sigma_{\theta_j}^{\otimes n}$
on 
the joint system of 
$\{1, \ldots, |{\cal T}_n|+1\}$ and
${\cal H}^{\otimes n}$.
The information processing inequality implies 
\begin{align}
&D(\rho_n\|\sum_{j=1}^{|{\cal T}_n|+1} q(j) \sigma_{\theta_j}^{\otimes n})
\\
\le &
D
(\sum_{j=1}^{|{\cal T}_n|+1} q(j) |j\rangle \langle j| 
\otimes \rho_n\otimes \|
\sum_{j=1}^{|{\cal T}_n|+1} q(j) |j\rangle \langle j| 
\otimes \sigma_{\theta_j}^{\otimes n}) \\
=&
\sum_{j=1}^{|{\cal T}_n|+1} q(j) 
D(\rho_n\|\sigma_{\theta_j}^{\otimes n})
\end{align}
\fi
Since 
$1^n_{p}
\int_{\Theta}\sigma_\theta^{\otimes n} \mu(d\theta)$
is a constant, 
there exists an element $\theta(n,p) \in \Theta$
such that
$1^n_{p}
\int_{\Theta}\sigma_\theta^{\otimes n} \mu(d\theta)
\le 
1^n_{p} \sigma_{\theta(n,p)}^{\otimes n}$.
Thus, for $\alpha \in [0,1)$, we have
\begin{align}
\Big(1^n_{p}
\int_{\Theta}\sigma_\theta^{\otimes n} \mu(d\theta)\Big)^{1-\alpha}
\le  &
\Big(1^n_{p} \sigma_{\theta(n,p)}^{\otimes n}\Big)^{1-\alpha} \notag\\
\le &
\sum_{p'\in {\cal T}_n}
\Big(1^n_{p} \sigma_{\theta(n,p')}^{\otimes n}\Big)^{1-\alpha}.
\end{align}
Hence,
\begin{align}
&\Big(\int_{\Theta}\sigma_\theta^{\otimes n} \mu(d\theta)\Big)^{1-\alpha}
=
\sum_{p\in {\cal T}_n}
\Big(1^n_{p}\int_{\Theta}\sigma_\theta^{\otimes n} \mu(d\theta)\Big)^{1-\alpha}
\notag\\
\le & 
\sum_{p\in {\cal T}_n}
\sum_{p'\in {\cal T}_n}
\Big(1^n_{p} \sigma_{\theta(n,p')}^{\otimes n}\Big)^{1-\alpha} 
=
\sum_{p'\in {\cal T}_n}
\Big( \sigma_{\theta(n,p')}^{\otimes n}\Big)^{1-\alpha} .
\end{align}
Therefore,
\begin{align}
&e^{(\alpha-1)D_\alpha({\cal F}_{n}\| \int_{\Theta}\sigma_\theta^{\otimes n} \mu(d\theta))}
=
\max_{\rho_n \in {\cal F}_{n}}
\Tr (\rho_n^{\alpha}
\Big(\int_{\Theta}\sigma_\theta^{\otimes n} \mu(d\theta)\Big)^{1-\alpha})\notag\\
\le &
\max_{\rho_n \in {\cal F}_{n}}
\Tr (\rho_n^{\alpha}
\sum_{p'\in {\cal T}_n}
\Big( 
\sigma_{\theta(n,p')}^{\otimes n}\Big)^{1-\alpha} )\notag\\
= &
\sum_{p'\in {\cal T}_n}
\max_{\rho_n \in {\cal F}_{n}}
\Tr( \rho_n^{\alpha}
\Big( \sigma_{\theta(n,p')}^{\otimes n}\Big)^{1-\alpha}) \notag\\
\le &
\sum_{p'\in {\cal T}_n}
\max_{\theta\in \Theta}
\max_{\rho_n \in {\cal F}_{n}}
\Tr \rho_n^{\alpha}
\Big( \sigma_{\theta}^{\otimes n}\Big)^{1-\alpha} \notag\\
=&
|{\cal T}_n|
e^{(\alpha-1)\min_{\theta\in \Theta}
D_\alpha({\cal F}_{n}\| \sigma_\theta^{\otimes n} )}\notag\\
\le &
(n+1)^{(d-1)}
e^{(\alpha-1)\min_{\theta\in \Theta}
D_\alpha({\cal F}_{n}\| \sigma_\theta^{\otimes n} )},
\end{align}
which implies 
\begin{align*}
e^{D_\alpha({\cal F}_{n}\| \int_{\Theta}\sigma_\theta^{\otimes n} \mu(d\theta))}
\ge 
(n+1)^{-\frac{(d-1)}{1-\alpha}}
e^{-\min_{\theta\in \Theta}
D_\alpha({\cal F}_{n}\| \sigma_\theta^{\otimes n} )}
\end{align*}
that coincides with \eqref{ASK1} for $\alpha \in [0,1)$.

\if0
Also, we have
\begin{align}
&D({\cal F}_{n}\| \int_{\Theta}\sigma_\theta^{\otimes n} \mu(d\theta))
=\min_{\rho_n \in {\cal F}_{n}}
D(\rho_{n}\| \int_{\Theta}\sigma_\theta^{\otimes n} \mu(d\theta))
\notag \\
=&
\min_{\rho_n \in {\cal F}_{n}}
\Tr \rho_{n}\Big(\log \rho_{n}
-\log( \int_{\Theta}\sigma_\theta^{\otimes n} \mu(d\theta))\Big)\notag \\
\ge &
\min_{\rho_n \in {\cal F}_{n}}
\Tr \rho_{n}\Big(\log \rho_{n}
-\log(
\sum_{p'\in {\cal T}_n}
 \sigma_{\theta(n,p')}^{\otimes n} )\Big)\notag \\
\end{align}
\fi

In addition, we have
\begin{align}
&\frac{1}{n}D \Big(
{\cal F}_n
\Big\| 
\int_{\Theta}\sigma_\theta^{\otimes n} \mu(d\theta)
\Big) 
\ge 
\frac{1}{n} D_\alpha \Big({\cal F}_n\Big\| 
\int_{\Theta}\sigma_\theta^{\otimes n} \mu(d\theta)\Big) \notag\\
\ge &
\min_{ \theta\in \Theta} D_\alpha({\cal F}_1\| \sigma_\theta)
-\frac{(d-1)}{n(1-\alpha)}\log (n+1)\Label{ASK1T}
\end{align}
for $\alpha \in [0,1)$.
Taking the limit $n \to \infty$ in \eqref{ASK1T}, 
we have 
\begin{align}
&\lim_{n\to \infty}
\frac{1}{n}D \Big(
{\cal F}_n
\Big\| 
\int_{\Theta}\sigma_\theta^{\otimes n} \mu(d\theta)
\Big) 
\ge 
\min_{ \theta\in \Theta} D_\alpha({\cal F}_1\| \sigma_\theta)
\Label{ASK1Y}
\end{align}
for $\alpha \in [0,1)$.
Taking the limit $\alpha \to 1$ in \eqref{ASK1Y},
we obtain the part $\ge$ in \eqref{CBQ3}.

We choose $\theta^* := \argmin_{\theta\in \Theta}
D({\cal F}_{1}\| \sigma_\theta )$.
We set a neighborhood  $U_{\theta^*,\epsilon}
:=\{\theta \in \Theta | \|\theta-\theta^*\| \le \epsilon\}$.
We have
\begin{align}
& D\Big(
{\cal F}_n
\Big\| 
\int_{\Theta}\sigma_\theta^{\otimes n} \mu(d\theta)
\Big) 
\le 
D\Big(
{\cal F}_n
\Big\| 
\int_{U_{\theta^*,\epsilon}}\sigma_\theta^{\otimes n} \mu(d\theta)
\Big) \notag \\
= &
D\Big(
{\cal F}_n
\Big\| \frac{1}{\mu(U_{\theta^*,\epsilon})}
\int_{U_{\theta^*,\epsilon}}\sigma_\theta^{\otimes n} \mu(d\theta)
\Big) 
-\log \mu(U_{\theta^*,\epsilon})\notag \\
\le &
\frac{1}{\mu(U_{\theta^*,\epsilon})}
\int_{U_{\theta^*,\epsilon}}
D({\cal F}_n\| \sigma_\theta^{\otimes n} ) 
\mu(d\theta)
-\log \mu(U_{\theta^*,\epsilon})\notag \\
\le &
\frac{n}{\mu(U_{\theta^*,\epsilon})}
\int_{U_{\theta^*,\epsilon}}
D({\cal F}_1\| \sigma_\theta) 
\mu(d\theta)
-\log \mu(U_{\theta^*,\epsilon}).
\Label{CBQ7}
\end{align}
Taking the limit $n\to \infty$, we have
\begin{align}
&\limsup_{n\to \infty}\frac{1}{n}D
\Big(
{\cal F}_n
\Big\| 
\int_{\Theta}\sigma_\theta^{\otimes n} \mu(d\theta)
\Big) \notag \\
\le &
\frac{1}{\mu(U_{\theta,\epsilon})}
\int_{U_{\theta,\epsilon}}
D({\cal F}_1\| \sigma_\theta) 
\mu(d\theta).
\Label{CBQ8}
\end{align}
Taking the limit $\epsilon \to 0$, we have
we obtain the part $\le$ in 
\eqref{CBQ3}.
\end{proof}

Combining Theorem \ref{Pr7-B} and Lemmas \ref{Pr6} and \ref{Pr6-5}, we have the following corollary.
\begin{corollary}\Label{CBD1}
Assume the same condition as Lemma \ref{Pr6-5}.
\if0
that a sequence of subsets ${\cal F}_n 
\subset {\cal S}_c({\cal H}^{\otimes n})$
satisfies the conditions (C1), (C2), and (B1).
\fi
Then, we have
\begin{align}
\lim_{n\to \infty}\frac{-1}{n}\log \beta_{\epsilon}({\cal F}_n\| 
{\cal S}^{\otimes n})
=\min_{ \theta\in \Theta} D({\cal F}_1\| \sigma_\theta),
\Label{CV6}
\end{align}
and
\begin{align}
&\lim_{n\to \infty}\frac{-1}{n}\log \beta_{\epsilon}
\Big(
{\cal F}_n
\Big\| 
\int_{\Theta}\sigma_\theta^{\otimes n} \mu(d\theta)
\Big) 
=
\min_{ \theta\in \Theta} D({\cal F}_1\| \sigma_\theta).
\Label{CV4}
\end{align}
\end{corollary}

\begin{proof}
The combination of \eqref{NMY2} in Theorem \ref{Pr7-B}
and Lemma \ref{Pr6}
yields
\begin{align}
\liminf_{n\to \infty}\frac{-1}{n}\log \beta_{\epsilon}({\cal F}_n\| 
{\cal S}^{\otimes n})
&\ge \min_{ \theta\in \Theta} D_\alpha({\cal F}_1\| \sigma_\theta)
\Label{XAB1}
\end{align}
for $\alpha \in [0,1)$. 
Using $\theta^* := \argmin_{\theta\in \Theta}
D({\cal F}_{1}\| \sigma_\theta )$,
we have
\begin{align}
& D({\cal F}_1\| \sigma_{\theta^*})
\stackrel{(a)}{=}
\limsup_{n\to \infty}\frac{-1}{n}\log \beta_{\epsilon}({\cal F}_n\| 
\sigma_{\theta^*}^{\otimes n}) \notag\\
\ge &
\limsup_{n\to \infty}\frac{-1}{n}\log \beta_{\epsilon}({\cal F}_n\| 
{\cal S}^{\otimes n})
\ge
\liminf_{n\to \infty}\frac{-1}{n}\log \beta_{\epsilon}({\cal F}_n\| 
{\cal S}^{\otimes n}) \notag\\
\stackrel{(b)}{\ge} & \min_{ \theta\in \Theta} D({\cal F}_1\| \sigma_\theta),
\end{align}
where $(a)$ follows from Corollary \ref{THC6}, 
which follows from
the combination of Lemma \ref{Pr6} and Theorem \ref{Pr7},
and $(b)$ follows from \eqref{XAB1} with the limit $\alpha\to 1$.
We obtain \eqref{CV6}.

The combination of \eqref{NMYB} in Theorem \ref{Pr7-B}
and \eqref{ASK1} in Lemma \ref{Pr6-5} yields
\begin{align}
\liminf_{n\to \infty}\frac{-1}{n}\log \beta_{\epsilon}
\Big(
{\cal F}_n
\Big\| 
\int_{\Theta}\sigma_\theta^{\otimes n} \mu(d\theta)
\Big)
&\ge \min_{ \theta\in \Theta} D_\alpha({\cal F}_1\| \sigma_\theta)
\end{align}
for $\alpha \in [0,1)$. Taking the limit $\alpha\to 1$, we have
\begin{align}
\liminf_{n\to \infty}\frac{-1}{n}\log \beta_{\epsilon}
\Big(
{\cal F}_n
\Big\| 
\int_{\Theta}\sigma_\theta^{\otimes n} \mu(d\theta)
\Big)
&\ge \min_{ \theta\in \Theta} D({\cal F}_1\| \sigma_\theta).
\Label{BV1}
\end{align}
\if0
The information processing inequality implies
\begin{align}
&D\Big(
{\cal F}_n
\Big\| 
\int_{\Theta}\sigma_\theta^{\otimes n} \mu(d\theta)
\Big) \\
\ge &
\epsilon \Big(\log \epsilon 
-\log \beta_{\epsilon}
\Big(
{\cal F}_n
\Big\| 
\int_{\Theta}\sigma_\theta^{\otimes n} \mu(d\theta)
\Big) \Big) \\
&+
(1-\epsilon) \Big(\log (1-\epsilon )
-\log \Big(1-\beta_{\epsilon}
\Big(
{\cal F}_n
\Big\| 
\int_{\Theta}\sigma_\theta^{\otimes n} \mu(d\theta)
\Big) \Big) \Big) \\
\ge &
\epsilon \Big(\log \epsilon 
-\log \beta_{\epsilon}
\Big(
{\cal F}_n
\Big\| 
\int_{\Theta}\sigma_\theta^{\otimes n} \mu(d\theta)
\Big) \Big) \\
\end{align}
\fi
Applying Theorem \ref{BFG8} with $\sigma_n=\int_{\Theta}\sigma_\theta^{\otimes n} \mu(d\theta)$
to the pair of \eqref{CBQ3} in Lemma \ref{Pr6-5} and \eqref{BV1}, 
we obtain \eqref{CV4}.
\end{proof}

\begin{remark}
As pointed in \cite[Appendix IV]{HY}, \cite{HMH},
the statement similar to 
Corollary \ref{CBD1} does not hold
without the convexity condition for ${\cal S}$ in the quantum case.
In the quantum case, the choice of measurement 
is crucial and it depends on the second argument of the relative entropy.
When ${\cal S}$ is not convex,
the closest element to ${\cal F}_n$ among ${\cal S}^{\otimes n}$
is not unique in general.
Hence, the optimal element determined by the closest element 
does not work for detecting another element of ${\cal S}^{\otimes n}$.
This problem structure causes the difficulty of the generalized Stein's lemma \cite{BP,HY}. 
\end{remark}

\section{Quantum empirical distribution}\Label{S6}
\subsection{Notations for quantum empirical distribution}
To set the stage for our results in the quantum domain,  
we review a quantum version of Sanov's theorem proposed by \cite{another},  
which is based on the concept of quantum empirical distribution and differs from the hypothesis-testing approach  
used in the quantum Sanov theorem by the paper \cite{Notzel,Bjelakovic}.  
The extension by the paper \cite{another} has a structure similar to that of the classical Sanov theorem \cite{Bucklew,DZ},  
relying on empirical distributions rather than hypothesis testing.  
We now introduce the necessary notation to present the quantum Sanov theorem for quantum empirical distribution.  

Consider a basis $\mathcal{B} = \{|v_j\rangle\}_{j=1}^d$ for the Hilbert space $\mathcal{H}$.  
The space $\mathcal{H}^{\otimes n}$ has a standard basis given by  
$|v[x^n]\rangle := |v_{x_1}, \ldots, v_{x_n}\rangle$ for $x^n = (x_1, \ldots, x_n)$.  
When performing a measurement on $\mathcal{H}^{\otimes n}$ in basis $\mathcal{B}$,  
the outcome is represented by the sequence $x^n$.  
The empirical state associated with $x^n$ under basis $\mathcal{B}$ is  
$\sum_{j=1}^n \frac{1}{n} |v_{x_j}\rangle \langle v_{x_j}|$,  
denoted as $e_{\mathcal{B}}(x^n)$.  
The pinching map ${\cal E}_{{\cal B}}$ is defined as  
\begin{align}  
{\cal E}_{{\cal B}}(\sigma) := \sum_x |v_{x}\rangle \langle v_{x}| \sigma |v_{x}\rangle \langle v_{x}|.  
\Label{DX2}  
\end{align}  

For the $n$-fold tensor product space,  
the set of diagonal states relative to $\mathcal{B}$ is denoted by $\mathcal{S}[\mathcal{B}]$:  
\begin{align}  
\mathcal{S}[\mathcal{B}] &= \left\{ \sum_{j=1}^d p_j |v_{j}\rangle \langle v_{j}| \;\middle|\; \sum_{j=1}^d p_j = 1, \quad p_j \ge 0 \right\},  
\Label{DX3B}  
\\  
\mathcal{S}_n[\mathcal{B}] &= \left\{ \sum_{j=1}^n \frac{1}{n} |v_{x_j}\rangle \langle v_{x_j}| \;\middle|\; (x_1, \ldots, x_n) \in \mathcal{X}^n \right\}.  
\Label{DX3}  
\end{align}  

For any empirical state $\rho \in \mathcal{S}_n[\mathcal{B}]$,  
the projection operator is defined as
\begin{align}  
1_{\rho,\mathcal{B}}^n := \sum_{x^n: e_{\mathcal{B}}(x^n) = \rho} |v[x^n]\rangle \langle v[x^n]|.  
\Label{DX1}  
\end{align}  
If $\rho \notin \mathcal{S}_n[\mathcal{B}]$, we set $1_{\rho,\mathcal{B}}^n = 0$.  

Then, the original classical Sanov theorem is stated as follows.
\begin{proposition}\Label{Sanov-ori}\cite[Theorem 2.1.10 and Eq. (2.1.12)]{DZ}
A pair of a state $\rho \in \mathcal{S}[\mathcal{B}]$ 
and a subset ${\cal S} \subset \mathcal{S}[\mathcal{B}]$
satisfy the following inequality
\begin{align}  
\Tr \rho^{\otimes n}  
\left(\sum_{\rho' \in {\cal S} \cap \mathcal{S}_n[\mathcal{B}]}  
1_{\rho',\mathcal{B}}^n \right)  
\le (n+1)^{d-1}  
2^{-n \inf_{\rho' \in {\cal S}} D(\rho'\|\rho)}. \Label{IOF3}  
\end{align}  
In addition, when the closure of the interior of ${\cal S}$
equals to the closure of ${\cal S}$,
\begin{align}  
\lim_{n\to \infty}\frac{-1}{n}\log \Tr \rho^{\otimes n}  
\left(\sum_{\rho' \in {\cal S} \cap \mathcal{S}_n[\mathcal{B}]}  
1_{\rho',\mathcal{B}}^n \right)  
= \inf_{\rho' \in {\cal S}} D(\rho'\|\rho). 
\Label{IOF4}  
\end{align}  
\end{proposition}

We now apply Schur duality for $\cH^{\otimes n}$,  
as used in \cite{H-01,H-02,KW01,CM,CHM,OW,HM02a,HM02b,Notzel,H-24,H-q-text,AISW}.  
A sequence of non-negative integers $\blambda = (\lambda_1, \ldots, \lambda_d)$,  
ordered in non-decreasing form, is known as a Young index.  
Although other references \cite{H-q-text,Group1,GW} order Young indices in decreasing form,  
we adopt the opposite convention for convenience.  
The set of Young indices $\blambda$ with $\sum_{j=1}^d \lambda_j = n$ is denoted by $Y_d^n$.  
We denote by $\mathcal{P}_d$ the set of probability distributions $p = (p_j)_{j=1}^d$  
satisfying $p_1 \le p_2 \le \ldots \le p_d$.  
For any density matrix $\rho$ on $\mathcal{H}$,  
$p(\rho)$ represents the eigenvalues of $\rho$, ordered in $\mathcal{P}_d$.  
The map $\overline{\cal E}_{\cal B}$ is defined as  
\begin{align}  
\overline{\cal E}_{\cal B}(\sigma) := (p(\sigma), {\cal E}_{\cal B}(\sigma)).  
\Label{NB3}  
\end{align}  
The set $\mathcal{P}_d^n$ consists of elements of $\mathcal{P}_d$  
where each $p_j$ is an integer multiple of $1/n$.  
The majorization relation $\succ$ between elements of $\mathcal{P}_d$ and $Y_d^n$  
is defined as follows: for $p, p' \in \mathcal{P}_d$,  
$p \succ p'$ if  
\begin{align}  
\sum_{j=1}^k p_j \ge \sum_{j=1}^k p_j' \quad \text{for} \quad k = 1, \ldots, d-1.  
\Label{DX5}  
\end{align}  
This condition also defines the majorization relation on $Y_d^n$.  

According to \cite[Section 6.2]{H-q-text}, the decomposition of $\cH^{\otimes n}$ is  
\begin{align}  
\cH^{\otimes n} = \bigoplus_{\blambda \in Y_d^n} \mathcal{U}_{\blambda} \otimes \mathcal{V}_{\blambda},  
\end{align}  
where $\mathcal{U}_{\blambda}$ and $\mathcal{V}_{\blambda}$ denote  
the irreducible representations of $\SU(d)$ and the permutation group $\frS_n$, respectively.  
The dimensions are defined as  
\begin{align}  
d_{\blambda} := \dim \mathcal{U}_{\blambda}, \quad d_{\blambda}' := \dim \mathcal{V}_{\blambda},  
\quad \overline{d}_n := \sum_{\blambda \in Y_d^n} d_{\blambda}.  
\Label{DX7}  
\end{align}  
From \cite[(6.16)]{H-q-text},  
\begin{align}  
d_{\blambda} \le (n+1)^{\frac{d(d-1)}{2}}.  
\Label{NMI}  
\end{align}  

For any subspace ${\cal K} \subset \mathcal{H}^{\otimes n}$,  
the projection operator onto ${\cal K}$ is denoted by $1_{\mathcal{K}}$.  
Given $np \in Y_d^n$ and $\rho \in \mathcal{S}_n[\mathcal{B}]$,  
the projection $1_{p,\rho,\mathcal{B}}^n$ is defined as  
\begin{align}  
1_{p,\rho,\mathcal{B}}^n := 
1_{\mathcal{U}_{np} \otimes \mathcal{V}_{np}} 1_{\rho,\mathcal{B}}^n.  
\Label{DU1}  
\end{align}  
The condition $1_{p,\rho,\mathcal{B}}^n \neq 0$ holds if and only if  
\begin{align}  
p \succ p(\rho).  \Label{ZXN}  
\end{align}  
The subspace ${\cal U}_{np}$ is decomposed as  
\begin{align}  
{\cal U}_{np} = \bigoplus_{\rho \in \mathcal{S}_n[\mathcal{B}]: p \succ p(\rho)} {\cal U}_{np,\rho},  
\end{align}  
where ${\cal U}_{np,\rho}$ is associated with the weight vector $np(\rho)$.  
Hence, we have
\begin{align}  
1_{p,\rho,\mathcal{B}}^n := 
1_{\mathcal{U}_{np,\rho} \otimes \mathcal{V}_{np}} 
=1_{\mathcal{U}_{np,\rho}} \otimes 1_{\mathcal{V}_{np}} 
\Label{DU1V}  .
\end{align}  

We define
\begin{align}
d_{np,\rho}:=\dim {\cal U}_{np,\rho}.\Label{DV1}
\end{align}
Then, we have
\begin{align}  
\Tr 1_{p,\rho,\mathcal{B}}^n = d_{np,\rho} d_{np}'.  
\end{align}  
The completely mixed state $\rho_{n,{\cal B}}(p,\rho)$ on the support of $1_{p,\rho,\mathcal{B}}^n$ is defined as 
\begin{align}  
\rho_{n,{\cal B}}(p,\rho) := \frac{1}{d_{np,\rho} d_{np}'} 1_{p,\rho,\mathcal{B}}^n.  
\Label{DV2}  
\end{align}  
In addition, we define the sets  
\begin{align}  
\mathcal{R}[\mathcal{B}] &:= \{(p,\rho) \in \mathcal{P}_d \times \mathcal{S}[\mathcal{B}] \mid p \succ p(\rho)\},  
\Label{DV3}  
\\  
\mathcal{R}_n[\mathcal{B}] &:= \mathcal{R}[\mathcal{B}] \cap 
(\mathcal{P}_d^n \times \mathcal{S}_n[\mathcal{B}]).  
\Label{DV4}  
\end{align}  

Hence, $\overline{\cal E}_{\cal B}(\sigma) \in \mathcal{R}[\mathcal{B}]$.  
The decomposition $\{1_{p,\rho,\mathcal{B}}^n\}_{(p,\rho) \in \mathcal{R}_n[\mathcal{B}]}$  
represents the measurement for obtaining the quantum empirical distribution under basis $\mathcal{B}$  
and the Schur sampling method \cite{KW01,CM,CHM,OW,HM02a,HM02b,AISW}.  
Tables \ref{notation-spaces}, \ref{notation-dimensions}, and \ref{notation-projections}  
summarize the notations introduced in this subsection.

\begin{table}[t]
\caption{Notations of spaces}
\label{notation-spaces}
\begin{center}
\begin{tabular}{|c|l|}
\hline
Symbol &Description     \\
\hline
\multirow{2}{*}{$\mathcal{U}_{\blambda}$} & Irreducible  representation space  \\ 
&of $\SU(d)$ 
\\ \hline
\multirow{2}{*}{$\mathcal{V}_{\blambda}$} & 
Irreducible representation space 
\\ 
& of permutation group $\frS_n$ 
\\ \hline
\multirow{2}{*}{${\cal U}_{np,\rho}$} &
Subspace of the space ${\cal U}_{np}$\\
& with weight vector $n p(\rho)$ 
\\ \hline
\end{tabular}
\end{center}
\end{table}

\begin{table}[t]
\caption{Notations of dimensions and related numbers}
\label{notation-dimensions}
\begin{center}
\begin{tabular}{|c|l|c|}
\hline
Symbol &Description &Eq. number    \\
\hline
$d_{\blambda}$ & Dimension of  $\mathcal{U}_{\blambda}$
& Eq. \eqref{DX7}
\\ \hline
$d_{\blambda}'$ & Dimension of  $\mathcal{V}_{\blambda}$
& Eq. \eqref{DX7}
\\ \hline
$\overline{d}_n$ &
$\sum_{\blambda \in Y_d^n}d_{\blambda}$
& Eq. \eqref{DX7}
\\ \hline
$d_{np,\rho}$ & Dimension of ${\cal U}_{np,\rho}$ &
Eq. \eqref{DV1} \\ \hline
\end{tabular}
\end{center}
\end{table}

\begin{table}[t]
\caption{Notations of projections}
\label{notation-projections}
\begin{center}
\begin{tabular}{|c|l|c|}
\hline
Symbol &Description &Eq. number    \\
\hline
$1_{\rho,\mathcal{B}}^n$ & 
$\sum_{x^n: e_{\mathcal{B}}(x^n)
=\rho }|v[x^n]\rangle \langle v[x^n]|$ &Eq. \eqref{DX1}
 \\
\hline
$1_{{\cal K}}$ & Projection to subspace ${\cal K}$
&\\ \hline
$1_{p,\rho,\mathcal{B}}^n$& &Eq. \eqref{DU1}
\\ \hline
\end{tabular}
\end{center}
\end{table}

\subsection{Quantum Sanov theorem based on quantum empirical distribution}
We define the TP-CP map ${\cal E}_{n,{\cal B}}$ as  
\begin{align}  
{\cal E}_{n,{\cal B}}(\sigma) &:=  
\sum_{(p,\rho) \in \mathcal{R}_n[\mathcal{B}]}  
(\Tr 1_{p,\rho,\mathcal{B}}^n \sigma)  
\rho_{n,{\cal B}}(p,\rho).  
\Label{DV5}  
\end{align}  

For a state $\rho \in \mathcal{S}[\mathcal{B}]$ and a pair $(p',\rho') \in \mathcal{R}[\mathcal{B}]$, we define  
\begin{align}  
D_{\mathcal{B}}((p',\rho') \|\rho) :=  
D(\rho'\| \rho) + S(\rho') - S(p'), \Label{DV6}  
\end{align}  
where $S(\rho) := - \Tr(\rho \log \rho)$.  
The following subset of $\mathcal{R}[\mathcal{B}]$ plays a significant role in the quantum Sanov theorem for quantum empirical distributions.  
For a state $\rho \in \mathcal{S}[\mathcal{B}]$ and a positive value $r > 0$, we define  
\begin{align}  
S_{\rho,r} :=  
\left\{ (p',\rho') \in \mathcal{R}[\mathcal{B}] \; \middle| \;  
D_{\mathcal{B}}((p',\rho') \|\rho) \le r \right\},  
\Label{DV7}  
\end{align}  
with $S_{\rho,r}^c$ representing its complement.  
The properties of the set $S_{\rho,r}$ are described in the following result. Readers may refer to \cite{another} for the proof.  

\begin{proposition}[\protect{\cite[Lemma 1]{another}}]\Label{LL1}  
For $\rho \in \mathcal{S}[\mathcal{B}]$ and $r > 0$, the following inequality holds:  
\begin{align}  
\Tr \rho^{\otimes n}  
\left(\sum_{(p',\rho') \in S_{\rho,r}^c \cap \mathcal{R}_n[\mathcal{B}]}  
1_{p',\rho',\mathcal{B}}^n \right)  
\le (n+1)^{\frac{(d+4)(d-1)}{2}}  
2^{-nr}. \Label{IOY}  
\end{align}  
Additionally, we have  
\begin{align}  
\lim_{n \to \infty} -\frac{1}{n} \log \Tr \rho^{\otimes n}  
\left(\sum_{(p',\rho') \in S_{\rho,r}^c \cap \mathcal{R}_n[\mathcal{B}]}  
1_{p',\rho',\mathcal{B}}^n \right)  
= r. \Label{IBT3}  
\end{align}  
\end{proposition}  

\subsection{Useful formulas related to the map ${\cal E}_{n,{\cal B}}$}
\Label{S6C}
This subsection presents 
several useful formulas related to the map ${\cal E}_{n,{\cal B}}$
for the preparation for the latter discussion.

\begin{lemma}
A permutation-invariant
state $\rho_n$ on ${\cal H}^{\otimes n}$ satisfies the relation
\begin{align}
\rho_n \le \overline{d}_n {\cal E}_{n,{\cal B}}(\rho_n).\Label{BR2T}
\end{align}
\end{lemma}

\begin{proof}
We diagonalize the matrix $1_{p,\rho,\mathcal{B}}^n
\rho_n
1_{p,\rho,\mathcal{B}}^n$
as $ \sum_{j} c_{p,\rho,j} |x(p,\rho,j)\rangle\langle x(p,\rho,j)|
\otimes 1_{\mathcal{V}_{\blambda}}$, where
$1_{\mathcal{V}_{\blambda}}$ is the identity operator on 
$\mathcal{V}_{\blambda}$.
Hence, 
the subspace ${\cal U}_{np,\rho}$
is spanned by $\{|x(p,\rho,j)\rangle\}_j$.
We choose a basis 
$\{|y(p,\rho,j)\rangle\}_j$
mutually unbiased to $\{|x(p,\rho,j)\rangle\}_j$
on the subspace ${\cal U}_{np,\rho}$.
Then, we have
\begin{align}
&\langle y(p,\rho,j)| \Big(\sum_{k} c_{p,\rho,k} |x(p,\rho,k)\rangle\langle x(p,\rho,k)|\Big)|y(p,\rho,j)\rangle \notag\\
=&\langle y(p,\rho,j')| \Big(\sum_{k} c_{p,\rho,k} |x(p,\rho,k)\rangle\langle x(p,\rho,k)|\Big)|y(p,\rho,j')\rangle.
\end{align}
Hence, we have
\begin{align}
&\sum_j \langle y(p,\rho,j)| 
\Big(\sum_{k} c_{p,\rho,k} |x(p,\rho,k)\rangle\langle x(p,\rho,k)|
\Big)
|y(p,\rho,j)\rangle \notag\\
&\cdot|y(p,\rho,j)\rangle \langle y(p,\rho,j)|\notag\\
=&
\Tr \Big(\sum_{k} c_{p,\rho,k} |x(p,\rho,k)\rangle\langle x(p,\rho,k))|
\Big)\notag\\
&\cdot\frac{1}{\dim {\cal U}_{np,\rho}}
\sum_j
|y(p,\rho,j)\rangle \langle y(p,\rho,j)|\notag\\
=&
\Tr \Big(\sum_{k} c_{p,\rho,k} |x(p,\rho,k)\rangle\langle x(p,\rho,k)|
\Big)
\frac{1}{\dim {\cal U}_{np,\rho}}
1_{{\cal U}_{np,\rho}}.
\end{align}
Thus, using $P_{p,\rho,j}:=|y(p,\rho,j)\rangle \langle y(p,\rho,j)| \otimes 1_{\mathcal{V}_{\blambda}}$,
we have
\begin{align}
&\sum_j 
 P_{p,\rho,j}  \rho_n P_{p,\rho,j}\notag\\
=&
\Big(\sum_j \langle y(p,\rho,j)| \Big(\sum_{k} c_{p,\rho,k} |x(p,\rho,k)\rangle\langle x(p,\rho,k)|\Big)
|y(p,\rho,j)\rangle \notag\\
&\cdot|y(p,\rho,j)\rangle \langle y(p,\rho,j)|\Big)
\otimes 1_{\mathcal{V}_{\blambda}} \notag\\
=&
\Tr \Big(\sum_{k} c_{p,\rho,k} |x(p,\rho,k)\rangle\langle x(p,\rho,k)|\Big)\notag\\
&\cdot\frac{1}{\dim {\cal U}_{np,\rho}}
(1_{{\cal U}_{np,\rho}}
\otimes 1_{\mathcal{V}_{\blambda}})  \notag\\
=&
\Tr  (\rho_n 1_{p,\rho,\mathcal{B}}^n )
\rho_{n,{\cal B}}(p,\rho).
\Label{BN6}
\end{align}
Now, we define the pinching map 
${\cal E}_{n,{\cal B},\rho_n}$ as
\begin{align}
{\cal E}_{n,{\cal B},\rho_n}(\sigma):=
 \sum_{p,\rho,j}P_{p,\rho,j} \sigma P_{p,\rho,j} .\Label{BD9}
\end{align}
Then, we have
\begin{align}
&\overline{d}_n {\cal E}_{n,{\cal B}}(\rho_n)
=
\overline{d}_n \sum_{(p,\rho) \in \mathcal{R}_n[\mathcal{B}]}
(\Tr 1_{p,\rho,\mathcal{B}}^n \rho_n)
\rho_{n,{\cal B}}(p,\rho) \notag\\
\stackrel{(a)}{=} &
\overline{d}_n {\cal E}_{n,{\cal B},\rho_n}( \rho_n) 
\stackrel{(b)}{\ge} 
 \rho_n,\Label{NM6A}
\end{align}
where
$(a)$ follows from \eqref{BN6}
and
$(b)$ follows from the pinching inequality \cite[Lemma 9]{H-02}.
\end{proof}

\begin{lemma}\Label{LB1}
A permutation-invariant
state $\rho_n$ on ${\cal H}^{\otimes n}$ 
and an element $\eta \in {\cal R}_n[{\cal B}]$
satisfy the relation
\begin{align}
D_\alpha({\cal E}_{n,{\cal B}}(\rho_n)\| \rho_{n,{\cal B}}(\eta))
+\log \overline{d}_n
\ge D_\alpha(\rho_n\| \rho_{n,{\cal B}}(\eta))
\Label{B52T}
\end{align}
for $\alpha\in [0,1)$.
Also, any state $\sigma\in \mathcal{S}[\mathcal{B}]$ satisfies
\begin{align}
D_\alpha({\cal E}_{n,{\cal B}}(\rho_n)\| \sigma^{\otimes n})
+\log \overline{d}_n
\ge D_\alpha(\rho_n\| \sigma^{\otimes n}).
\Label{B52T3}
\end{align}
for $\alpha\in [0,1)$.
In particular, when $\sigma$ is invertible, the relation
\eqref{B52T3} holds with $\alpha=1$.
\end{lemma}

\begin{proof}
In this proof, we employ the pinching map
${\cal E}_{n,{\cal B},\rho_n}$ defined in \eqref{BD9}.
We have
\begin{align}
& e^{(\alpha-1)D_\alpha({\cal E}_{n,{\cal B}}(\rho_n)\| \rho_{n,{\cal B}}(\eta))}
\stackrel{(a)}{=}  e^{(\alpha-1)D_\alpha({\cal E}_{n,{\cal B},\rho_n}(\rho_n)\| \rho_{n,{\cal B}}(\eta))}\notag\\
=& \Tr 
\Big(\rho_{n,{\cal B}}(\eta)^{\frac{1-\alpha}{2\alpha}}
({\cal E}_{n,{\cal B},\rho_n}(\rho_n))
\rho_{n,{\cal B}}(\eta)^{\frac{1-\alpha}{2\alpha}}
\Big)^\alpha \notag\\
\stackrel{(b)}{=} & \Tr 
\Big(
{\cal E}_{n,{\cal B},\rho_n}
\big(\rho_{n,{\cal B}}(\eta)^{\frac{1-\alpha}{2\alpha}}
\rho_n
\rho_{n,{\cal B}}(\eta)^{\frac{1-\alpha}{2\alpha}}\big)
\Big)^\alpha \notag\\
=& \Tr 
\Big(
{\cal E}_{n,{\cal B},\rho_n}
\big(\rho_{n,{\cal B}}(\eta)^{\frac{1-\alpha}{2\alpha}}
\rho_n
\rho_{n,{\cal B}}(\eta)^{\frac{1-\alpha}{2\alpha}}\big)
\Big) \notag\\
&\cdot \Big(
{\cal E}_{n,{\cal B},\rho_n}
\big(\rho_{n,{\cal B}}(\eta)^{\frac{1-\alpha}{2\alpha}}
\rho_n
\rho_{n,{\cal B}}(\eta)^{\frac{1-\alpha}{2\alpha}}\big)
\Big)^{\alpha-1} \notag\\
=& \Tr 
\Big(
\rho_{n,{\cal B}}(\eta)^{\frac{1-\alpha}{2\alpha}}
\rho_n
\rho_{n,{\cal B}}(\eta)^{\frac{1-\alpha}{2\alpha}}
\Big) \notag\\
&\cdot \Big(
{\cal E}_{n,{\cal B},\rho_n}
\big(\rho_{n,{\cal B}}(\eta)^{\frac{1-\alpha}{2\alpha}}
\rho_n
\rho_{n,{\cal B}}(\eta)^{\frac{1-\alpha}{2\alpha}}\big)
\Big)^{\alpha-1} \notag\\
\stackrel{(c)}{\le} 
& \overline{d}_n^{1-\alpha}\Tr 
\Big(\rho_{n,{\cal B}}(\eta)^{\frac{1-\alpha}{2\alpha}}
\rho_n
\rho_{n,{\cal B}}(\eta)^{\frac{1-\alpha}{2\alpha}}
\Big) \notag\\
&\cdot 
\Big(\rho_{n,{\cal B}}(\eta)^{\frac{1-\alpha}{2\alpha}}
\rho_n
\rho_{n,{\cal B}}(\eta)^{\frac{1-\alpha}{2\alpha}}
\Big)^{\alpha-1} \notag\\
=& 
\overline{d}_n^{1-\alpha} e^{(\alpha-1)D_\alpha(
\rho_n\| \rho_{n,{\cal B}}(\eta))},
\end{align}
where
$(a)$ follows from $(a)$ of \eqref{NM6A},
$(b)$ follows from the relation 
${\cal E}_{n,{\cal B},\rho_n}(\rho_{n,{\cal B}}(\eta))
=\rho_{n,{\cal B}}(\eta)$,
and $(c)$ follows from \eqref{BR2T}.
Hence, we obtain \eqref{B52T}.

Since $\sigma^{\otimes n}$ is the sum of 
$\rho_{n,{\cal B}}(\eta)$, 
we can show \eqref{B52T3} in the same way.
In particular, when $\sigma$ is invertible, 
taking the limit $\alpha\to 1$,
we have the relation \eqref{B52T3} with $\alpha=1$.
\end{proof}

For $(p,\rho)\in \mathcal{R}[\mathcal{B}]$, we define the sets
\begin{align}
{\cal S}_{{\cal B}}(p)&:=\{\sigma\in {\cal S}({\cal H})|
p(\sigma)=p\} \Label{AU1} \\
{\cal S}_{{\cal B}}(p,\rho,\epsilon)&:=\{\sigma\in {\cal S}({\cal H})|
p(\sigma)=p,
D(\rho \| {\cal E}_{\cal B}(\sigma))
\le \epsilon \} \Label{AU2}\\
{\cal S}_{{\cal B}}^c(p,\rho,\epsilon)&:=\{\sigma\in {\cal S}({\cal H})|
p(\sigma)=p,
D(\rho \| {\cal E}_{\cal B}(\sigma))
> \epsilon \} .\Label{AU3}
\end{align}
Since 
${\cal S}_{{\cal B}}(p)$ is a homogeneous space for the special unitary group $\SU({\cal H})$, i.e.,
\begin{align*}
{\cal S}_{{\cal B}}(p)=
\{ U \sigma U^\dagger \}_{U \in \SU({\cal H})}
\end{align*}
for any element $\sigma \in {\cal S}_{{\cal B}}(p)$,
we consider the invariant probability measure $\mu$
on ${\cal S}_{{\cal B}}(p)$.

\begin{lemma}\Label{LB2}
For $(p,\rho)\in \mathcal{R}_n[\mathcal{B}]$,
we have
\begin{align}
& \frac{1}{\mu({\cal S}_{{\cal B}}(p,\rho,\epsilon))}
\int_{{\cal S}_{{\cal B}}(p,\rho,\epsilon)}
\sigma^{\otimes n} \mu (d \sigma)
\notag \\
\ge &c^n(p,\rho,\epsilon) 
\rho_{n,{\cal B}}(p,\rho),
\Label{BR8T}
\end{align}
where
\begin{align}
&c^n(p,\rho,\epsilon)\notag\\
:=&
\frac{d_{np,\rho}}{d_{np}\mu({\cal S}_{{\cal B}}(p,\rho,\epsilon))}
\Big((n+d)^{-\frac{(d+2)(d-1)}{2}} 
-d_{np} e^{-n \epsilon}
\Big) .\Label{XC1}
\end{align}
\end{lemma}

\begin{proof}
As shown in \cite[(6.23)]{Group1},
for $\sigma \in {\cal S}_{{\cal B}}(p)$, we have
\begin{align}
&(n+d)^{-\frac{(d+2)(d-1)}{2}} %e^{-n D(p_n \| p)}
\le \Tr (\sigma^{\otimes n} 
1_{\mathcal{U}_{np} \otimes \mathcal{V}_{np}})
 \Label{DF1}.
%\le & (n+1)^{\frac{d(d-1)}{2}} e^{-n D(p_n \| p)}. 
\end{align}
Hence,
\begin{align}
&(n+d)^{-\frac{(d+2)(d-1)}{2}} 
1_{\mathcal{U}_{np} \otimes \mathcal{V}_{np}}\notag\\
\le &
1_{\mathcal{U}_{np} \otimes \mathcal{V}_{np}}
\Tr \Big(
1_{\mathcal{U}_{np} \otimes \mathcal{V}_{np}}
\int_{{\cal S}_{{\cal B}}(p)}
\sigma^{\otimes n}
\mu (d \sigma) \Big)\notag\\
= &
d_{np} d_{np}' 
1_{\mathcal{U}_{np} \otimes \mathcal{V}_{np}}
 \Big(\int_{{\cal S}_{{\cal B}}(p)}
\sigma^{\otimes n}
\mu (d \sigma) \Big)
1_{\mathcal{U}_{np} \otimes \mathcal{V}_{np}}.
\Label{BR4TO}
\end{align}
Since
$1_{p,\rho,\mathcal{B}}^n = 
1_{\mathcal{U}_{np} \otimes \mathcal{V}_{np}} 1_{\rho,\mathcal{B}}^n$  
due to the definition \eqref{DU1}, we have  
\begin{align}
&(n+d)^{-\frac{(d+2)(d-1)}{2}} 
1_{p,\rho,\mathcal{B}}^n\notag\\
\le &
1_{p,\rho,\mathcal{B}}^n
\Tr \Big(
1_{\mathcal{U}_{np} \otimes \mathcal{V}_{np}}
\int_{{\cal S}_{{\cal B}}(p)}
\sigma^{\otimes n}
\mu (d \sigma) \Big)\notag\\
= &
d_{np} d_{np}' 
1_{p,\rho,\mathcal{B}}^n
 \Big(\int_{{\cal S}_{{\cal B}}(p)}
\sigma^{\otimes n}
\mu (d \sigma) \Big)
1_{p,\rho,\mathcal{B}}^n
\Label{BR4T}
\end{align}

Also, 
as shown in \cite[(6.5)]{Group1},
for $\sigma \in {\cal S}_{{\cal B}}^c(p,\rho,\epsilon)$,
we have
\begin{align}
\Tr (\sigma^{\otimes n} 
1_{p,\rho,{\cal B}}^n)
\le 
\Tr (\sigma^{\otimes n} 1_{\rho,{\cal B}}^n )
\le 
e^{-n D(\rho \| {\cal E}_{\cal B}(\sigma)  )}
\le e^{-n\epsilon}.
\Label{ZBHP}
\end{align}

Due to the permutation-invariance,
the operator $1_{p,\rho,{\cal B}}^n
\Big(\int_{{\cal S}^c(p,\rho,\epsilon)}
\sigma^{\otimes n}
\mu (d \sigma) \Big)1_{p,\rho,{\cal B}}^n$
has the form
$X\otimes 1_{{\cal V}_{np}}$,
with an positive semidefinite matrix $X$ on $\mathcal{U}_{np,\rho}$
satisfying
$\Tr X=
\frac{1}{d_{np}'}
\Tr (1_{p,\rho,{\cal B}}^n
\int_{{\cal S}^c(p,\rho,\epsilon)}
\sigma^{\otimes n}
\mu (d \sigma) )$
because 
$d_{np}'=\dim {\cal V}_{np}$.
Hence, we have
\begin{align}
d_{np}' X \le 
1_{\mathcal{U}_{np,\rho}}
\Tr (1_{p,\rho,{\cal B}}^n
\int_{{\cal S}^c(p,\rho,\epsilon)}
\sigma^{\otimes n}
\mu (d \sigma) ).\Label{ZLP7}
\end{align}
Thus, we have
\begin{align}
&d_{np}' 
1_{p,\rho,{\cal B}}^n
\Big(\int_{{\cal S}^c(p,\rho,\epsilon)}
\sigma^{\otimes n}
\mu (d \sigma) \Big)1_{p,\rho,{\cal B}}^n
\notag\\
=&d_{np}' X \otimes 1_{{\cal V}_{np}} \notag\\
\stackrel{(a)}{\le} &
1_{\mathcal{U}_{np,\rho}} \otimes 1_{{\cal V}_{np}} 
\Tr (1_{p,\rho,{\cal B}}^n
\int_{{\cal S}^c(p,\rho,\epsilon)}
\sigma^{\otimes n}
\mu (d \sigma) )
\notag\\
\stackrel{(b)}{=} &
1_{p,\rho,{\cal B}}^n \Tr (1_{p,\rho,{\cal B}}^n
\int_{{\cal S}^c(p,\rho,\epsilon)}
\sigma^{\otimes n}
\mu (d \sigma) )
\notag\\
= &
1_{p,\rho,{\cal B}}^n 
\int_{{\cal S}^c(p,\rho,\epsilon)}
\Tr (1_{p,\rho,{\cal B}}^n
\sigma^{\otimes n})
\mu (d \sigma) 
\notag\\
\stackrel{(c)}{\le} & 1_{p,\rho,{\cal B}}^n
\Big( \int_{{\cal S}^c(p,\rho,\epsilon)}
e^{-n \epsilon} 
\mu (d \sigma) \Big)
= 
1_{p,\rho,{\cal B}}^n
e^{-n \epsilon},
\Label{BR5T}
\end{align}
where 
$(a)$ follows from \eqref{ZLP7},
$(b)$ follows from \eqref{DU1V},
and $(c)$ follows from \eqref{ZBHP}.

Hence, we have
\begin{align}
&d_{np} d_{np}' 
1_{p,\rho,{\cal B}}^n
\Big( \int_{{\cal S}_{{\cal B}}(p,\rho,\epsilon)}
\sigma^{\otimes n}
\mu (d \sigma) \Big)1_{p,\rho,{\cal B}}^n
\notag\\
=& d_{np} d_{np}' 1_{p,\rho,{\cal B}}^n
\Big(\int_{{\cal S}_{{\cal B}}(p)}
\sigma^{\otimes n} \mu (d \sigma)\Big)1_{p,\rho,{\cal B}}^n
\notag\\
&- d_{np} d_{np}' 1_{p,\rho,{\cal B}}^n
\Big(\int_{{\cal S}_{{\cal B}}(p)\cap {\cal S}_{{\cal B}}(p,\rho,\epsilon)^c}
\sigma^{\otimes n} \mu (d \sigma)\Big)1_{p,\rho,{\cal B}}^n
\notag\\
\ge &
d_{np} d_{np}' 1_{p,\rho,{\cal B}}^n
\Big( \int_{{\cal S}_{{\cal B}}(p,\rho,\epsilon)}
\sigma^{\otimes n}
\mu (d \sigma) \Big)1_{p,\rho,{\cal B}}^n
\notag\\
&- d_{np} d_{np}' 1_{p,\rho,{\cal B}}^n
\Big(\int_{{\cal S}_{{\cal B}}(p,\rho,\epsilon)^c}
\sigma^{\otimes n} \mu (d \sigma)\Big)1_{p,\rho,{\cal B}}^n
\notag\\
\ge &
\Big((n+d)^{-\frac{(d+2)(d-1)}{2}} 
-d_{np} e^{-n \epsilon }
\Big)
1_{p,\rho,{\cal B}}^n.
\Label{BR6T}
\end{align}
where the final inequality follows from \eqref{BR4T} and \eqref{BR5T}.
Thus, 
since $\int_{{\cal S}_{{\cal B}}(p,\rho,\epsilon)}
\sigma^{\otimes n} \mu (d \sigma)$ is commutative with 
$1_{p,\rho,{\cal B}}^n$,
we have
\begin{align}
& \int_{{\cal S}_{{\cal B}}(p,\rho,\epsilon)}
\sigma^{\otimes n} \mu (d \sigma)
\notag\\
\ge &
1_{p,\rho,{\cal B}}^n
\Big( \int_{{\cal S}_{{\cal B}}(p,\rho,\epsilon)}
\sigma^{\otimes n}
\mu (d \sigma) \Big)1_{p,\rho,{\cal B}}^n
\notag\\
\ge &
\frac{1}{d_{np} d_{np}' }\Big((n+d)^{-\frac{(d+2)(d-1)}{2}} 
-d_{np} e^{-n \epsilon }
\Big)
1_{p,\rho,{\cal B}}^n
\notag\\
= &
\frac{d_{np,\rho}}{d_{np}}
\Big((n+d)^{-\frac{(d+2)(d-1)}{2}} 
-d_{np} e^{-n \epsilon}
\Big)
\rho_{n,{\cal B}}(p,\rho)
\Label{BR6T},
\end{align}
the final equation follows from \eqref{DV2}.
Then, dividing both sides of \eqref{BR6T} 
by $\mu({\cal S}_{{\cal B}}(p,\rho,\epsilon))$, we obtain \eqref{BR8T}.
\end{proof}

The notations introduced in this section are summarized 
as Tables \ref{notation-maps}, \ref{notation-sets}, and \ref{notation-q}.

\begin{table}[t]
\caption{Notations of maps}
\label{notation-maps}
\begin{center}
\begin{tabular}{|c|l|c|}
\hline
Symbol &Description &Eq. number    \\
\hline
${\cal E}_{{\cal B}}$ &
Pinching map based on basis $\mathcal{B}$ & Eq. \eqref{DX2} 
\\ \hline
$\overline{\cal E}_{\cal B}$ & Map & Eq. \eqref{NB3} \\
\hline
${\cal E}_{n,{\cal B}}$ & TP-CP map & Eq. \eqref{DV5}
\\ \hline
\multirow{2}{*}{${\cal E}_{n,{\cal B},\rho_n}$} & 
Pinching map based on 
& \multirow{2}{*}{Eq. \eqref{BD9}} \\
&permutation-invariant state $\rho_n$ &
\\ \hline
\end{tabular}
\end{center}
\end{table}

\begin{table}[t]
\caption{Notations of sets}
\label{notation-sets}
\begin{center}
\begin{tabular}{|c|l|c|}
\hline
\multirow{2}{*}{$\mathcal{S}[\mathcal{B}]$}
&Diagonal states for
&\multirow{2}{*}{Eq. \eqref{DX3B}} \\
&the basis $\mathcal{B}$ &
\\ \hline
\multirow{2}{*}{$\mathcal{S}_n[\mathcal{B}]$} &
Set of empirical states 
& \multirow{2}{*}{Eq. \eqref{DX3}} \\
& under the basis $\mathcal{B}$ & \\
\hline
\multirow{2}{*}{$Y_d^n$} & Set of Young indices & \\
&with length $n$ and depth $d$
& \\ \hline
\multirow{3}{*}{$\mathcal{P}_d$} &
Set of probability distributions 
&  \\
&$(p_j)_{j=1}^d$ with condition & \\
& $p_1\le p_2 \le \ldots \le p_d$ & 
 \\
\hline
\multirow{3}{*}{$\mathcal{P}_d^n$} &
Set of probability distributions 
&  \\
&$(p_j)_{j=1}^d\in \mathcal{P}_d^n$ with condition & \\
&
$p_j$ is an integer multiple of $1/n$& 
 \\
 \hline
$\mathcal{R}[\mathcal{B}]$ &
Subset of $\mathcal{P}_d\times \mathcal{S}[\mathcal{B}]$
 & Eq. \eqref{DV3}
\\ \hline
$\mathcal{R}_n[\mathcal{B}]$ & 
Subset of $\mathcal{P}_d^n \times \mathcal{S}_n[\mathcal{B}]$
& Eq. \eqref{DV4}
\\ \hline
$S_{\rho,r}$ & 
Subset of $\mathcal{R}[\mathcal{B}]$ & Eq. \eqref{DV7}
\\ \hline
${\cal S}_{{\cal B}}(p)$ & Subset of ${\cal S}({\cal H})$& Eq. \eqref{AU1} 
\\ \hline
${\cal S}_{{\cal B}}(p,\rho,\epsilon)$ &Subset of ${\cal S}_{{\cal B}}(p)$
& Eq. \eqref{AU2}
\\ \hline
${\cal S}_{{\cal B}}^c(p,\rho,\epsilon)$ &
Subset of ${\cal S}_{{\cal B}}(p)$ & Eq. \eqref{AU3}
\\ \hline
\end{tabular}
\end{center}
\end{table}

\begin{table}[t]
\caption{Other notations of quantum system}
\label{notation-q}
\begin{center}
\begin{tabular}{|c|l|c|}
\hline
Symbol &Description &Eq. number    \\
\hline
$e_{\mathcal{B}}(x^n)$ &
Empirical state 
$\sum_{j=1}^n \frac{1}{n}|v_{x_j}\rangle \langle v_{x_j}|$& 
 \\ \hline
\multirow{2}{*}{$p(\rho)$} &
Element of $\mathcal{P}_d$ composed &
\\ 
& of eigenvalues of $\rho$ &
\\ \hline
$p \succ p'$ &
Majorization relation
& Eq. \eqref{DX5}
\\ \hline
\multirow{2}{*}{$\rho_{n,{\cal B}}(p,\rho)$} &
Completely mixed state 
%$\frac{1}{d_{np,\rho} d_{np}'}1_{p,\rho,\mathcal{B}}^n$
& \multirow{2}{*}{Eq. \eqref{DV2}} \\
&on ${\cal U}_{np,\rho}\otimes {\cal V}_{np}$&
\\ \hline
$D_{\mathcal{B}}((p',\rho') \|\rho)$ & & Eq. \eqref{DV6}
\\ \hline
\multirow{2}{*}{$\mu$} & Invariant probability measure 
& 
\\ 
& on ${\cal S}_{{\cal B}}(p)$ &
\\ \hline
\end{tabular}
\end{center}
\end{table}

\section{Quantum generalization:
Proof of Theorem \ref{THC4}}\Label{S7}
The aim of this section is to show Theorem \ref{THC4}.
In this section, we use the following condition for 
a sequence of subsets ${\cal F}_n \subset {\cal S}({\cal H}^{\otimes n})$ instead of the condition (B2) in the beginning.
The following condition can be considered as a quantum version of the condition (D1).

\begin{description}
\item[(D2)]
For any element $\sigma \notin {\cal F}_1$ and $\epsilon>0$, 
there exists a sequence $\eta_n \in \mathcal{R}_n[\mathcal{B}]$ such that
$ \eta_n \to \overline{\cal E}_{\cal B}(\sigma)=( p(\sigma),{\cal E}_{\cal B}(\sigma)  )
\in \mathcal{R}[\mathcal{B}]$ and
\begin{align}
\liminf_{n\to \infty}
-\frac{1}{n}\log \beta_\epsilon
(\rho_{n,{\cal B}}(\eta_n)
\| {\cal E}_{n,{\cal B}}({\cal F}_n))> 0.\Label{NM9}
\end{align}
\end{description}

The structure of this section is illustrated by Fig. \ref{D1}.
Although the aim of this section is the derivation of 
\eqref{BNAT} from the conditions (C1), (C2), and (B2)
to show Theorem \ref{THC4},
we derive \eqref{BNAT} from the conditions (C1), (C2), and (D2) at the fisrt step.
Later, we derive the condition (D2) from 
the conditions (C1), (C2), and (B2).

\begin{figure}[tbhp]
\begin{center}
\includegraphics[scale=0.45]{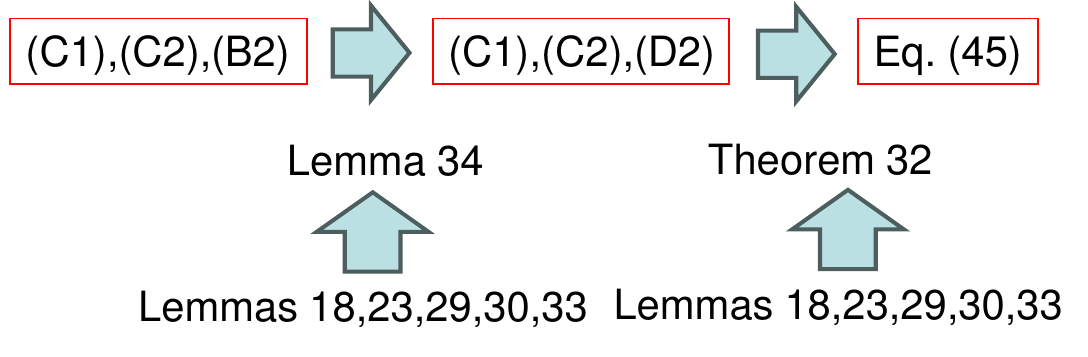}
\end{center}
\caption{Organization of Section \ref{S7}}
\Label{D2}
\end{figure}

As a quantum extension of Theorem \ref{THY},
we have the following theorem, which can be considered as a meta theorem proving Theorem \ref{THC4}. 
\begin{theorem}\Label{THYQ}
Assume that a sequence of subsets ${\cal F}_n \subset {\cal S}({\cal H}^{\otimes n})$
satisfies the conditions (C1) and (C2).
Also, 
given a subset ${\cal L} \subset 
\mathcal{R}[\mathcal{B}]$,
we assume that
any sequence $\eta_n \in \mathcal{R}_n[\mathcal{B}]$ with
$ \eta_n \to \eta \notin {\cal L}$ satisfies
\begin{align}
\liminf_{n\to \infty}
-\frac{1}{n}\log \beta_\epsilon
(\rho_{n,{\cal B}}(\eta_n)\| {\cal E}_{n,{\cal B}}({\cal F}_n))> 0.\Label{BNA}
\end{align}
Then, any state $\sigma\in \mathcal{S}[\mathcal{B}]$ satisfies
\begin{align}
&\liminf_{n\to \infty}-\frac{1}{n}\log \beta_\epsilon({\cal F}_n\|\sigma^{\otimes n})\notag \\
\ge &\liminf_{n\to \infty}-\frac{1}{n}\log \beta_\epsilon(
{\cal E}_{n,{\cal B}}({\cal F}_n)\|
\sigma^{\otimes n})
\ge 
D_{\mathcal{B}}({\cal L}\|\sigma).
\Label{NHAT}
\end{align}
\end{theorem}

\begin{proof}
{\bf Step 1:}
We choose an arbitrary compact set ${\cal G}$ included in
${\cal L}^c$.
%Due to Lemma \ref{LG7},
%we can restrict our states to permutation-invariant states ${\cal F}_{n,inv}$.
Any state $\rho_n \in {\cal E}_{n,{\cal B}}({\cal F}_{n})$
has the form 
$\sum_{\rho \in \mathcal{R}_n[\mathcal{B}]}c_\rho 
1^n_{\rho,\mathcal{B}}$ with 
$c_\rho \in [0,1]$.
Hence, we choose our test with the form
%without loss of generality, our test can be written as
$\Pi_n:= \sum_{ \eta \in {\cal G}\cap \mathcal{R}_n[\mathcal{B}] }
1_{\eta,\mathcal{B}}^n$.
Then, we show the following relation by contradiction;
\begin{align}
\max_{\rho_n \in {\cal E}_{n,{\cal B}}({\cal F}_n)} 
\Tr \Pi_n \rho_n \to 0.
\Label{NI8Q}
\end{align}

Assume that \eqref{NI8Q} does not hold.
%Due to Lemma \ref{LG7},
There exists a sequence of permutation-invariant states
$\rho_n \in {\cal E}_{n,{\cal B}}({\cal F}_{n})$ such that
\begin{align}
c:=\limsup_{n\to \infty} \Tr (\Pi_n \rho_n) >0.
\Label{NI2Q}
\end{align}
We choose a subsequence $\{n_k\}$ such that
\begin{align}
c=\lim_{k\to \infty} c_{n_k} ,\quad
c_n:= \Tr (\Pi_{n} \rho_{n}) .
\Label{NI3Q}
\end{align}
There exists an element $\eta_n \in {\cal G}\cap \mathcal{R}_n[\mathcal{B}]$
such that
\begin{align}
\Tr (1_{\eta_n,{\cal B}}^n \rho_{n} )
\ge \frac{c_n}{|{\cal G}\cap \mathcal{R}_n[\mathcal{B}]|}
\ge c_n (n+1)^{-(d-1)},
\Label{NI4Q}
\end{align}
where the final inequality follows
from the relation
$|{\cal G}\cap \mathcal{R}_n[\mathcal{B}]|
\le |\mathcal{R}_n[\mathcal{B}]| \le (n+1)^{(d-1)}$.
Since ${\cal G}$ is compact, 
there exists a subsequence $\{m_k\}$ of $\{n_k\}$ such that
$\eta_{m_k}$ converges to an element $\eta \in \mathcal{R}[\mathcal{B}]$.
%We apply the condition (D1) to the sequence $q_{m_k}$.

Now, we consider the discrimination between 
$\rho_{n,{\cal B}}(\eta_n)$
and ${\cal E}_{n,{\cal B}}({\cal F}_n)$,
where $\rho_{n,{\cal B}}(\eta_n)$
and any element of ${\cal E}_{n,{\cal B}}({\cal F}_n)$
have a form $\sum_{\eta \in \mathcal{R}_n[\mathcal{B}]} 
c_{\eta} 1^n_{\eta,{\cal B}}$.
To accept $\rho_{n,{\cal B}}(\eta_n)$ with probability $1-\epsilon$,
the test $\Pi_n$ needs to be $(1-\epsilon)1^n_{\eta_n,{\cal B}}$.

Then, we have
\begin{align}
&\lim_{k\to \infty}-\frac{1}{m_k}\log \Tr (\Pi_{m_k} \rho_{m_k})\notag\\
=&
\lim_{k\to \infty}-\frac{1}{m_k}(\log (\Tr (1^{m_k}_{q_{m_k},{\cal B}}\rho_{m_k}))
+\log (1-\epsilon)) \notag\\
\le &
\lim_{k\to \infty}-\frac{1}{m_k}(\log c_n (m_k+1)^{-2(d-1)}
+\log (1-\epsilon))
=0,
\end{align}
which contradicts with \eqref{BNA}.
Hence, we obtain \eqref{NI8Q}.

{\bf Step 2:}
We have
\begin{align}
\Tr (\sigma^{\otimes n} (I-\Pi_n))
=\sum_{\eta \in {\cal G}^c\cap \mathcal{R}_n[\mathcal{B}]} 
\Tr (1^n_{\eta,{\cal B}} \sigma^{\otimes n}).
\end{align}
Thus, the quantum Sanov theorem for quantum empirical distributions
(Proposition \ref{LL1})
implies
\begin{align}
&\liminf_{n\to \infty}\frac{-1}{n}\log 
\Tr (\sigma^{\otimes n} (I-\Pi_n)) \notag\\
=& \inf_{\eta \in {\cal G}^c}D_{\mathcal{B}}(\eta \|\sigma )
=D_{\mathcal{B}}({\cal G}^c \|\sigma ).\Label{MNT5Q}
\end{align}
The combination of \eqref{NI8Q} and \eqref{MNT5Q}
implies 
\begin{align}
\liminf_{n\to \infty}-\frac{1}{n}\log \beta_\epsilon
({\cal E}_{n,{\cal B}}({\cal F}_n)\|\sigma^{\otimes n})
\ge D_{\mathcal{B}}({\cal G}^c \|\sigma ).
\end{align}
Since ${\cal G}$ is an arbitrary compact set contained in ${\cal L}^c$, we have
\begin{align}
\liminf_{n\to \infty}-\frac{1}{n}\log \beta_\epsilon
({\cal E}_{n,{\cal B}}({\cal F}_n)\|\sigma^{\otimes n})
\ge D_{\mathcal{B}}({\cal L}\|\sigma).
\end{align}
\end{proof}

Using Theorem \ref{THYQ},
we have the following theorem.

\begin{theorem}\Label{THCQ}
When a sequence of subsets ${\cal F}_n \subset {\cal S}({\cal H}^{\otimes n})$
satisfies the conditions (C1), (C2), and (D2),
any state $\sigma\in \mathcal{S}[\mathcal{B}]$ satisfies
\eqref{BNAT}.
\end{theorem}

\begin{proof}
Applying Theorem \ref{THYQ} to the case when ${\cal L}$ is 
$\overline{\cal E}_{\cal B}({\cal F}_1)$,
we have
\begin{align}
&\liminf_{n\to \infty}-\frac{1}{n}\log \beta_\epsilon
({\cal E}_{n,{\cal B}}({\cal F}_n)\|
\sigma^{\otimes n})
\ge 
D_{\mathcal{B}}(\overline{\cal E}_{\cal B}({\cal F}_1)\|\sigma)
 \notag\\
=&
\lim_{n\to \infty}\frac{1}{n}D
({\cal E}_{n,{\cal B}}({\cal F}_1^{\otimes n})\|
\sigma^{\otimes n})\notag \\
\ge &
\limsup_{n\to \infty}\frac{1}{n}D
({\cal E}_{n,{\cal B}}({\cal F}_n)\|
\sigma^{\otimes n}).
\Label{NH1AQ}
\end{align}
Then, Theorem \ref{BFG8} with $\sigma=\sigma^{\otimes n}$ implies 
\begin{align}
&\lim_{n\to \infty}-\frac{1}{n}\log \beta_\epsilon
({\cal E}_{n,{\cal B}}({\cal F}_n)\|
\sigma^{\otimes n})
=D_{\mathcal{B}}(\overline{\cal E}_{\cal B}({\cal F}_1)\|\sigma)\notag\\
=&D({\cal F}_1\|\sigma)
=
\lim_{n\to \infty}\frac{1}{n}D
({\cal E}_{n,{\cal B}}({\cal F}_n)\|
\sigma^{\otimes n}).
\Label{NHUQ}
\end{align}
Also, we have
\begin{align}
&\liminf_{n\to \infty}-\frac{1}{n}\log 
\beta_\epsilon({\cal F}_n\|\sigma^{\otimes n}) \notag\\
\ge &\lim_{n\to \infty}-\frac{1}{n}\log \beta_\epsilon
({\cal E}_{n,{\cal B}}({\cal F}_n)\|
\sigma^{\otimes n}).\Label{NHUQ2}
\end{align}

Since
$\frac{1}{n}D({\cal F}_n\|\sigma^{\otimes n}) 
\ge
\frac{1}{n}D
({\cal E}_{n,{\cal B}}({\cal F}_n)\|
\sigma^{\otimes n})$, 
Lemma \ref{LB1} implies 
\begin{align}
\lim_{n\to \infty}\frac{1}{n}D({\cal F}_n\|\sigma^{\otimes n}) 
=
\lim_{n\to \infty}\frac{1}{n}D
({\cal E}_{n,{\cal B}}({\cal F}_n)\|
\sigma^{\otimes n}).\Label{NHUQ3}
\end{align}

Information processing inequality implies
\begin{align}
&D({\cal F}_n\|\sigma^{\otimes n}) \notag\\
\ge &
\min_{0\le x\le \epsilon}
\Big(x 
\big(\log x -\log \big(1-\beta_\epsilon({\cal F}_n\|\sigma^{\otimes n})\big)\big) \notag\\
&+(1-x)
\big(\log (1-x)-\log \beta_\epsilon({\cal F}_n\|\sigma^{\otimes n})
\big)\Big)\notag \\
= &
\epsilon \big(\log \epsilon -\log \big(1-\beta_\epsilon({\cal F}_n\|\sigma^{\otimes n})\big)\big) \notag\\
&+(1-\epsilon)
\big(\log (1-\epsilon)-\log \beta_\epsilon({\cal F}_n\|\sigma^{\otimes n})\big) \notag \\
= &
-h(\epsilon) 
-\epsilon \log (1-\beta_\epsilon({\cal F}_n\|\sigma^{\otimes n})) \notag\\
&-(1-\epsilon)\log \beta_\epsilon({\cal F}_n\|\sigma^{\otimes n}) \notag \\
\ge &
-h(\epsilon) 
-(1-\epsilon)\log \beta_\epsilon({\cal F}_n\|\sigma^{\otimes n}), \end{align}
where $h(x):=-x\log x -(1-x)\log(1-x)$.
Hence, we have
\begin{align}
\frac{-1}{n}
\log \beta_\epsilon({\cal F}_n\|\sigma^{\otimes n}))
\le \frac{D({\cal F}_n\|\sigma^{\otimes n}) +h(\epsilon) 
}{(1-\epsilon)n} .
\end{align}
Taking the limit, we have
\begin{align}
\limsup_{n\to \infty}\frac{-1}{n}
\log \beta_\epsilon({\cal F}_n\|\sigma^{\otimes n}))
\le 
\limsup_{n\to \infty}
\frac{1}{n} D({\cal F}_n\|\sigma^{\otimes n})\Label{NHUQ4}.
\end{align}
The combination of \eqref{NHUQ}, \eqref{NHUQ2}, \eqref{NHUQ3},
and \eqref{NHUQ4}
implies \eqref{BNAT}.
\end{proof}

To check the condition (D2), the following lemma is useful
as a quantum extension of Lemma \ref{LB89T}.

\begin{lemma}\Label{LB89}
Assume that a sequence of subsets ${\cal F}_n \subset {\cal S}({\cal H}^{\otimes n})$
satisfies the conditions (C1) and (C2).
Then, any state 
$\eta_n \in \mathcal{R}_n[\mathcal{B}]$
satisfies
\begin{align}
&-\frac{1}{n}\log \beta_\epsilon
(\rho_{n,{\cal B}}(\eta_n)
\| {\cal E}_{n,{\cal B}}({\cal F}_n)) \notag\\
\ge &
\frac{(1-\alpha)}{n}
D_\alpha({\cal E}_{n,{\cal B}}({\cal F}_{n})\| \rho_{n,{\cal B}}(\eta_n))
\Label{FG3}
\end{align}
for $\alpha \in [0,1)$.
\end{lemma}

In fact, when the limit of RHS of \eqref{FG3} is strictly positive,
the condition (D2) holds.

\begin{proof}
To accept $\rho_{n,{\cal B}}(\eta_n)$ with probability $1-\epsilon$,
the test to support $\rho_{n,{\cal B}}(\eta_n)$ needs to be $(1-\epsilon)1_{\eta_n,{\cal B}}^n$.
Since $\Tr \rho_{n,{\cal B}}(\eta_n)1_{\eta_n,{\cal B}}^n=1$,
we have 
$\rho_{n,{\cal B}}(\eta_n)1_{\eta_n,{\cal B}}^n \ge 
{\cal E}_{n,{\cal B}}(\rho_n) 1_{\eta_n,{\cal B}}^n $.
Thus, we have
\begin{align}
& \beta_\epsilon
(\rho_{n,{\cal B}}(\eta_n)
\| {\cal E}_{n,{\cal B}}({\cal F}_n)) \notag\\
= &
\max_{\rho_n \in {\cal F}_{n}}
\Tr ({\cal E}_{n,{\cal B}}(\rho_n) (1-\epsilon)1_{\eta_n,{\cal B}}^n) \notag\\
\le &
(1-\epsilon) \max_{\rho_n \in {\cal F}_{n}}
\Tr ({\cal E}_{n,{\cal B}}(\rho_n)^{\alpha} 
\rho_{n,{\cal B}}(\eta_n)^{1-\alpha}) \notag\\
=&(1-\epsilon) e^{-(1-\alpha)D_\alpha({\cal E}_{n,{\cal B}}({\cal F}_{n})\| \rho_{n,{\cal B}}(\eta_n))}\notag \\
\le &
e^{-(1-\alpha)D_\alpha({\cal E}_{n,{\cal B}}({\cal F}_{n})\| \rho_{n,{\cal B}}(\eta_n))},
\end{align}
which implies \eqref{FG3}.
\end{proof}

Combining Lemma \ref{LB89} with Lemmas 
\ref{NK8}, \ref{Pr6-5}, \ref{LB1}, and \ref{LB2}, 
we obtain the following lemma.

\begin{lemma}\Label{BF28}
Assume that a sequence of subsets ${\cal F}_n \subset {\cal S}({\cal H}^{\otimes n})$
satisfies the conditions (C1), (C2), and (B2).
Then, the condition (D2) holds.
\end{lemma}

Therefore, the combination of Theorem \ref{THCQ} and Lemma \ref{BF28}
implies Theorem \ref{THC4}.

\begin{proofof}{Lemma \ref{BF28}}
For any element $\mh{\sigma_0} \notin {\cal F}_1$ and $\epsilon>0$, 
there exists a sequence $\eta_n \in \mathcal{R}_n[\mathcal{B}]$ such that
$ \eta_n \to \overline{\cal E}_{\cal B}(\mh{\sigma_0})$.
We choose a POVM $M$ to satisfy the condition (B2).
We denote the number of measurement outcomes of 
$M$ by $d_M$.
\mh{Since $M$ is tomographically complete,} 
we can choose a sufficiently small number $\epsilon>0$ such that
\mh{\begin{align}
&\lim_{n\to \infty}\min_{ \sigma\in {\cal S}_{{\cal B}}(\eta_n,\epsilon)} D_\alpha(
M({\cal F}_1)\|  M(\sigma)) \notag \\
=&
\min_{ \sigma\in {\cal S}_{{\cal B}}(\sigma_0,\epsilon)} D_\alpha(
M({\cal F}_1)\|  M(\sigma))
>0.
\Label{GF1}
\end{align}}

Any element $\rho_n \in {\cal F}_{n}$
satisfies
\begin{align}
& D_\alpha({\cal E}_{n,{\cal B}}(\rho_n)\| \rho_{n,{\cal B}}(\eta_n))
+\log \overline{d}_n \notag\\
\stackrel{(a)}{\ge}& 
D_\alpha(\rho_n\| \rho_{n,{\cal B}}(\eta_n)) 
\stackrel{(b)}{\ge}
D_\alpha(M^{\otimes n}(\rho_n)\| M^{\otimes n}(\rho_{n,{\cal B}}(\eta_n))) \notag\\
\stackrel{(c)}{\ge}& 
D_\alpha\Big(M^{\otimes n}(\rho_n) \Big\| \frac{1}{\mu({\cal S}_{{\cal B}}(p,\rho,\epsilon))}
\int_{{\cal S}_{{\cal B}}(\eta_n,\epsilon)}
M(\sigma)^{\otimes n} \mu (d \sigma)\Big) \notag\\
&+\log c^n(\eta_n,\epsilon)\Label{GF2}
\end{align}
for $\alpha\in (0,1)$,
where $(a)$ follows from Lemma \ref{LB1},
$(b)$ follows from the information processing inequality,
and
$(c)$ follows from Lemmas \ref{LB2} and \ref{NK8}.

Thus, we have
\begin{align}
& D_\alpha({\cal E}_{n,{\cal B}}({\cal F}_{n})\| \rho_{n,{\cal B}}(\eta_n))
+\log \overline{d}_n \notag\\
\stackrel{(a)}{\ge}& 
D_\alpha\Big(M^{\otimes n}({\cal F}_{n}) 
\Big\| \frac{1}{\mu({\cal S}_{{\cal B}}(\eta_n,\epsilon))}
\int_{{\cal S}_{{\cal B}}(\mh{\eta_n},\epsilon)}
M(\sigma)^{\otimes n} \mu (d \sigma)\Big) \notag\\
&+\log c^n(\eta_n,\epsilon)\notag\\
\stackrel{(b)}{\ge} &
n \min_{ \sigma\in {\cal S}_{{\cal B}}( \mh{\eta_n},\epsilon)} D_\alpha(
M({\cal F}_1)\|  M(\sigma))
-\frac{(d_M-1)}{1-\alpha}\log (n+1) \notag\\
&+\log c^n(\eta_n,\epsilon),\Label{GF3}
\end{align}
where 
$(a)$ follows from \eqref{GF2}.
Since the condition (B2) guarantees that 
$M^{\otimes n}({\cal F}_{n})$ satisfies the condition (B1), i.e.,
Lemma \ref{Pr6-5} can be applied to $M^{\otimes n}({\cal F}_{n})$,
Lemma \ref{Pr6-5} guarantees the step $(b)$.
Using \eqref{XC1}, \eqref{GF1}, and \eqref{GF3}, we have
\mh{
\begin{align}
& \liminf_{n\to \infty}\frac{1}{n}D_\alpha({\cal E}_{n,{\cal B}}({\cal F}_{n})\| \rho_{n,{\cal B}}(\eta_n)) \notag\\
\ge & 
\lim_{n\to \infty}\min_{ \sigma\in {\cal S}_{{\cal B}}(\eta_n,\epsilon)} D_\alpha(
M({\cal F}_1)\|  M(\sigma)) \notag \\
=&
\min_{ \sigma\in {\cal S}_{{\cal B}}(\sigma_0,\epsilon)} D_\alpha(
M({\cal F}_1)\|  M(\sigma))
>0 .
\Label{XJI}
\end{align}}
The combination of Lemma \ref{LB89} and \eqref{XJI} implies \eqref{NM9}. Hence, we obtain the condition (D2).
\end{proofof}

\section{Conclusion}\Label{S12}
We have presented a generalization of the quantum Sanov theorem  
for hypothesis testing under simpler assumptions  
than those used in \cite{LBR,FFF},
and have shown that 
our condition is weaker than that used in \cite{LBR}.
The measurement used here is the same as
the measurement used in the old papers 
\cite{H-01,H-02}.
This result shows that the minimum value $D({\cal F}_1\|\sigma)$ works as a general detectability measure.
Section \ref{S2B} has introduced both the classical and quantum formulations,  
stated in Theorems \ref{THC3} and \ref{THC4}.  
While the method in \cite{LBR}  
addressed the case where ${\cal F}_n$ is the set of separable states,  
it do not demonstrate applicability to the cases  
where ${\cal F}_n$ is the set of positive partial transpose (PPT) states  
or the convex hull of stabilizer states.  
In addition, the paper \cite{FFF} clarifies that the approach by \cite{FFF} 
does not cover the PPT state case.  
In contrast, our method covers all three cases: separable states, PPT states, and the convex hull of stabilizer states, demonstrating  
the broader applicability and generality of our approach.  
The method in \cite{LBR} utilizes the blurring map,  
while \cite{FFF} employs polar sets.  
Our central technique, however, is the quantum empirical distribution,  
which depends on the choice of basis.  
This method was originally introduced  
to establish a quantum analog of the Sanov theorem  
using empirical distributions,  
as the original classical Sanov theorem describes the asymptotic behavior  
of empirical distributions.  

In this paper,  
we first have developed two classical extensions of the quantum Sanov theorem  
for hypothesis testing, detailed in Sections \ref{S4} and \ref{S5}.  
Section \ref{S4} has presented a classical generalization  
(Theorem \ref{THY}) using empirical distributions,  
while Section \ref{S5} has provided another extension (Theorem \ref{Pr7}).  
Both results, which include non-single-letterized forms,  
can be regarded as meta-theorems.  
Lemma \ref{Pr6} is essential for deriving Theorem \ref{THC3},  
and both approaches are necessary to obtain our main result, Theorem \ref{THC4}.  
Next, in Subsection \ref{S6C}, we have derived several useful formulas  
for quantum empirical distributions.  
In Section \ref{S7},
we have established a quantum generalization (Theorem \ref{THYQ}),  
which also qualifies as a meta-theorem due to its non-single-letterized structure.  
Combining Theorem \ref{THYQ}, the formulas from Subsection \ref{S6C},  
and the quantum Sanov theorem for quantum empirical distributions by the paper \cite{another},  
we have derived our final result, Theorem \ref{THC4}.  

This paper did not discuss the infinite-dimensional case.
In the multi-mode Bosonic system, we can consider various examples for the set of the resource-free states, e.g.,
the set of Gaussian mixture of coherent states,
the convex hull of the set of coherent states,
the set of Gaussian states,
and the set of non-negative Wigner function states, etc
\cite{PRXQuantum-Non-Gaussian}.
These examples satisfy the conditions (A2)--(A5) and the convexity and closedness conditions. 
Also, to satisfy the condition (B3) for these examples,
we can consider using the heterodyne measurement, i.e., the POVM given by the resolution of the identity by the coherent states.
However, it is needed to extend the obtained result to the 
multi-mode Bosonic system to apply it to these examples.
This is an interesting future study.

\section*{Acknowledgement}
The author is thankful to Professor Hayata Yamasaki for helpful discussions.
Also,
the author is thankful to Mr. Zhiwen Lin for helpful comments.
In particular, he pointed out a technical problem 
in Lemma \ref{Pr6-5}, which was resolved in the current version.

\section*{Ethics approval and consent to participate}
Not applicable.

\section*{Consent for publication}
Not applicable.

\section*{Availability of supporting data}
Data sharing is not applicable to this article as no datasets were generated
or analyzed during the current study.

\section*{Competing interests}
There are no competing interests.

\section*{Funding}
The author was supported in part by the National Natural Science Foundation of China No. 62171212.
and the General R \& D Projects of 1+1+1 
CUHK-CUHK(SZ)-GDST Joint Collaboration Fund (Grant No.
GRDP2025-022).

\appendix 
\section{Relation between Conditions (A6) and (D1)}\Label{S8}
\subsection{Equivalence relation between (A6), (D1), and additional condition}
This appendix studies the relation between Conditions (A6) and (D1) under the classical setting.
Then, we prepare the following lemmas.
\begin{lemma}\Label{Pr4}
Assume that a sequence of subsets ${\cal F}_n 
\subset {\cal S}_c({\cal H}^{\otimes n})$
satisfies the condition (A5).
Any sequence $p_n \in {\cal T}_n$ %with $ p_n \to p$ 
satisfies
\begin{align}
\lim_{n\to \infty}\frac{1}{n}
\Big|
D(\rho_n(p_n)\|{\cal F}_n)
+\log \beta_\epsilon(
\rho_n(p_n)\|{\cal F}_n)\Big|
=0. \Label{BHY2}
\end{align}
\end{lemma}

\begin{proofof}{Lemma \ref{Pr4}}
To calculate $\beta_\epsilon(
\rho_n(p_n)\|{\cal F}_n)$
and $D(\rho_n(p_n)\|{\cal F}_n)$,
due to Lemma \ref{LG7},
we discuss $\beta_\epsilon(
\rho_n(p_n)\|{\cal F}_{n,inv})$
and $D(\rho_n(p_n)\|{\cal F}_{n,inv})$.
Since any element of ${\cal F}_{n,inv}$ and 
$\rho_n(p_n)$ are permutation-invariant,
without loss of generality,
the test $\Pi_n$ to support 
$\rho_n(p_n)$ has the form 
$c1^n_{p_n}$ because 
a permutation-invariant element 
in ${\cal S}_c({\cal H}^{\otimes n})$
takes a constant value on ${\cal T}_{n,p}$.
Hence, we have
\begin{align}
D(\rho_n(p_n)\|\sigma)
=- \log \Tr (1^n_{p_n}\sigma).\Label{NK1}
\end{align}
Thus, we have
\begin{align}
-\log\beta_\epsilon(\rho_n(p_n)\|\sigma)
=- \log \Tr (1^n_{p_n}\sigma) -\log\epsilon. 
\Label{NK2}
\end{align}
The combination of \eqref{NK1} and \eqref{NK2}
implies \eqref{BHY2}.
\end{proofof}

\begin{lemma}\Label{LJ1}
Assume that a sequence of subsets ${\cal F}_n 
\subset {\cal S}_c({\cal H}^{\otimes n})$
satisfies the conditions 
(A1), (A2), and (A4).
Any sequence $p_n \in {\cal T}_n$ with $ p_n \to p$ 
satisfies
\begin{align}
\frac{1}{m} D( p^{\otimes m}\| {\cal F}_{m})
%\lim_{n\to \infty}\frac{1}{n}D(
%p^{\otimes n}\|{\cal F}_n)
\ge 
\limsup_{n\to \infty}\frac{1}{n}D(
\rho_n(p_n)\|{\cal F}_n).
\Label{BHYY}
\end{align}
\end{lemma}

\begin{lemma}\Label{LJ2}
Assume that a sequence of subsets ${\cal F}_n 
\subset {\cal S}_c({\cal H}^{\otimes n})$
satisfies the conditions 
(A1), (A2), and (A4).
Any sequence $p_n \in {\cal T}_n$ with $ p_n \to p$ satisfies
\begin{align}
\lim_{n\to \infty}
\frac{1}{n} D(p^{\otimes n}\|{\cal F}_n) 
\le 
\liminf_{n\to \infty}\frac{1}{n}D(
\rho_n(p_n)\|{\cal F}_n).
\Label{BHYC}
\end{align}
\end{lemma}

Combining Lemmas \ref{Pr4}, \ref{LJ1}, and \ref{LJ2}, 
we obtain the following theorem.

\begin{theorem}\Label{Pr5}
Assume that a sequence of subsets ${\cal F}_n 
\subset {\cal S}_c({\cal H}^{\otimes n})$
satisfies the conditions (A1), (A2), (A4), and (A5).
Then, the conditions (A6), (D1), and (B6) are equivalent.
\end{theorem}

\begin{proof}
Since (A5) holds,
Lemma \ref{Pr4} enables us to replace 
$-\log \beta_\epsilon(\rho_n(p_n)\|{\cal F}_n)$
by $D(\rho_n(p_n)\|{\cal F}_n)$
in the condition (D1).
In the following, we handle the condition (D1)
by using $D(\rho_n(p_n)\|{\cal F}_n)$ instead of 
$-\log \beta_\epsilon(\rho_n(p_n)\|{\cal F}_n)$.

Lemma \ref{LJ1} guarantees that the condition (D1) implies the condition (A6).
Lemma \ref{LJ2} guarantees that the condition (A6) implies the condition (D1).
\end{proof}

\subsection{Proof of Lemma \ref{LJ1}}
\noindent{\bf Step 1:}
Choice of $m,\epsilon_0$.

We fix an arbitrarily small real number
$\epsilon_0>0$
and an arbitrary positive integer $m$.
We choose $\sigma_m:=\argmin_{\sigma\in {\cal F}_m}
D(p^{\otimes m}\|\sigma)$ 
and $\sigma_{m,\epsilon_0}:=(1-\epsilon_0)\sigma_m
+\epsilon_0 \sigma_0^{\otimes m}$.
Then, we have
\begin{align}
\frac{-1}{m}\Tr (p^{\otimes m}\log \sigma_{m,\epsilon_0})
&\le 
\frac{-1}{m} \Tr (p^{\otimes m}\log \sigma_{m})
+\frac{-1}{m}\log (1-\epsilon_0)\Label{BM5} \\
\frac{-1}{m}\Tr (p^{\otimes m}\log \sigma_{m,\epsilon_0})
&\le 
- \log \sigma_{0}
+\frac{-1}{m}\log \epsilon_0\Label{BM5+} .
\end{align}

\noindent{\bf Step 2:}
Choice of $k,l,\epsilon_1,\epsilon_2,\epsilon_3$.

For an arbitrary positive integer $n$, we choose 
two positive integers $k$ and $l$ such that
\begin{align}
n=mk+l \hbox{ with }
l \le m.
\end{align}
We define the subset 
${\cal A}_{k,l}\subset {\cal T}_{mk}$
composed of 
the empirical distributions 
the first $mk$ components of 
$x\in {\cal T}_{mk+l,p_{mk+l}}$, i.e.,
\begin{align}
{\cal A}_{k,l}:=
\{ (x_1, \ldots, x_{mk})\}_{x\in {\cal T}_{mk+l,p_{mk+l}}}\subset{\cal T}_{mk}.
\end{align}

For any element $p' \in {\cal A}_{k,l}$,
we choose an element $q_{k,l}(p')\in {\cal T}_{l} $
such that
\begin{align}
mk p'+ l q_{k,l}(p')=(mk+l)p_{mk+l},\Label{NM1}
\end{align}
i.e., 
\begin{align}
1^{mk}_{p'}\otimes 1^l_{q_{k,l}(p')}
\le 1^{mk+l}_{p_{mk+l}}.
\end{align}
That is, we have
\begin{align}
\sum_{p'\in {\cal A}_{k,l}}1^{mk}_{p'}\otimes 1^l_{q_{k,l}(p')}
= 1^{mk+l}_{p_{mk+l}}.
\end{align}
There exists a distribution $Q(p')$ on $ {\cal A}_{k,l}$ such that
\begin{align}
\sum_{p'\in {\cal A}_{k,l}}
Q(p')
\rho_{mk}(p')\otimes \rho_l(q_{k,l}(p'))
= \rho_{mk+l}(p_{mk+l}).\Label{NMT}
\end{align}
Also, \eqref{NM1} implies
\begin{align}
mk\| p' - p_{mk+l} \|_1 \le 2l,
\end{align}
where $\|~\|_1$ expresses the $l_1$-norm.
Hence,
\begin{align}
\| p' - p \|_1 \le \frac{2l}{mk}+\frac{1}{mk+l}
\le \frac{2l+1}{mk}
\le \frac{2m+1}{mk}.
\Label{NMT5}
\end{align}

We choose arbitrary small real numbers 
$\epsilon_1,\epsilon_2,\epsilon_3>0$.
We consider the system
${\cal H}^{\otimes m k}$.
We define the typical subspace $U_{k, p^{\otimes m},\epsilon_1}$
of $p^{\otimes m}$ as
\begin{align}
U_{k, p^{\otimes m},\epsilon_1}:=
\{ x \in ({\cal X}^m)^{k}|
\| e_k(x)- p^{\otimes m}\|_1\le \epsilon_1\}.
\end{align}
Here, $e_k(x)$ is the empirical distribution of $k$ data $x$ 
on ${\cal X}^m$.
Due to \eqref{NMT5}, we can choose sufficiently large $k$ such that
\begin{align}
\Tr (\rho_{mk}(p_{p'}) 1_{U_{k, p^{\otimes m},\epsilon_1}})
&\ge 1- \epsilon_2 \Label{BM1}\\
\Big\|
P(U_{k, p^{\otimes m},\epsilon_1})
\Big(\frac{1}{k} \log \sigma_{m,\epsilon_0}^{\otimes k}
- \Tr (p^{\otimes m}\log \sigma_{m,\epsilon_0}) \Big)\Big\| &\le \epsilon_3\Label{BM2}
\end{align}
for $p' \in {\cal A}_{k,l}$, where
$\|~\|$ expresses the matrix norm,
and $P(U_{k, p^{\otimes m},\epsilon_1})$
is the projection to the projection $U_{k, p^{\otimes m},\epsilon_1}$.

\noindent{\bf Step 3:}
Evaluation with large $k$.

Using \eqref{BM2}, we have
\begin{align}
& 
\frac{-1}{k} \Tr (\rho_{mk}(p') 
P(U_{k, p^{\otimes m},\epsilon_1})
\log \sigma_{m,\epsilon_0}^{\otimes k})
\notag\\
\le &
\Tr (\rho_{mk}(p') P(U_{k, p^{\otimes m},\epsilon_1}))
\cdot -\Tr (p^{\otimes m}\log \sigma_{m,\epsilon_0})
+ \epsilon_3 
\notag\\
\le &
-\Tr ( p^{\otimes m}\log \sigma_{m,\epsilon_0})
+ \epsilon_3 .\Label{BM3}
\end{align}
Using \eqref{BM5+} and \eqref{BM1}, we have
\begin{align}
& \frac{-1}{mk} \Tr (\rho_{mk}(p') 
(I-P(U_{k, p^{\otimes m},\epsilon_1}))
\log \sigma_m^{\otimes k}) \notag\\
\le & \epsilon_2 (\frac{-1}{mk}\log \epsilon_0 
+ \log r_0).
\Label{BM4}
\end{align}
Combining \eqref{BM3} and \eqref{BM4}, we have
\begin{align}
& \frac{-1}{mk} \Tr (\rho_{mk}(p')  \log \sigma_{m,\epsilon_0}^{\otimes k} )\notag\\
\le & 
\frac{-1}{m} \Tr (p^{\otimes m}\log \sigma_{m})
+\frac{-1}{m}\log (1-\epsilon_0) \notag\\
&+\epsilon_3+
\epsilon_2 (\frac{-1}{mk}\log \epsilon_0 
+ \log r_0).
\Label{BM6}
\end{align}
Since 
$\log \sigma_0^{\otimes l} \le l\log r_0$,
we have
\begin{align}
& \frac{-1}{mk+l} \Tr 
\Big(\rho_{mk}(p') \otimes \rho_l(q_{k,l}(p')) 
\log \big(\sigma_{m,\epsilon_0}^{\otimes k} \otimes
\sigma_0^{\otimes l}\big)\Big)\notag\\
\le & 
\frac{mk}{mk+l}\Big(
\frac{-1}{m} (\Tr p^{\otimes m}\log \sigma_{m})
+\frac{-1}{m}\log (1-\epsilon_0) \notag\\
&+\epsilon_3+
\epsilon_2 (\frac{-1}{mk}\log \epsilon_0 
+ \log r_0) \Big) 
+\frac{l}{mk+l}\log r_0.
\Label{BM6H}
\end{align}
Combining \eqref{NMT} and \eqref{BM6H}, we have
\begin{align}
& \frac{-1}{mk+l} \Tr 
\Big(\rho_{mk+l}(p_{mk+l}) 
\log \big(\sigma_{m,\epsilon_0}^{\otimes k} \otimes
\sigma_0^{\otimes l}\big)\Big)\notag\\
\le & 
\frac{mk}{mk+l}\Big(
\frac{-1}{m} (\Tr p^{\otimes m}\log \sigma_{m})
+\frac{-1}{m}\log (1-\epsilon_0) \notag\\
&+\epsilon_3+
\epsilon_2 (\frac{-1}{mk}\log \epsilon_0 
+ \log r_0) \Big) 
+\frac{l}{mk+l}\log r_0.
\Label{BM6Y}
\end{align}

\noindent{\bf Step 4:}
Limit $k\to \infty$.

Taking the limit $k\to \infty$, we have
\begin{align}
&\limsup_{n\to \infty} \min_{\sigma \in {\cal F}_n}
\frac{-1}{n} \Tr (\rho_n(p_n)\log \sigma) \notag\\
\le & \limsup_{k\to \infty}
\frac{-1}{mk+l} \Tr 
\Big(\rho_{mk+l}(p_{mk+l}) 
\log \big(\sigma_{m,\epsilon_0}^{\otimes k} \otimes
\sigma_0^{\otimes l}\big)\Big)\notag\\
\le & 
\frac{-1}{m} \Tr (p^{\otimes m}\log \sigma_{m})
+\frac{-1}{m}\log (1-\epsilon_0) \notag\\
&+\epsilon_3 + \epsilon_2 \log r_0.\Label{BM7}
\end{align}
Taking the limit $\epsilon_0,\epsilon_2,\epsilon_3\to 0$, we have
\begin{align}
&\limsup_{n\to \infty} \min_{\sigma \in {\cal F}_n}
\frac{-1}{n} \Tr (\rho_n(p_n)\log \sigma) \notag\\
\le & 
\frac{-1}{m} \Tr (p^{\otimes m}\log \sigma_{m}).
\Label{BM8}
\end{align}
Since $\frac{1}{n}S(\rho_{n}(p_{n}))\to S(p)$, we have
\begin{align}
&\limsup_{n\to \infty}\frac{1}{n}D(
\rho_n(p_n)\|{\cal F}_n) 
\le 
\frac{1}{m} D( p^{\otimes m}\| \sigma_{m}),
\Label{BM9}
\end{align}
which implies \eqref{BHYY}.

\subsection{Proof of Lemma \ref{LJ2}}
\noindent{\bf Step 1:}
Choice of $m,\epsilon_0$.

We fix an arbitrary small real number $\epsilon_0>0$.
We choose $\sigma_m:=\argmin_{\sigma\in {\cal F}_m}
D(\rho_m(p_m)\|\sigma)$
and $\sigma_{m,\epsilon_0}:=(1-\epsilon_0)\sigma_m
+\epsilon_0 \sigma_0^{\otimes m}$.
We have
\begin{align}
\Tr (1^{m}_{p_m} \rho_m(p_m))&=1 \notag\\
\Tr (1^{m}_{p_m} \sigma_m)&=
e^{-D(\rho_m(p_m)\|\sigma_m)} \notag\\
\Tr (1^{m}_{p_m} \sigma_{m,\epsilon_0})&\ge
(1-\epsilon_0) e^{-D(\rho_m(p_m)\|\sigma_m)} 
\Label{BC1}.
\end{align}
Also,
\begin{align}
\Tr (1^{k_2 m}_{p''} \sigma_{m,\epsilon_0}^{\otimes k_2})
\ge \epsilon_0^{k_2} e^{-k_2 m r_0}
\Label{BC2}.
\end{align}

%We define $a_m:= \|1_{m}(p_m) \sigma_m\|$.
We choose an arbitrarily small rational number $\epsilon_1,\epsilon_2>0 $.
There exists a sufficiently large number $m_0$ such that
$ \|p_m - p\|_1\le \epsilon_0$ 
for $m \ge m_0$.
We assume that $k$ is a positive integer such that
$k_2:= k (\epsilon_1+\epsilon_2)$ is an integer.
We choose positive integers $k_1$ as $k_1:=k-k_2$.

For $p'\in {\cal T}_{km}$ satisfying $\|p'-p\|_1\le \epsilon_2$,
we have
$\|p_m-p'\|_1\le \epsilon_1+\epsilon_2$.
Then, we choose $p''$ such that
$k_1 m p_m+ k_2 m p''= k m p'$.
Then, we have
\begin{align}
(1^{m}_{p_m})^{\otimes k_1}\otimes 1^{k_2 m}_{p''} \le
1^{k m}_{p'}.
\Label{BC3}
\end{align}
Using \eqref{BC1} and \eqref{BC2}, we have
\begin{align}
& \Tr ( (1^{m}_{p_m})^{\otimes k_1}\otimes 1^{k_2 m}_{p''} \sigma_{m,\epsilon_0}^{\otimes k}) \notag\\
\ge & (1-\epsilon_0)^{k_1} e^{-k_1 D(\rho_m(p_m)\|\sigma_m)} 
\epsilon_0^{k_2} e^{-k_2 m r_0}.
\Label{BC4}
\end{align}

\noindent{\bf Step 2:}
Averaged state with respect to random permutation.

We denote the average state
of random permutation on ${\cal H}^{\otimes mk}$
to $\sigma_{m,\epsilon_0}^{\otimes k}$ by $\sigma_{m,k,\epsilon_0}$.
We have
\begin{align}
& \Tr (1^{km}_{p'} \sigma_{m,k,\epsilon_0} )
\stackrel{(a)}{=}
\Tr  (1^{km}_{p'} \sigma_{m,\epsilon_0}^{\otimes k}) \notag\\
\stackrel{(b)}{\ge} &\Tr ( (1^{m}_{p_m})^{\otimes k_1}\otimes 
1^{k_2 m}_{p''} \sigma_{m,\epsilon_0}^{\otimes k}) \notag\\
\stackrel{(c)}{\ge} & (1-\epsilon_0)^{k_1} e^{-k_1 D(\rho_m(p_m)\|\sigma_m)} 
\epsilon_0^{k_2} e^{-k_2 m r_0},
\Label{BC5}
\end{align}
where
$(a)$ follows from the permutation-invariance of $1^{km}_{p'}$,
$(b)$ follows from \eqref{BC3}, and 
$(c)$ follows from \eqref{BC4}.

For $p'\in {\cal T}_{km}$ satisfying $\|p'-p\|_1> \epsilon_2$,
we have
\begin{align}
&(\log \Tr (1^{km}_{p'} p^{\otimes km})
-\log \Tr (1^{km}_{p'} \sigma_{m,k,\epsilon_0})) \notag\\
\stackrel{(a)}{\le} &
-\log \Tr (1^{km}_{p'} \sigma_{m,k,\epsilon_0} )
\stackrel{(b)}{\le}  km r_0,
\Label{BC6}
\end{align}
where 
$(a)$ follows from
$\Tr (1^{km}_{p'} p^{\otimes km}) \le 1$
and
$(b)$ follows from the definition of $\sigma_0$.

\noindent{\bf Step 3:}
Evaluation with a sufficiently large integer $k$.

We assume that
$k$ is sufficiently large such that
\begin{align}
\sum_{p'\in {\cal T}_{km}: \|p'-p\|_1> \epsilon_2}
\Tr (1^{km}_{p'} p^{\otimes km})
\le \epsilon_2.\Label{BC7}
\end{align}
Thus,
\begin{align}
&\frac{1}{km} D(p^{\otimes km}\|{\cal F}_{km}) 
\le 
\frac{1}{km} D(p^{\otimes km}\|\sigma_{m,k,\epsilon_0}) \notag\\
\stackrel{(a)}{=} &
\frac{1}{km} \sum_{q\in {\cal T}_{km}}
\Tr (1^{km}_{q} p^{\otimes km})\notag\\
&\cdot
(\log \Tr (1^{km}_{q} p^{\otimes km})
-\log \Tr (1^{km}_{q} \sigma_{m,k,\epsilon_0})) \notag\\
\stackrel{(b)}{\le} &
\frac{1}{km} \sum_{p'\in {\cal T}_{km}: \|p'-p\|_1\le \epsilon_2}
\Tr (1^{km}_{p'} p^{\otimes km}) \notag\\
&\cdot
(\log \Tr (1^{km}_{p'} p^{\otimes km})
-\log \Tr (1^{km}_{p'} \sigma_{m,k,\epsilon_0})) \notag\\
&+r_0 \sum_{p'\in {\cal T}_{km}: \|p'-p\|_1> \epsilon_2}
\Tr (1^{km}_{p'} p^{\otimes km}) \notag\\
\stackrel{(c)}{\le} &
\frac{1}{km} \sum_{p'\in {\cal T}_{km}: \|p'-p\|_1\le \epsilon_2}
\Tr (1^{km}_{p'} p^{\otimes km}) \notag\\
&\cdot
\Big(-\log \Big((1-\epsilon_0)^{k_1} e^{-k_1 D(\rho_m(p_m)\|\sigma_m)} 
\epsilon_0^{k_2} e^{-k_2 m r_0}\Big)\Big)
\notag\\
&+r_0 \epsilon_2\notag\\
\stackrel{(d)}{\le} &
\frac{1}{km} 
\Big( -k_1 \log (1-\epsilon_0) +k_1 D(\rho_m(p_m)\|\sigma_m)
\notag\\
&-k_2 \log \epsilon_0+ k_2 m r_0 \Big)
+r_0 \epsilon_2\notag\\
= &
-\frac{k_1}{km} \log (1-\epsilon_0) +\frac{k_1}{km} D(\rho_m(p_m)\|\sigma_m)
\notag\\
&-\frac{k_2}{km} \log \epsilon_0+ \frac{k_2}{k} r_0 
+r_0 \epsilon_2 \notag\\
= &
  -\frac{1-\epsilon_1-\epsilon_2}{m} \log (1-\epsilon_0) 
  +\frac{1-\epsilon_1-\epsilon_2}{m} D(\rho_m(p_m)\|\sigma_m)
\notag\\
&-\frac{\epsilon_1+\epsilon_2}{m} \log \epsilon_0+ (\epsilon_1+\epsilon_2) r_0 
+r_0 \epsilon_2,
\Label{BC8}
\end{align}
where
$(a)$ follows from the permutation-invariance of 
$p^{\otimes km}$ and $\sigma_{m,k,\epsilon_0}$,
$(b)$ follows from \eqref{BC6}, 
$(c)$ follows from \eqref{BC5} and \eqref{BC7}, and 
$(d)$ follows from the inequality $
\sum_{p'\in {\cal T}_{km}: \|p'-p\|_1\le \epsilon_2}
\Tr (1^{km}_{p'} p^{\otimes km}) \le 1$.

\noindent{\bf Step 4:}
Evaluation with the limit.

Since Proposition \ref{Pr3} guarantees the existence of the limit 
$\lim_{n\to \infty}\frac{1}{n} D(p^{\otimes n}\|{\cal F}_n)$,
we have
\begin{align}
&\lim_{n\to \infty}
\frac{1}{n} D(p^{\otimes n}\|{\cal F}_n) 
\notag\\
\le &
  -\frac{1-\epsilon_1-\epsilon_2}{m} \log (1-\epsilon_0) +\frac{1-\epsilon_1-\epsilon_2}{m} D(\rho_m(p_m)\|\sigma_m)
\notag\\
&-\frac{\epsilon_1+\epsilon_2}{m} \log \epsilon_0+ (\epsilon_1+\epsilon_2) r_0 
+r_0 \epsilon_2.
\Label{BC9}
\end{align}
Taking the limit $\epsilon_2\to 0$, we have
\begin{align}
&\lim_{n\to \infty}
\frac{1}{n} D(p^{\otimes n}\|{\cal F}_n)\notag \\
\le &
  -\frac{1-\epsilon_1}{m} \log (1-\epsilon_0) 
  +\frac{1-\epsilon_1}{m} D(\rho_m(p_m)\|\sigma_m)\notag\\
&-\frac{\epsilon_1}{m} \log \epsilon_0
+ \epsilon_1 r_0 .
\Label{BC10}
\end{align}
Taking the limit 
$\liminf_{m\to \infty}$, we have
\begin{align}
&\lim_{n\to \infty}
\frac{1}{n} D(p^{\otimes n}\|{\cal F}_n) \notag \\
\le& (1-\epsilon_1)
\liminf_{m\to \infty} \frac{1}{m} D(\rho_m(p_m)\|\sigma_m)+ \epsilon_1 r_0 .
\end{align}
Since $\epsilon_1$ is an arbitrary small number, we have 
\begin{align}
\lim_{n\to \infty}
\frac{1}{n} D(p^{\otimes n}\|{\cal F}_n) 
\le 
\liminf_{m\to \infty} \frac{1}{m} D(\rho_m(p_m)\|\sigma_m),
\Label{BC11}
\end{align}
which implies \eqref{BHYC}.

\section{Proof of Theorem \ref{BFG8}}\Label{S9}
According to the paper \cite{4069150},
for $0 < \epsilon \le 1$,
we consider the hypothesis-testing exponent
for two general sequences of states $\vec{\rho}:=(\rho_n)$ and $\vec{\sigma}:=
(\sigma_n)$ as
\begin{align}
&B^\dagger(\epsilon| \vec{\rho}\|\vec{\sigma}) \notag\\
:=&
\sup_{\Pi_n}
\Big\{
\liminf_{n \to \infty} -\frac{1}{n}\log
\Tr (\sigma_n \Pi_n ) \Big|
\liminf_{n \to \infty}
\Tr (\rho_n (I-\Pi_n))< \epsilon\Big\} \\
&B (\epsilon| \vec{\rho}\|\vec{\sigma}) \notag\\
:=&
\sup_{\Pi_n}
\Big\{
\liminf_{n \to \infty} -\frac{1}{n}\log
\Tr (\sigma_n \Pi_n) \Big|
\limsup_{n \to \infty}
\Tr ( \rho_n (I-\Pi_n))\le \epsilon\Big\}.
\end{align}
In this setting, $\sigma_n$ do not need to be normalized
while the states $\rho_n$ need to be normalized.
The paper \cite{4069150} defines
\begin{align}
D^\dagger(\epsilon| \vec{\rho}\|\vec{\sigma})
&:=
\sup\big\{a \big|
\liminf_{n\to \infty}
\Tr (\rho_n \{ \rho_n \le 2^{na} \sigma_n \} )< \epsilon
\big\} \notag\\
&=
\inf\big\{a \big|
\liminf_{n\to \infty}
\Tr (\rho_n \{ \rho_n \le 2^{na} \sigma_n \}) \ge \epsilon
\big\} \Label{eq:d}\\
D(\epsilon| \vec{\rho}\|\vec{\sigma})
&:=
\sup\big\{a \big|
\limsup_{n\to \infty}
\Tr (\rho_n \{ \rho_n \le 2^{na} \sigma_n \} ) \le \epsilon
\big\} \Label{eq:d2}.
\end{align}

We have the following lemma.
\begin{lemma}[\protect{\cite[Theorem 1]{4069150}}]\Label{LL8}
For $0 \le \epsilon < 1$, we have
\begin{align}
D^\dagger(\epsilon| \vec{\rho}\|\vec{\sigma})=&
B^\dagger(\epsilon| \vec{\rho}\|\vec{\sigma}) \\
D(\epsilon| \vec{\rho}\|\vec{\sigma})=&
B(\epsilon| \vec{\rho}\|\vec{\sigma}).
\end{align}
\end{lemma}

We define $\rho_n:=\argmin_{\rho \in {\cal F}_n}
D(\rho\|\sigma_n)$.
Then, we have
\begin{align}
\frac{-1}{n}\log 
\beta_\epsilon(\rho_n\|\sigma_n)
&\ge 
\frac{-1}{n}\log 
\beta_\epsilon({\cal F}_n\|\sigma_n).
\Label{NI4N}
\end{align}
The combination of \eqref{NH4} and \eqref{NI4N}
implies
\begin{align}
&\liminf_{n\to \infty}\frac{-1}{n}\log 
\beta_\epsilon(\rho_n\|\sigma_n)\notag\\
\ge & 
\liminf_{n\to \infty}\frac{-1}{n}\log \beta_\epsilon({\cal F}_n\|\sigma_n)
\ge 
\limsup_{n\to \infty}\frac{1}{n}D({\cal F}_n\|\sigma_n)
\Label{NH4C}.
\end{align}

Also, using \eqref{NH4}, we have
\begin{align}
\liminf_{n\to \infty}\frac{-1}{n}\log 
\beta_\epsilon(\rho_n\|\sigma_n)
\ge 
\limsup_{n\to \infty}\frac{1}{n}D(\rho_n\|\sigma_n).
\Label{NH4N}
\end{align}
For any $\epsilon>0$, we choose 
$N(\epsilon)$, we have
\begin{align}
\frac{-1}{n}\log \beta_\epsilon(\rho_n\|\sigma_n)
+\epsilon
\ge 
\limsup_{n\to \infty}\frac{1}{n}D(\rho_n\|\sigma_n)
\Label{NH4V}
\end{align}
for $n \ge N(\epsilon)$.
Since 
\begin{align}
\liminf_{n\to \infty}\frac{-1}{n}\log \beta_\epsilon(\rho_n\|\sigma_n)
<
\limsup_{n\to \infty}\frac{1}{n}D(\rho_n\|\sigma_n),
\Label{NY4N}
\end{align}
we have 
$\lim_{\epsilon\to 0} N(\epsilon)=\infty$.
We set $\epsilon (n)$ as
\begin{align}
\epsilon (n):= \sup
\{ 2 \epsilon|  n \ge N(\epsilon) \}.
\end{align}
Hence, $\lim_{n\to \infty}\epsilon (n)=0$.
Also, we have \eqref{NH4V} with $\epsilon= \epsilon (n)$ 
because $n \ge N(\epsilon (n))$.

Then, we have
\begin{align}
D(0| \vec{\rho}\|\vec{\sigma})
=B(0| \vec{\rho}\|\vec{\sigma})
\ge \limsup_{n\to \infty}\frac{1}{n}
D(\rho_n\|\sigma_n).\Label{LNA}
\end{align}

\begin{lemma}\Label{BCV}
When $\rho_n$ is 
commutative with $\sigma_n$ and
\begin{align}
D(0| \vec{\rho}\|\vec{\sigma})
\ge \limsup_{n\to \infty}\frac{1}{n}
D(\rho_n\|\sigma_n),
\end{align}
we have
\begin{align}
D(\epsilon| \vec{\rho}\|\vec{\sigma})
=D^\dagger(\epsilon| \vec{\rho}\|\vec{\sigma})
= \lim_{n\to \infty}\frac{1}{n}
D(\rho_n\|\sigma_n)
\end{align}
for $\epsilon \in [0,1)$.
\end{lemma}

Combining \eqref{LNA} and Lemma \ref{BCV},
we have
\begin{align}
\lim_{n\to \infty}\frac{-1}{n}\log 
\beta_\epsilon(\rho_n\|\sigma_n)
=
\lim_{n\to \infty}\frac{1}{n}D(\rho_n\|\sigma_n)
\Label{YH4N}.
\end{align}
The combination of \eqref{NH4C} and \eqref{YH4N}
yields \eqref{NH5}, which is the main statement of
Theorem \ref{BFG8}.

\begin{proofof}{Lemma \ref{BCV}}
We show this statement by contradiction.
Assume that 
$a_\epsilon:=
D^\dagger(\epsilon| \vec{\rho}\|\vec{\sigma})>
\limsup_{n\to \infty}\frac{1}{n}
D(\rho_n\|\sigma_n)$.
We choose 
$b_\epsilon<a_\epsilon$
and $b< \limsup_{n\to \infty}\frac{1}{n}D(\rho_n\|\sigma_n)$.
We define the projection 
\begin{align}
Q_n(a)
&:=\{ \rho_{n} \le 2^{{n} a} \sigma_{n} \} \notag\\
&=\Big\{ 
\frac{1}{n}(\log \rho_{n}-\log \sigma_{n}) \le a
\Big\} .
\end{align}
That is,
\begin{align}
I-Q_n(a)
=\Big\{ 
\frac{1}{n}(\log \rho_{n}-\log \sigma_{n}) > a\Big\} .
\end{align}

We choose a subsequence $n_k$ such that
$\lim_{k\to \infty}\Tr (\rho_{n_k} Q_{n_k}(b_{n_k}))
= \epsilon$.
We have
$\lim_{n\to \infty}\Tr (\rho_{n} Q_n(b))
= 0$.

We have
\begin{align}
&\Tr ( \rho_{n_k} Q_{n_k}(b)
(\log \rho_{n_k}-\log \sigma_{n_k}) )\notag\\
=&
\Tr (\rho_{n_k} Q_{n_k}(b))
\Tr (\frac{\rho_{n_k} Q_{n_k}(b)}{\Tr (\rho_{n_k} Q_{n_k}(b))})\notag\\
&\cdot\Big(\log \frac{\rho_{n_k} Q_{n_k}(b)}{\Tr (\rho_{n_k} Q_{n_k}(b))} 
-\log \frac{\sigma_{n_k}Q_{n_k}(b)}{\Tr (\sigma_{n_k}Q_{n_k}(b))}\Big) \notag\\
&+
\Tr (\rho_{n_k} Q_{n_k}(b)
\log \frac{\Tr (\sigma_{n_k}Q_{n_k}(b))}{\Tr (\rho_{n_k} Q_{n_k}(b))} ) \notag\\
\ge &
\Tr (\rho_{n_k} Q_{n_k}(b)
\log \Tr \sigma_{n_k}Q_{n_k}(b)).
\Label{BM99}
\end{align}
Then, we have
\begin{align}
&\frac{1}{n_k}D(\rho_{n_k}\| \sigma_{n_k})
=\frac{1}{n_k}
\Tr (\rho_{n_k}
(\log \rho_{n_k}-\log \sigma_{n_k})) \notag\\
=&
\frac{1}{n_k}
\Tr (\rho_{n_k} (I-Q_{n_k}(b_{\epsilon}))
(\log \rho_{n_k}-\log \sigma_{n_k})) \notag\\
&+\frac{1}{n_k}
\Tr (\rho_{n_k} (Q_{n_k}(b_{\epsilon})-Q_{n_k}(b))
(\log \rho_{n_k}-\log \sigma_{n_k}) )\notag\\
&+\frac{1}{n_k}
\Tr ( \rho_{n_k} Q_{n_k}(b)
(\log \rho_{n_k}-\log \sigma_{n_k}) )
\notag \\
\ge &
b_{\epsilon} \Tr (\rho_{n_k} (I-Q_{n_k}(b_{\epsilon})))
+b \Tr (\rho_{n_k} (Q_{n_k}(b_{\epsilon})-Q_{n_k}(b)))\notag\\
&+\frac{1}{n_k}
\Tr (\rho_{n_k} Q_{n_k}(b)
\log \Tr \sigma_{n_k}Q_{n_k}(b)).
\end{align}
Taking the limit $k\to \infty$, we have
\begin{align}
\liminf_{k\to \infty}\frac{1}{n_k}D(\rho_{n_k}\| \sigma_{n_k})
= b_{\epsilon}(1-\epsilon)+b \epsilon.
\end{align}
Taking the limit
$b_\epsilon\to D^\dagger(\epsilon| \vec{\rho}\|\vec{\sigma})$
and $b\to \limsup_{n\to \infty}\frac{1}{n}D(\rho_n\|\sigma_n)$,
we have
\begin{align}
&\limsup_{n\to \infty}\frac{1}{n}D(\rho_n\|\sigma_n)
\ge \liminf_{k\to \infty}\frac{1}{n_k}D(\rho_{n_k}\| \sigma_{n_k}) \notag\\
=& D^\dagger(\epsilon| \vec{\rho}\|\vec{\sigma})(1-\epsilon)+
\limsup_{n\to \infty}\frac{1}{n}D(\rho_n\|\sigma_n)
 \epsilon \notag\\
\ge & D^\dagger(\epsilon| \vec{\rho}\|\vec{\sigma})(1-\epsilon)+
\liminf_{k\to \infty}\frac{1}{n_k}D(\rho_{n_k}\| \sigma_{n_k})
 \epsilon,
\end{align}
which implies
\begin{align}
D(0| \vec{\rho}\|\vec{\sigma})
&\ge \limsup_{n\to \infty}\frac{1}{n}D(\rho_n\|\sigma_n)
\ge D^\dagger(\epsilon| \vec{\rho}\|\vec{\sigma}) \\
D(0| \vec{\rho}\|\vec{\sigma})
&\ge 
\liminf_{k\to \infty}\frac{1}{n_k}D(\rho_{n_k}\| \sigma_{n_k})
\ge D^\dagger(\epsilon| \vec{\rho}\|\vec{\sigma}) .
\end{align}
Thus, we have
\begin{align}
&D(0| \vec{\rho}\|\vec{\sigma})
= D^\dagger(\epsilon| \vec{\rho}\|\vec{\sigma}) \notag\\
=& \limsup_{n\to \infty}\frac{1}{n}D(\rho_n\|\sigma_n)
=\liminf_{k\to \infty}\frac{1}{n_k}D(\rho_{n_k}\| \sigma_{n_k}).
\end{align}
\end{proofof}

\section{Counter example}\Label{S3}
The paper \cite{LBR} presented an example to clarify 
that the conditions (A1)-(A5) are not sufficient to realize the relation \eqref{BNAT}.
However, in their example, 
$D({\cal F}_1\|\sigma)$ is infinity.
We present a slightly different example such that 
$D({\cal F}_1\|\sigma)$ is a finite value, the conditions (A1)--(A5) hold, and the relation \eqref{BNAT} does not fold.

We consider the qubit case when ${\cal H}$ is spanned by 
$|0\rangle $ and $|1\rangle$.
We consider the uniform distribution on 
${\cal H}^{\otimes n}$ by $\rho_{n,\mix}$.
We denote the set of types with length $n$ by ${\cal T}_n$.
Given $p \in {\cal T}_n$, 
we denote the set of elements 
whose empirical distribution is $p$ by ${\cal T}_{n,p}$.
we denote the uniform distribution over ${\cal T}_{n,p}$
by $\rho_n(p)$.
We also define the identity $1^n_{p}$ over ${\cal T}_{n,p}$.

In the binary case, for $p \in [0,1]$, we simplify
$\rho_n((p,1-p))$ to $\rho_n(p)$.
We define
$\rho(p):= p|0\rangle \langle 0|+(1-p)|1\rangle \langle 1|$.

Given two rational number $p_1<p_2 \in [0,1/2]$, 
we choose positive integers $m_1,m_2,m_3$ as
$p_1=\frac{m_1}{m_3}$ and $p_2=\frac{m_2}{m_3}$.
We also choose $p\in (0,1)$ and a positive integer $m$. 
For $n=km$, 
we define the state
\begin{align}
\rho_n:= p \rho_n(p_1)+(1-p) \rho_n(p_2).
\end{align}
We define the set ${\cal F}_{1,n}$ as
\begin{align}
{\cal F}_{1,n}:=\{ \Tr_{km-n} \rho_{km} \}_{k=1}^{\infty},
\end{align}
where $\Tr_{km-n}$ is the partial trace on the initial $km-n$ subsystems.
We define the set ${\cal F}_{2,n}$ as
\begin{align}
{\cal F}_{2,n}:=
\Big\{  \rho_1 \otimes \cdots \otimes \rho_k \Big|
\rho_j \in {\cal F}_{1,n_j}, \sum_{j=1}^k n_j=n\Big\}.
\end{align}
We denote the convex hull of ${\cal F}_{2,n}$
by ${\cal F}_{3,n}$.
${\cal F}_{3,1}$ is composed of only one element
$\rho(p p_1+(1-p)p_2)$

\begin{lemma}
The sets $\{{\cal F}_{3,n}\}_n$ satisfy the conditions (A1)--(A5).
\end{lemma}

Since we have
\begin{align}
D(\rho(p_2)^{\otimes km}\| \rho_{km})=O(\log n),
\end{align}
the following lemma holds.

\begin{lemma}
The sequence of sets $\{{\cal F}_{3,n}\}_n$ does not satisfy the condition (A6).
\end{lemma}

In addition, we obtain the following lemma.
\begin{lemma}
The sequence of sets $\{{\cal F}_{3,n}\}_n$ does not satisfy the condition (A8).
\end{lemma}

\begin{proof}
We choose two positive integers $k_1,k_2$.
Also, for any $\epsilon >0$ and $a>0$, we have
$a 1^{km}_{p_2+\epsilon} \in {\cal F}_{3,km}$.
For any $p' \in [p_1,p_2]$, we have
$1^{km}_{p'} \in {\cal F}_{3,km}$.
When $ \frac{k_1}{k_1+k_2}p' + \frac{k_2}{k_1+k_2}(p_2+\epsilon)
\in [p_1,p_2]$
We have
\begin{align}
\Tr ((1^{k_1m}_{p'}\otimes 1^{k_2m}_{p_1+\epsilon}) \rho_{(k_1+k_2)m})
>0.
\end{align}
That is,
\begin{align}
\lim_{a \to +\infty}\Tr 
((1^{k_1m}_{p'}\otimes a 1^{k_2m}_{p_1+\epsilon}) 
\rho_{(k_1+k_2)m})
=+\infty.
\end{align}
For a sufficiently large $a>0$, we have
$1^{k_1m}_{p'}\otimes a 1^{k_2m}_{p_1+\epsilon} 
\notin {\cal F}_{3,(k_1+k_2)m}$.
\end{proof}

In addition, we obtain the following lemma.
\begin{lemma}
The sequence of sets $\{{\cal F}_{3,n}\}_n$ does not satisfy the condition \eqref{BNA}
\end{lemma}

\begin{proof}
Assume that $\rho_{\mix}$ is the completely mixed state.
When $\epsilon <1-p$.
\begin{align}
& \lim_{k\to \infty}-\frac{1}{km}\log \beta_\epsilon({\cal F}_{km}\|\rho_{\mix}^{\otimes km})
\notag\\
\le &
\lim_{k\to \infty}-\frac{1}{km}\log \beta_\epsilon(\rho_{km}\|\rho_{\mix}^{\otimes km})
=
\log 2- h(p_2).
\end{align}

On the other hand, we have
\begin{align}
\lim_{n\to \infty}\frac{1}{n}
\min_{\rho \in {\cal F}_n} S(\rho)
=h(p p_1+(1-p)p_2).
\end{align}
Hence,
\begin{align}
\lim_{n\to \infty}\frac{1}{n}D({\cal F}_n\|\rho_{\mix}^{\otimes n})
=\log 2 - h(p p_1+(1-p)p_2).
\end{align}
Since $h(p p_1+(1-p)p_2) < h(p_2)$,
we have
\begin{align}
\lim_{n\to \infty}-\frac{1}{n}\log \beta_\epsilon({\cal F}_n\|\rho_{\mix}^{\otimes n})
<\lim_{n\to \infty}\frac{1}{n}D({\cal F}_n\|\rho_{\mix}^{\otimes n}).
\end{align}
for $\epsilon <1-p$.
\end{proof}

\bibliography{detection}

\end{document}